\documentclass[10pt]{iopart}

\usepackage{iopams}
\usepackage{amssymb}
\usepackage{amsfonts}
\usepackage{amstext}
\usepackage{graphicx}
\usepackage{setstack}
\usepackage{amscd}
\usepackage{mathrsfs}
\usepackage{stmaryrd}
\usepackage{slashed}
\usepackage{color}
\usepackage{upgreek}
\usepackage{tikz-cd}
\usepackage{hhline}

\newcommand{\Sp}{\mathrm{Sp}}

\renewcommand{\Re}{\mathrm{Re}}
\renewcommand{\Im}{\mathrm{Im}}
\newcommand{\Ted}{\underset{\rightarrow}{\mathbb{T}e}}
\newcommand{\Teg}{\underset{\leftarrow}{\mathbb{T}e}}
\newcommand{\Ped}{\underset{\rightarrow}{\mathbb{P}e}}

\newcommand{\llangle}{\langle \hspace{-0.2em} \langle}
\newcommand{\rrangle}{\rangle \hspace{-0.2em} \rangle}

\newcommand{\dist}{\mathrm{dist}}
\newcommand{\Ad}{\mathrm{Ad}}
\newcommand{\id}{\mathrm{id}}
\newcommand{\xrightarrow}[1]{\overset{#1}{\longrightarrow}}

\newcommand{\Env}{\mathrm{Env}}
\newcommand{\Der}{\mathrm{Der}}
\newcommand{\Aut}{\mathrm{Aut}}
\newcommand{\InnAut}{\mathrm{InnAut}}
\newcommand{\Hom}{\mathrm{Hom}}
\newcommand{\Lin}{\mathrm{Lin}}
\newcommand{\dnc}{\mathbf{d}}
\newcommand{\Diff}{\mathrm{Diff}}
\newcommand{\Dis}{\mathbf{D}}
\newcommand{\Lie}{\mathrm{Lie}}
\newcommand{\Obj}{\mathrm{Obj}}
\newcommand{\Morph}{\mathrm{Morph}}

\newtheorem{defi}{Definition}

\newtheorem{prop}{Property}

\newtheorem{propo}{Proposition}
\newtheorem{lemma}{Lemma}

\newenvironment{proof}{\noindent \textit{Proof:}\small }{\normalsize \hfill $\Box$ \\}
\newenvironment{example}[2]{\noindent \textbf{Example #1: #2} \hrulefill \\ \small}{\hrule \vspace{0.4cm}}

\begin{document}

\title{Metrics and geodesics on fuzzy spaces}

\author{David Viennot}
\address{Institut UTINAM (CNRS UMR 6213, Universit\'e de Franche-Comt\'e, Observatoire de Besan\c con), 41bis Avenue de l'Observatoire, BP1615, 25010 Besan\c con cedex, France.}

\begin{abstract}
We study the fuzzy spaces (as special examples of noncommutative manifolds) with their quasicoherent states in order to find their pertinent metrics. We show that they are naturally endowed with two natural ``quantum metrics'' which are associated with quantum fluctuations of ``paths''. The first one provides the length the mean path whereas the second one provides the average length of the fluctuated paths. Onto the classical manifold associated with the quasicoherent state (manifold of the mean values of the coordinate observables in the state minimising their quantum uncertainties) these two metrics provides two minimising geodesic equations. Moreover, fuzzy spaces being not torsion free, we have also two different autoparallel geodesic equations associated with two different adiabatic regimes in the move of a probe onto the fuzzy space. We apply these mathematical results to quantum gravity in BFSS matrix models, and to the quantum information theory of a controlled qubit submitted to noises of a large quantum environment.
\end{abstract}


\section{Introduction}
Fuzzy spaces \cite{Barrett} are special cases of Connes' noncommutative geometry \cite{Connes}; fuzzy spaces are inspired from the fundamental example of the fuzzy sphere \cite{Madore}. They are mathematical framework for matrix models of quantum gravity (nonperturbative regime of type IIB string theory): BFSS (Banks-Fischler-Shenker-Susskind) \cite{BFSS} and IKKT (Ishibashi-Kawai-Kitazawa-Tsuchiya) \cite{IKKT} matrix models. Simple fuzzy spaces arise by considering a fermionic string linking a D2-brane and a D0-brane \cite{Steinacker, Berenstein} (the spacetime dimension could be reduced to 3+1 by considering a truncation by taking an orbifold \cite{Berenstein}). The D2-brane is then assimilated to a noncommutative manifold. Quantum gravity is not the only domain of application of fuzzy spaces. Indeed adiabatic control of qubits entangled with an environment exhibits higher gauge structures similar to the ones of string theory \cite{Viennot1}, and the Hamiltonian in the interaction picture is similar to the geometric operator of a fuzzy space.\\
In the spirit of the Connes' noncommutative geometry, we can simply seen a fuzzy space as the quantisation of an extended ``solid body''. Usual quantum mechanics (obtained by application of the canonical quantisation rules -- first quantisation --) is the quantisation of classical mechanics of point systems. Quantum field theory is the quantisation of classical field theory (second quantisation, by application of the canonical quantisation rules onto conjugated fields -- the four-potential vector field and the electric field for electrodynamics for example --); so the quantisation of the ``scalar, vector or spinor degrees of freedom'' by sustaining classical the base points of the fields. We can imagine the fuzzy spaces as resulting from a noncanonical third quantisation (``fuzzyfication'') of rigid or deformable bodies. Table \ref{fuzzyfication} summarises briefly this idea for the fuzzyfication of a surface.
\begin{table}
\begin{center}
  \begin{tabular}{rcl}
    \textit{Classical point mechanics} & $\xrightarrow{\text{first quantisation}}$ & \textit{Quantum mechanics} \\
    $\vec x \in \mathbb R^3$ & & $\psi \in L^2(\mathbb R^3,d\vec x) = \mathcal H$ \\
    $(x^i,p^i) \in \mathbb R^6$ & & $\hat x^i,\hat p^i \in \mathcal L(\mathcal H)$ \\
    & & $[\hat x^i,\hat p^j] = \imath \delta^{ij} \id$ \\
    \hline
    \textit{Classical field theory} & $\xrightarrow{\text{second quantisation}}$ & \textit{Quantum field theory} \\
    $\{\vec x \mapsto \psi(\vec x)\} \in \underline{\mathbb C}_{\mathbb R^3}$ & & $\{\vec x \mapsto \Psi(\vec x),\Psi^+(\vec x)\} \in \underline{\mathcal L(\mathcal F_\pm)}_{\mathbb R^3}$ \\
    & & $[\Psi(\vec x),\Psi^+(\vec x^\prime)]_\pm = \delta(\vec x-\vec x^\prime) \id $\\
    \hline
    \textit{Classical solid mechanics} & $\xrightarrow{\text{fuzzyfication}}$ & \textit{Fuzzy space theory} \\
    $M = \{\vec x \in \mathbb R^3 \text{ with } f(\vec x)=0\}$ & & $\mathfrak X = \Lin_{\mathbb R}(X^1,X^2,X^3)$ \\
    & & with $\vec X \in \mathcal L(\mathcal H)^3$, $f(\vec X) = 0$ \\
    & & and $[X^i,X^j] \not= 0$ \\
    $\frac{\partial}{\partial x^i} \in TM$ & & $-\imath [X^i,\bullet] \in \Der(\mathfrak X)$ \\
    $dx^i \in \Omega^1 M$ & & $\Der(\mathfrak X) \ni L \mapsto \dnc X^i(L) = L(X^i)$
  \end{tabular}
  \caption{\label{fuzzyfication} Comparison between first and second quantisation and fuzzyfication. $\mathcal H$ is a Hilbert space, $\mathcal F_\pm$ a fermionic or bosonic Fock space. $\Psi^+(x) = \sum_i \overline{\phi_i(x)} a_i^+$ and $\Psi(x) = \sum_i \phi_i(x) a_i$ are the field operators of creation and annihilation of a particle at point $x$ ($(\phi_i)_i$ being an orthonormal basis of the Hilbert space of a single particle and $a_i^+/a_i$ being the creation/annihilation operators of a particle on the mode $i$). $(X^i)_i$ is a set of coordinate observables of the noncommutative manifold, the superoperator $-\imath [X_i,\bullet]$ playing the role of the noncommutative tangent vector field in the $i$-direction and $\dnc X^i$ (with $\dnc$ the Koszul differential) playing the role of the noncommutative cotangent vector field (differential 1-form) in the $i$-direction.}
\end{center}
\end{table}
\\
An important notion concerning the fuzzy space concerns their quasicoherent states \cite{Schneiderbauer, Steinacker2}. These ones are named by analogy with the Perelomov coherent states of a Lie algebra \cite{Perelomov}, and are the states such that the Heisenberg uncertainties concerning the coordinate observables are minimised. Moreover these ones are the ground eigenstates of the geometry operator $\slashed D_x$ of the fuzzy space. In the BFSS matrix model \cite{Berenstein}, $\slashed D_x$ is the Dirac operator of the fermionic string, and the quasicoherent states are the ones for which the displacement energy (the ``tension energy'' of the string) is zero. In quantum information theory, the interaction Hamiltonian of the qubit has a structure similar to $\slashed D_x$, and the quasicoherent state is (in the language of ref. \cite{Viennot1}) a *-eigenvector associated with a noncommutative eigenvalue, corresponding to a qubit control with zero energy uncertainty. In a recent work \cite{Viennot2}, it is proposed that the quasicoherent picture permits to define the emergent gravity at the Planck scale in the BFSS model. Indeed, we can see the BFSS model as the quantisation of a flat spacetime in which the noncommutativity induces a quantum non-zero (Weitzenböck) torsion (see ref. \cite{Aldrovandi} for an introduction to Weitzenböck torsions). At the semi-classical thermodynamical limit (number of strings tending to infinity with constant density), gravity (spacetime curvature) emerges at the macroscopic scale from the noncommutativity at the microscopic one (more precisely, the noncommutativity relations define at the thermodynamical limit a Poisson structure which defines a spacetime effective metric). Quasicoherent states are associated with a classical ``eigenmanifold'' (as a quantum object, a fuzzy space has for ``eigenvalues'' -- measurement outputs -- a classical quantity, here a classical manifold). In ref. \cite{Viennot2} it is argued that this eigenmanifold is the emergent curved spacetime at the Planck scale, in the meaning where it is the classical manifold closest to the quantum geometry of the fuzzy space (since it is associated with the states minimising the Heisenberg uncertainties and with adiabatic transport which is the quantum dynamical regime closest to classical dynamics). This emerging spacetime is endowed with a natural metric and with a Lorentz connection (inducing curvature and torsion).\\
In the present paper, we want to analyse more precisely the proposition of ref. \cite{Viennot2} with more general and mathematical point of view of the possibility to endow a fuzzy space and its eigenmanifold with metrics and geodesic equations. Different works explore the application of the Connes' metric onto fuzzy spaces \cite{Rieffel,Dandrea,Dandrea2,Dandrea3,Dandrea4}, but in this paper we present another way by considering metrics associated with the quasicoherent picture. We will show that a fuzzy space can be endowed with two different metrics, one being a quantised metric as an infinitesimal square length (as in the Connes theory) and the other one being a quantised metric as a field of inner product of tangent vectors. The presence of two metrics is in accordance with the ``quantum fluctuations'' onto a fuzzy manifold. ``Paths'' onto a fuzzy manifold are submitted to these quantum fluctuations, and so the square root of the averaging of the first quantum metric in quasicoherent states (which is the metrics onto the eigenmanifold of ref. \cite{Viennot2}) is the length of the mean infinitesimal path, whereas the averaging of the square root of the second quantum metric is the mean length of the fluctuating infinitesimal quantum paths (the mean length of the paths being larger than the length of the mean path, as for a classical Brownian motion). This effect of the quantum uncertainty (of the quantum indeterminism) explains why the two notions of metric (infinitesimal square length and field of inner product), which are totally equivalent in classical geometry, are two distinct objects in quantum geometry. Moreover, in addition to the interpretation to quantum gravity, we want also here explore the interpretation of the quantum metrics for quantum information theory (as we will seen, they are associated with residual energy uncertainties associated with measurement misalignments, i.e. small errors in the application of adiabatic quantum control of the qubit onto the manifold of zero energy uncertainty).\\
This paper is organised as follows. Section 2 sets definitions concerning fuzzy spaces and noncommutative geometry. The role of this section is essentially to fix some notations and some definitions which can be slightly change from an author to another one. Section 3 is a review about the eigen geometry of a fuzzy space (theory of quasicoherent states and eigenmanifolds). Section 4 introduces the two quantum metrics, their relations with the geometry of the eigenmanifold and their interpretations. Section 5 studies the adiabatic transport of a probe onto the eigenmanifold of a fuzzy space. The torsion is intimately related to this question, and we introduce then the associated auto-parallel geodesics and their interpretations. Section 6 explains how generalise the present discussion to a time-dependent fuzzy space. Finally in a concluding section, we discuss the generalisation to higher dimensional fuzzy spaces (the dimension is set to be 3 in this paper for the sake of simplicity), the analogy between quantum gravity and quantum information theory provided by the common model of fuzzy space; and we draw some futur research directions. Four appendices finish this paper. The first one introduces the perturbation theory of fuzzy spaces, which is needed to treat concrete examples. The second one presents a generalisation of section 4 which focuses only on fuzzy spaces homogeneous concerning their entanglement properties. In the second appendix we relax this property. The third appendix presents some computations associated with the Lorentz connection of a fuzzy space, which are not necessary to the understanding of the main discussion of this paper. And the last one presents the relation between the fuzzy geometry and the category theory (the presence of two quantum metrics being related to the need of defining a metric for the objects and another one for the arrows - morphisms - onto a categorical manifold). The fuzzy space geometry is often illustrated in the literature by the highly symmetric examples of the fuzzy sphere, the fuzzy plane, and the fuzzy complex projective spaces. Throughout this paper, we illustrate the present results by several examples, and especially the case of fuzzy surface plots (fuzzyfication of classical surfaces defined by a cartesian equation $z=f(x,y)$). The examples are splitted in the whole of the paper in order to illustrate a notion immediately after its introduction.\\

\textit{Some notations are used throughout this paper. We adopt the Einstein notations concerning the repetition of an index at lower and upper positions which is equivalent to a summation. For a tensor $T$, we denote the symmetrisation and the antisymmetrisation of two indices by $T_{...(a...b)...} \equiv T_{...a...b...}+T_{...b...a...}$ and $T_{...[a...b]...} \equiv T_{...a...b...}-T_{...b...a...}$. In a tensor, latin indices run from 1, whereas greek indices run from 0. For a linear set of operators (or matrices) $\mathfrak X$, $\Env(\mathfrak X)$ denotes its enveloping $C^*$-algebra and $\Der(\mathfrak X)$ its linear space of inner derivatives (Lie derivatives $-\imath [X,\bullet]$). For a set of vectors $V$ in a vector space $E$ over $\mathbb K$, $\Lin_{\mathbb K}(V)$ denotes the vector subspace of $E$ generated by $V$. For a manifold $M$, $T_xM$ denotes the set of tangent vectors at $x$, $\Omega^n M$ the set of differential $n$-forms and $\Diff M$ its set of diffeomorphisms. $\simeq$ between two manifolds stands for ``diffeomorphic to''. $\underline{E}_M$ denotes the $E$-valued differentiable functions on $M$. For a category $\mathscr C$, $\Obj \mathscr C$ denotes its set of objects, $\Morph \mathscr C$ denotes its set of arrows, $s$, $t$ and $\id$ denote their source, target and identity maps; the composition of arrows being denoted by $\circ$.}\\

\textit{For the applications, we use the unit system such that $\hbar = c = G = 1$ ($\ell_P = m_P = t_P = 1$ Planck units) or $\hbar = \frac{e^2}{4\pi \epsilon_0} = 1$ (atomic units).}

\section{Fuzzy spaces}
In this section we present the general concept of fuzzy spaces \cite{Barrett} and some related notions about noncommutative geometry \cite{Connes}. We present no new result, but this section permits to fix some notations and some definitions for the sequel of this paper.

\begin{defi}[Fuzzy space]
  A (3D) fuzzy space is a noncommutative manifold defined by a spectral triple $\mathfrak M = (\Env(\mathfrak X), \mathbb C^2 \otimes \mathcal H, \slashed D_x)$ where
  \begin{itemize}
    \item $\mathcal H$ is a separable Hilbert space.
    \item $\mathfrak X = \Lin_{\mathbb R}(X^1,X^2,X^3,\id)$ is a space generated by three self-adjoint linear operators $X^i$ of the Hilbert space $\mathcal H$ (and the identity operator on $\mathcal H$) and $\Env(\mathfrak X)$ is the $C^*$-enveloping algebra of $\mathfrak X$.
    \item $\slashed D_x = \sigma_i \otimes (X^i -x^i)$ is the Dirac operator of the noncommutative manifold where $(\sigma_i)$ are the Pauli matrices and $x \in \mathbb R^3$ is a classical parameter.
  \end{itemize}
\end{defi}
Note that the case where one $X^i$ be the zero operator is not excluded. It can be useful to introduce the non-self-adjoint coordinate observable $Z=X^1+\imath X^2$ associated with the complex parameter $z=x^1+\imath x^2$ to write $\slashed D_x = \left(\begin{array}{cc} X^3-x^3 & Z^\dagger-\bar z \\ Z-z & -X^3+x^3 \end{array} \right)$.\\

As usual in noncommutative geometry, the non-abelian $C^*$-algebra $\Env(\mathfrak X)$ plays the role of the space of functions of $\mathfrak M$: $\mathcal C^\infty_{n.c.}(\mathfrak M) \sim \Env(\mathfrak X)$, and $\mathbb C^2 \otimes \mathcal H$ plays the role of a spinor field space on $\mathfrak M$. Let $\mathcal Z(\Env(\mathfrak X))$ be the center of $\Env(\mathfrak X)$. The space of derivatives $\Der(\mathfrak X) = \Lin_{\mathcal Z(\Env(\mathfrak X))}(L_{X^1},L_{X^2},L_{X^3})$ (with $L_X(Y) = -\imath[X,Y]$, $\forall X \in \mathfrak X, Y \in \Env(\mathfrak X)$) plays the role of the space of tangent vector fields on $\mathfrak M$: $T_{n.c.} \mathfrak M \sim \Der(\mathfrak X)$. The set of noncommutative differential $n$-forms $\Omega^n_\Der(\mathfrak X)$ is the set of $\mathcal Z(\Env(\mathfrak X))$-multilinear antisymmetric maps from $\Der (\mathfrak X)^n$ to $\Env(\mathfrak X)$.\\
The differential $\dnc: \Omega^n_\Der (\mathfrak X) \to \Omega^{n+1}_\Der (\mathfrak X)$ is defined by the Koszul formula:
\begin{eqnarray}
  & & \dnc \xi(L_1,...,L_{n+1})  =  \sum_{i=1}^{n+1} (-1)^{i+1} L_i\xi(L_1,...,\check L_i,...,L_{n+1}) \nonumber \\
  & \quad & + \sum_{1\leq i<j\leq n+1} (-1)^{i+j} \xi([L_i,L_j],L_1,...,\check L_i,...,\check L_j,...,L_{n+1})
\end{eqnarray}
where $\check L_i$ means ``deprive of $L_i$''.\\
In contrast with a commutative manifold, the duality relation between the derivative with respect to $X^i$ and the differential of $X^j$ fails:
\begin{equation}
  \langle \dnc X^i , L_{X^j} \rangle_{\mathfrak X} = -\imath [X^j,X^i] \not= \delta^{ij}
\end{equation}
where $\langle \bullet, \bullet \rangle_{\mathfrak X}$ stands for the duality bracket between $\Omega^1_{\Der}(\mathfrak X)$ and $\Der(\mathfrak X)$. $\Theta^{ij} = -\imath [X^i,X^j]$ defines a Poisson structure on $\mathfrak M$ with the following Poisson bracket $\Theta:\mathfrak X^2 \to \Env(\mathfrak X)$: $\Theta(f_i X^i,g_j X^j) = -\imath [f_i X^i,g_j X^j] = \Theta^{ij} f_i g_j$ (with $f_i,g_j \in \mathcal Z(\Env(\mathfrak X))$ or often in $\mathbb R$). Let $\theta_i \in \Omega^1_{\Der}(\mathfrak X)$ be the 1-forms defined by duality with $L_{X^j}$:
\begin{equation}
  \langle \theta_i,L_{X^j} \rangle_{\mathfrak X} = \delta^j_i
\end{equation}
By construction we have $\dnc X^i = \Theta^{ji} \theta_j$.\\

\begin{defi}[States of a fuzzy space]
  We call normal states of a fuzzy space $\mathfrak M$ the normal states \cite{Bratteli} of the $C^*$-algebra $\Env(\mathfrak X)$:
  $$ \mathcal E(\mathfrak M) = \{\omega: \Env(\mathfrak X) \to \mathbb R, \exists \rho \in \mathcal L(\mathcal H), \rho^\dagger=\rho, \rho\geq 0, \tr \rho=1, \text{ with } \omega(X) = \tr(\rho X)\} $$
  We call pure states of $\mathfrak M$ the pure states \cite{Bratteli} of the $C^*$-algebra $\mathcal L(\mathbb C^2) \otimes \mathcal B(\mathcal H)$:
  $$ \mathcal P(\mathfrak M) = \{\omega: \Env(\mathfrak X) \to \mathbb R, \exists P \in \mathcal L(\mathbb C^2 \otimes \mathcal H), P^\dagger=P, \tr P=1, P^2=P \text{ with } \omega(X) = \tr(P X)\} $$
\end{defi}
The density matrices $\rho$ play the role of the density functions (or distributions) on $\mathfrak M$ and the states play the role of the integration on $\mathfrak M$:
\begin{equation}
  \int_{\mathfrak M}^{n.c.} Y \rho \dnc X^1 \dnc X^2\dnc X^3 \sim \tr(\rho Y) \qquad \forall Y \in \Env(\mathfrak X)
\end{equation}
The choice to define pure and normal states with respect to two different $C^*$-algebras can be surprising. But we want to consider pure states of the bipartite system (spin+coordinate degrees of freedom). The normal states, the density matrices, onto the coordinate degrees of freedom must result from partial trace of pure states onto the spin degree of freedom. And so the character pure or mixed of the density matrices is related to the character separable or entangled of the pure states of the bipartite system.\\
The pure states replace the notion of points on $\mathfrak M$, in the meaning that $\omega_P(Y)=\tr(PY)$ is equivalent to $P(f)=f(P)$ for $P$ a point in a commutative manifold and $f$ a function onto this manifold; a point $P$ on a commutative manifold defining a Dirac distribution density centred on itself. We do not have a notion of points on a noncommutative manifold, since two different pure states are not necessary separated $P_1P_2 \not=0$ even if $P_1 \not= P_2$. Note that we define the pure states of $\mathfrak M$ by projectors of $\mathbb C^2 \otimes \mathcal H$ and not of $\mathcal H$. We have then $\omega_P(X) = \tr(PX) = \tr(\rho_P X)$ where $\rho_P = \tr_{\mathbb C^2} P$ (with $\tr_{\mathbb C^2}$ the partial trace onto $\mathbb C^2$, note that in $\tr(PX)$ $\tr$ denotes the trace onto $\mathbb C^2 \otimes \mathcal H$ whereas in $\tr(\rho_P X)$ it denotes the trace onto $\mathcal H$). So a pure state of $\mathfrak M$ defines a normal state $\tr(\rho_P \bullet)$ (which is not pure for the $C^*$-algebra $\Env(\mathfrak X)$). $P$ defines also the (mean value) vector $\vec n(P) = (\tr(P \sigma_1), \tr(P \sigma_2),\tr(P \sigma_3)) \in \mathbb R^3$. At this stage we can think $\vec n(P)$ as the mean local orientation of $\mathfrak M$ at the pure state $\omega_P = \tr(P \bullet)$. This is the possible quantum entanglement between the local orientation degree of freedom ($\mathbb C^2$) and the coordinate degrees of freedom ($\mathcal H$) which needs to define pure and normal states of $\mathfrak M$ with respect to two different $C^*$-algebras.\\

It is important not to confuse with the notion of fuzzy manifolds \cite{Chang, Ferraro} which is related to the notion of fuzzy sets which are ``sets'' ``containing'' elements without certainty (each element has a classical probability to belong to the set). The two notions are close but fuzzy manifolds deal with classical probabilities whereas fuzzy spaces deal with quantum probabilities (noncommutative probabilities or free probabilities in the language of the mathematicians). To be more precise and avoid any confusion, it could be more convenient to call \textit{topological fuzzy manifolds} the case associated with classical probabilities and \textit{quantum fuzzy manifolds} or \textit{noncommutative fuzzy manifolds} the case associated with quantum probabilities. But in the whole of this paper, the term \textit{fuzzy spaces} refers only to this last case.\\

The operators $(X^i)$ play the role of noncommutative coordinate observables on $\mathfrak M$. In other words, for $\omega \in \mathcal P(\mathfrak M)$, $\omega(X^i)$ is the mean value of the ``$i$-th coordinate'' of the pure state $\omega$ on $\mathfrak M$, and $\Delta_\omega X^i = \sqrt{\omega({X^i}^2) - \omega(X^i)^2}$ is the quantum uncertainty onto the $i$-th coordinate of $\omega$ on $\mathfrak M$. Due to the noncommutativity of the coordinate observable we have Heisenberg uncertainty relations: $\forall \omega \in \mathcal P(\mathfrak M)$
\begin{equation}
  \Delta_\omega X^i \Delta_\omega X^j \geq \frac{1}{2} |\omega([X^i,X^j])|
\end{equation}
$\mathfrak M$ is said ``fuzzy'' since the coordinates of its pure states are ``delocalised'' as a quantum particle in the usual space. This is the existence of the quantum coordinate observables which is the main particularity of a fuzzy space with respect to any noncommutative manifold.\\

The choice of $\slashed D_x$ is made to correspond to the string theory BFSS matrix model (the spacetime being reduced to 3+1 dimensions by truncation with a supersymmetric orbifold \cite{Berenstein}). In this one, $\mathfrak M$ is a noncommutative D2-brane and the classical parameter $x$ is the classical coordinates of a probe D0-brane; a fermionic string (of spin state space $\mathbb C^2$) links the D2-brane and the probe D0-brane. $\slashed D_x$ is the displacement energy observable (the ``tension energy'' of the fermionic string). If the D0-brane is ``far away'' from the D2-brane the tension energy of the fermionic string increases. From the point of view of the fuzzy geometry, $\slashed D_x$ represents the minimal coupling between the spin (local orientation) observables ($\sigma_i$) and the coordinate observables ($X^i$), responsible of the quantum entanglement. $x \in \mathbb R^3$ is the position of the probe of the classical observer (this one has $\mathbb R^3$ as the physical space in her mind). This probe is needed to define the observation of the fuzzy space by the observer (the observer realises experiments ``somewhere'' in the classical space of her mind). As in usual quantum mechanics, we cannot completely separate the ``properties'' of the quantum system from their observations. Here this does not imply that the Hilbert space is $L^2(\mathbb R^3,dx)$ (not a quantum particle) but this implies the presence of $x$ in $\slashed D_x$ (the geometry observable of $\mathfrak X$). In contrast with usual quantum mechanics where the quantum system and the observer share a common spacetime background, here the quantum system (the quantum spacetime itself) and the classical observer have not a common background. The geometry observable $\slashed D_x$ (which is related both to the quantum system and to the observer) depends on the space observables $(X^i)$ but also on the classical space generated by $(x^i)$ in which the observer places the measurement outcomes. The space generated by $(x^i)$ is not physical for the point of view of $\mathfrak M$ but is a parameter space (in the mind of the observer) permitting for example to define the change of the place where the measures are performed (by varying the values of $x^i$). The spacetime $\mathbb R^3 \times \mathbb R$ permits to the observer to place the events corresponding to the measurement outcomes, permitting to easily endow this set of events with a causal structure \cite{Penrose}. $\mathbb R^3 \times \mathbb R$ is not physical (it is the spacetime in the observer's mind), but the causal structure is. We can also understand this by comparison with commutative geometry in the context of general relativity. To observe the geometry of the spacetime, an observer needs to measure the geodesic moves of a test particle. In a same way, to observe the fuzzy geometry of $\mathfrak M$, an observer needs to make measurement on a test classical particle (of position $x$ in the classical space $\mathbb R^3$ of the measurement outcomes of position).\\

The choice of $\slashed D_x$ corresponds also to a problem of quantum information theory consisting to study a qubit (assimilated to a $1/2$-spin system) manipulated by a magnetic field $\vec B$ and in contact with an environment inducing a decoherence phenomenon \cite{Viennot3}. In the interaction picture, the Hamiltonian of the qubit is $\slashed D_x$ where $x^i = \frac{\mu}{2} B^i$ are the parameters of the magnetic field controlling the qubit to realise single-qubit logical gates ($\mu$ is the magnetic moment magnitude), and $X^i = U_{\mathcal E}^\dagger(t) V^i U_{\mathcal E}(t)$ where $U_{\mathcal E}(t)$ is the evolution operator of the environment and $V = \sigma_i \otimes V^i$ is the interaction operator between the qubit and its environment. $X^i$ are then the environmental noise observables. In a first step, we treat the problem with time-independent observables $X^i$. The time-dependent case will be treated section \ref{TDFS}, where we will see (property \ref{dynFuzzy}) that for the quantum information models, the time-dependent behaviour is deduced from the stationary case. \\

Some authors prefer to consider fuzzy spaces with $\Box_x = \delta_{ij} (X^i -x^i)(X^j-x^j)$ a noncommutative equivalent of a Laplace-Beltrami operator or of a D'Alembertian operator as fundamental geometry observable in place of $\slashed D_x$ (noncommutative equivalent of a Dirac operator), see for example \cite{Schneiderbauer}. We have $\slashed D_x^2 = \id \otimes \Box_x + \frac{\imath}{2} {\varepsilon_{ij}}^k \sigma_k \otimes[X^i,X^j]$, and so $\Box_x = \tr_{\mathbb C^2} \slashed D_x^2$. It is clear that the geometric informations encoded in $\Box_x$ are also in $\slashed D_x$, in contrast $\slashed D_x$ includes informations concerning the local orientation which are lost with $\Box_x$ (these ones are clearly erased by the partial trace onto $\mathbb C^2$). In particular, possible entanglements between states of local orientation and states of location of $\mathfrak M$ is \textit{a priori} ignored if we consider firstly $\Box_x$ in place of $\slashed D_x$. The choice of $\slashed D_x$ as fundamental geometry observable is then more general.\\

\begin{example}{1}{Fuzzy sphere}
  $\mathfrak X$ is generated by $X^i = rJ^i$ with $r\in \mathbb R^{+*}$ and $[J^i,J^j] = \imath {\varepsilon^{ij}}_k J^k$ ($\Lie(X^1,X^2,X^2)$ is the $\mathfrak{su}(2)$ algebra), $\mathcal H = \mathscr H_j$ is the space of the unitary irreducible representation of $\mathfrak{su}(2)$ of dimension $2j+1$ ($j \in \frac{1}{2} \mathbb N$). This Fuzzy space is called Fuzzy sphere since the coordinate observables satisfy the following equation
  \begin{equation}
    (X^1)^2+(X^2)^2+(X^3)^2 = r^2j(j+1)\id
  \end{equation}
  mimicking the equation of a classical sphere of radius $r\sqrt{j(j+1)}$. In quantum gravity matrix model, the thermalisation of a Fuzzy sphere has been proposed as a model of event horizon of a quantum black hole \cite{Dolan, Iizuka}. In quantum information theory, consider a qubit (assimilated to a $\frac{1}{2}$-spin system $\vec S$) controlled by a magnetic field $\vec B$ and interacting with $N$ other qubits $(\vec S_a)_{a=1,...,N}$ forming its environment (the quantum computer) \cite{Viennot3}. The interaction Hamiltonian of the qubit is then
  \begin{eqnarray}
    \slashed D_x & = & - \mu \vec S \cdot \vec B + \mathcal J \vec S \cdot \sum_{a=1}^N \vec S_a \\
    & = & \sigma_i \otimes (X^i-x^i)
  \end{eqnarray}
  ($\mu$ is the magnetic moment magnitude of the spin system and $\mathcal J$ is the exchange integral) with $\vec x = \frac{\mu}{2} \vec B$ and
  \begin{equation}
    X^i = \frac{\mathcal J}{4} \bigoplus_{j=0 \text{ or } 1/2}^{N/2} J^i_{(j)}
  \end{equation}
  where $(J^i_{(j)})_i$ are the $(2j+1)$ dimensional unitary irreducible representations of the $\mathfrak{su}(2)$ generators onto $\mathscr H_j$ ($\mathcal H = \bigoplus_{j=0\text{ or }1/2}^{N/2} \mathscr H_j$). The system can be viewed as $\lceil N/2 \rceil$ concentric Fuzzy spheres of radii $\frac{\mathcal J}{4} \sqrt{j(j+1)}$. 
\end{example}

\begin{example}{2}{Fuzzy surface plots}
  \noindent \textbf{Fuzzy plane:} Let $(a,a^+,\id)$ be the generator of the CCR algebra $[a,a^+]=\id$ (annihilation and creation operators) onto the Fock space $\mathscr F$ (of canonical basis $(|n\rangle)_{n\in \mathbb N}$). The Fuzzy space defined by $\mathcal H = \mathfrak F$, $Z=a$ and $X^3=0$ is a Fuzzy plane. In quantum gravity matrix model, it is the fundamental example of flat spacetime slice \cite{Karczmarek}. The geometry observable can be written in the canonical basis of $\mathbb C^2$ as
  \begin{equation}
    \slashed D_\alpha = \left(\begin{array}{cc} -x^3 & a^+ - \bar \alpha \\ a-\alpha & x^3 \end{array} \right)
  \end{equation}
  with $\alpha = x^1 + \imath x^2$. This operator can be viewed as the Hamiltonian of a qubit considered as an atomic two-level system controlled by a strong classical electric field $\vec E$ in the rotating wave approximation \cite{Puri} ($\alpha=\frac{1}{2 \Omega} \vec \mu \cdot \vec {\mathcal E}$ where $\vec \mu$ is the atomic dipolar moment and $\vec {\mathcal E}$ is the complex representation of the electric field, $\Omega x^3$ being the detuning between the atomic transition energy gap and the field frequency). The atom is also in interaction with a reservoir of single-mode bosons (described by $a$ and $a^+$) constituting its environment ($\Omega$ is the coupling strength between the atom and the bosons). These bosons can be for example photons of another electromagnetic field or phonons of a solid system in which the atom is included.\\
 
  Since $a$ and $a^+$ are unbounded operators, $\mathfrak X$ is not included in its $C^*$-enveloping algebra $\Env(\mathfrak X)$ which is generated by the elements of the form $\mathcal D(\alpha) = e^{\alpha a^+ - \bar \alpha a}$ (for $\alpha \in \mathbb C$). But if (by definition) $\Env(\mathfrak X)$ is closed for the norm topology, it is not for the strong topology. By the Stone theorem we have: $\lim_{s \to 0} \frac{1}{s}(e^{s(\alpha a^+ - \bar \alpha a)}-\id) \psi = (\alpha a^+ - \bar \alpha a) \psi$ ($\forall \psi \in \mathfrak F$), and so $\mathfrak X$ is included into the topological closure of $\Env(\mathfrak X)$ for the strong topology.\\
  The fuzzy plane can be used to define the ``fuzzyfication'' of surface plots $x^3=f(\alpha,\bar \alpha)$ for $f$ a polynomial function or depending only on $|\alpha|$.\\

  \noindent \textbf{Fuzzy elliptic paraboloid:} $Z=a$ and $X^3 = \epsilon(a^+a+\frac{1}{2}\id)$ (with $\epsilon \in \mathbb R^*$). The coordinate observables satisfy $X^3 = \epsilon ((X^1)^2+(X^2)^2)$ mimicking the equation of a classical paraboloid. For the physical point of view, $\slashed D_x$ is an effective Hamiltonian \cite{Puri} where $X^3$ adds a perturbative nonlinear interaction of the two-level atom with its environment (boson scattering by the two-level atom).\\

  \noindent \textbf{Fuzzy hyperbolic paraboloid:} $Z=a$ and $X^3 = \frac{\epsilon}{2} (a^2+(a^+)^2)$, the coordinate observables satisfying $X^3 =\epsilon ((X^1)^2-(X^2)^2)$, another possible perturbative nonlinear interaction (absorption and emission of two bosons by the two-level atom).\\

  \noindent \textbf{Fuzzy hyperboloid:} $Z=a$ and $X^3 = \epsilon \sqrt{(r^2+1/2)\id+a^+a} = \epsilon \sum_{n=0}^{+\infty} \sqrt{r^2+1/2+n}|n\rangle \langle n|$ with $r \in\mathbb R^*$. The coordinate observables satisfy then $(X^3)^2-\epsilon^2(X^1)^2-\epsilon^2(X^2)^2 = \epsilon^2 r^2$.\\
  Note that another way to define a Fuzzy hyperboloid consists to set $X^i =r K^i$ where $K^i$ are unitary irreducible representations of the generators of the $\mathfrak{su}(1,1)$ algebra, obtains a Fuzzy space of the same family than the Fuzzy sphere.\\

  \noindent \textbf{Fuzzy Flamm's paraboloid:} $Z=a$ and $X^3= 2\sqrt{r_S}\sqrt{\sqrt{a^+a}-r_S} = 2\sqrt{r_S} \sum_{n=0}^{+\infty} \sqrt{\sqrt{n}-r_S}|n\rangle \langle n|$ with $r_S \in \mathbb R^+$. This system can be viewed in quantum gravity matrix model as a Fuzzy slice of a Schwarzschild spacetime. It has been numerically studied in a previous work as a toy model of quantum black hole \cite{Viennot4}.
\end{example}

$\mathbb C^2 \otimes \mathcal H$ can be viewed as a right $\mathcal L(\mathbb C^2)$-$C^*$-module \cite{Viennot1} with the following inner product:
\begin{equation}
  \forall \Psi,\Phi \in \mathbb C^2 \otimes \mathcal H, \quad \langle \Psi|\Phi \rangle_* = \tr_{\mathcal H}|\Phi \rrangle \llangle \Psi| \in \mathcal L(\mathbb C^2)
\end{equation}
where $\llangle \bullet | \bullet \rrangle$ stands for the inner product induced by the tensor product of $\mathbb C^2 \otimes \mathcal H$ viewed as a Hilbert space. By construction the square $C^*$-norm of $\Psi \in \mathbb C^2 \otimes \mathcal H$, $\langle \Psi|\Psi\rangle_* = \tr_{\mathcal H} |\Psi \rrangle \llangle \Psi|$ is a density matrix onto $\mathbb C^2$. The reason of the consideration of this structure is the following. The choice of the local orientation of $\mathfrak M$ is arbitrary. It is a gauge choice. Let $u \in SU(2) \subset \mathcal L(\mathbb C^2)$ be a change of local orientation of $\mathfrak M$. Let $\omega_P \in \mathcal E(\mathfrak M)$. The state $\omega_{uPu^{-1}}$ has then the same physical meaning than the state $\omega_P$ (in the same way than in commutative geometry, the change of the local orientation at the neighbourhood of a point on a classical manifold does not change the point). So $u$ for the fuzzy space is equivalent for a quantum particle to a phase change. If the phases for the states of the fuzzy space are elements of $U(2)$, $\mathbb C$ must be replaced by $\mathcal L(\mathbb C^2)$ as algebra of ``scalars''. Note that $\langle \bullet | \bullet \rangle_*$ is not invariant under $SU(2)$-phase changes but is equivariant:
\begin{equation}
  \langle u\Psi|u\Phi\rangle_* = u \langle \Psi|\Phi \rangle_* u^{-1}
\end{equation}
(in a commutative algebra as $\mathbb C$, equivariance and invariance are the same thing).\\
For the BFFS matrix model, $u \in SU(2)$ is a rotation of the fermionic string spin. For the point of view of the quantum information model, the $C^*$-norm of a state is the density matrix of the qubit under the effects of the environment and the action of $u$, $\|u\Psi\|^2_* = u\|\Psi\|^2_*u^{-1}$ is a (local) unitary operation onto the qubit. These density matrices define then probability laws onto $\mathbb C^2$:
\begin{defi}[Almost surely properties]
  Let $\rho = \|\Psi\|_*^2$ be a density matrix of $\mathbb C^2$. Let $p_F$ be a property satisfied by an observable $A \in \mathcal L(\mathbb C^2)$ if $F(A)=0$ for $F: \mathcal L(\mathbb C^2) \to \mathcal L(\mathbb C^2)$ (not necessarily linear). We say that $A$ satisfies $p_F$ almost surely with respect to $\rho$ ($\rho$-a.s.), if $\tr(\rho F(A))=0$. If $\mathcal F = \{A \in \mathcal L(\mathbb C^2), \text{ such that } F(A)=0\}$ we denote by $\mathcal F_{a.s.} = \{A \in \mathcal L(\mathbb C^2), \text{ such that } \tr(\rho F(A))=0\}$ the set of observables satisfying $p_F$ almost surely ($\mathcal F \subset \mathcal F_{a.s.}$).
\end{defi}
For example, with $F(A)=A^\dagger-A$, $A$ is said almost surely self-adjoint if $\tr(\rho(A^\dagger-A))=0$.\\ 

The dynamics of $\Psi \in \mathbb C^2 \otimes \mathcal H$ is governed by the Schr\"odinger-like equation:
\begin{equation}\label{SchroEq}
  \imath \dot \Psi = \slashed D_{x(t)} \Psi
\end{equation}
Note that $x$ can be time dependent if the observer moves its probe. This one is the Dirac equation of a massless fermionic string in the BFSS model in Weyl representation \cite{Berenstein}. The inner dynamics of $\mathfrak M$ is then submitted to the following evolution of their pure states:
\begin{equation}
  \omega_P \circ \Ad_{U_{\slashed D_x}(t)} = \omega_{U_{\slashed D_x}(t) P U_{\slashed D_x}^{-1}(t)}
\end{equation}
(with $\Ad_u(X) = u^{-1}Xu$) where $U_{\slashed D_x}(t) = \Teg^{-\imath \int_0^t \slashed D_{x(t)} dt} \in \mathcal U(\mathbb C^2\otimes \mathcal H)$ is the evolution operator ($\Teg$ stands for the time ordered exponential - the Dyson series -, i.e. $\imath \partial_t U_{\slashed D_x} = \slashed D_{x(t)} U_{\slashed D_x}$ with $U_{\slashed D_x}(0)=\id$). Even if we consider a state without entanglement between the local orientation and the coordinate states, $U_{\slashed D_x}(t)$ being not separable (in general) it induces the growing of the entanglement with the time.\\
We are interested by the adiabatic regimes for this equation, because they are similar to semi-classical approximations for the dynamics (for the point of view of the evolution operator $\Teg^{-\imath \frac{\hbar}{T} \int_0^s \slashed D_{x(s)}ds}$ -- by restoring the writing of the fundamental constants -- the semi-classical limit $\hbar \to 0$ and the adiabatic limit $T \to +\infty$ are equivalent, with $s=t/T$, $T$ being the total duration of the dynamics). The adiabatic regimes are then associated with the situations closest to classical dynamics. Suppose that $|\Lambda(x)\rrangle$ is an instantaneous eigenvector of $\slashed D_x$, the dynamics is said:
\begin{itemize}
\item strongly adiabatic if \cite{Teufel}
  \begin{equation}
    |\Psi(t)\rrangle \simeq e^{\imath \varphi(t)} |\Lambda(x(t))\rrangle, \qquad e^{\imath \varphi(t)} \in U(1) \label{StrAdiabTransp}
  \end{equation}
\item weakly adiabatic if \cite{Viennot3}
  \begin{equation}
    |\Psi(t)\rrangle \simeq u(t) |\Lambda(x(t))\rrangle, \qquad u(t) \in U(2)_{a.s.} \label{WkAdiabTransp}
  \end{equation}
\end{itemize}
($u$ is almost surely unitary with respect to $\rho_\Lambda = \tr_{\mathcal H} |\Lambda \rrangle \llangle \Lambda|$). So, the regime is strongly adiabatic if it is adiabatic for whole bipartite system (the state remains on the instantaneous eigenvector up to a phase change), and weakly adiabatic if it is adiabatic for the subsystem described by $(X^i)$ but not for the one described by $(\sigma^i)$ (the state remains on the instantaneous eigenvector up to a local orientation change). For the point of view of quantum information theory, in the strong adiabatic cases, the control $t \mapsto x(t)$ onto the qubit $(\sigma_i)$ is sufficiently slow not to  induce transition of its instantaneous state nor of the environment state. But in weak adiabatic cases, the control induces transition of the instantaneous qubit state, but the environment $(X^i)$ is sufficiently large to the evolution of the qubit does not induce transition of its instantaneous state via the couplings $\sigma_i \otimes X^i$. For the point of view of the quantum gravity, the strong adiabatic regime corresponds to move the probe sufficiently slowly not to induce transition of the instantaneous quantum spacetime state. In the weak adiabatic regime, the local orientation defined by the probe spin $(\sigma_i$) is rotated by gravitational effects, but the move is sufficiently slow not to deform spacetime by transition of the instantaneous location state (associated with $(X^i)$).  

\section{The eigen geometry}
\subsection{The eigenmanifold and the quasicoherent states}
$\mathfrak M$ is a quantum object, but the outcomes of the measurements are classical quantities. These outcomes are eigenvalues of pertinent observables. We can think that the spectral properties of the geometry observable $\slashed D_x$ define a classical manifold close to $\mathfrak M$.
\begin{defi}[Eigenmanifold]
  We call eigenmanifold $M_\Lambda$ of a fuzzy space $\mathfrak M$ the subset of $\mathbb R^3$ defined by
  $$ M_\Lambda = \{x \in \mathbb R^3, \ker(\slashed D_x) \not= \{0\} \} $$
\end{defi}
Note that $\dim M_\Lambda \leq 2$ and $M_\Lambda$ is not necessary connected.

\begin{defi}[Quasicoherent state]
  Let $x \in M_\Lambda$ a point on the eigenmanifold of a fuzzy space $\mathfrak M$. We call quasicoherent state of $\mathfrak M$ at $x$ a normalised vector $|\Lambda(x)\rrangle \in \mathbb C^2 \otimes \mathcal H$ such that
  \begin{equation}
    \slashed D_x |\Lambda(x)\rrangle = 0
  \end{equation}
\end{defi}
The reason of these definitions is the following. Let $\omega_x = \tr(|\Lambda(x)\rrangle \llangle \Lambda(x)|\bullet) \in \mathcal P(\mathfrak M)$ be the pure state associated with $|\Lambda(x)\rrangle$. We have $\omega_x(X^i) = x^i$ (because $\frac{1}{2}\{\sigma^i,\slashed D_x\}= X^i-x^i$) and we can prove that $\omega_x$ minimises the Heisenberg uncertainty relation $\Delta_{\omega_x} \vec X^2 = \omega_x(\vec X^2)-\omega_x(\vec X)^2$ \cite{Schneiderbauer} as the Perelomov coherent states of an usual quantum system \cite{Perelomov}. Coherent states being the quantum states closest to classical ones (since the quantum delocalisation is minimised), the quasicoherent states $\omega_x$ are the pure states of $\mathfrak M$ closest to classical pure states of a classical (commutative) manifold. So, $\omega_x$ are the pure states of $\mathfrak M$ closest to the classical notion of points. We can then think $\omega_x$ as a ``quantum point'' of $\mathfrak M$, which is naturally labelled by a point $x$ of the classical space of the observer. Nevertheless, in general $\llangle \Lambda(y)|\Lambda(x)\rrangle \not= \delta(x-y)$, the pseudo-points $\omega_x$ and $\omega_y$ are not separated (if the system is in the state $|\Lambda(x)\rrangle$ the probability that the outcome of the location measured by the observer be $y$ is not zero: $|\omega_x(|\Lambda(y)\rrangle\llangle \Lambda(y)|)|^2=|\llangle \Lambda(y)|\Lambda(x)\rrangle|^2$). From the viewpoint of string theory, $|\Lambda(x)\rrangle$ being the eigenvector of $\slashed D_x$ associated with the zero eigenvalue, it is the state which minimises the displacement energy (the ``tension energy'' of the fermionic string). The fermionic string having no tension, its two ends are ``at the same place''. So the probe D0-brane is closest as possible to the D2-brane. These three arguments (minimal quantum delocalisation, pure states closest to points, minimal ``distance'' between the probe and the noncommutative manifold) show that the eigenmanifold $M_\Lambda$ is the classical (commutative) manifold closest to $\mathfrak M$ (the geometry defined by $M_\Lambda$ is the classical geometry which looks the most the quantum geometry of the fuzzy space). $M_\Lambda$ is the slice of space which is the more representative of the geometry of $\mathfrak M$.\\
By construction, in quantum information theory, restricting $x$ to be on $M_\Lambda$ consists to adiabatically control the qubit by keeping constant at $0$ its energy dressed by the environment with zero energy uncertainty. If we denote by $H_{int}$ any interaction Hamiltonian, $\slashed D_x = H_{int} - \tr_{\mathbb C^2} H_{int}$ is the part which induces the state transitions and which is leave invariant by changes of the potential energy origin. The eigenstate of zero eigenvalue of $\slashed D_x$ is then the state minimising the energy exchanges between the qubit and the environment. Moreover, in contrast with the other eigenstates of $\slashed D_x$, the quasicoherent state minimises the quantum uncertainties $\Delta X^i$ and so minimises the noises coming from the environment and perturbing the qubit (and responsible of the decoherence of qubit mixed state $\rho$). Adiabatically controlling the qubit on $M_\Lambda$ permits then to minimise the effects of the environment noises onto the qubit.\\

\begin{example}{1}{Fuzzy sphere}
  The eigenmanifold of a Fuzzy sphere is the classical sphere $M_\Lambda = \{\vec x \in \mathbb R^3, \|\vec x\| = rj\}$ and its quasicoherent state is \cite{Viennot2}
  \begin{equation}
    |\Lambda(x)\rrangle = |\zeta\rangle_{1/2} \otimes |\zeta \rangle_j
  \end{equation}
  with $\zeta = e^{\imath \varphi} \tan \frac{\theta}{2}$ for $x=rj(\sin\theta \cos \varphi,\sin \theta \sin \varphi,\cos \theta)$. $|\zeta\rangle_j$ are the Perelomov $SU(2)$ coherent states \cite{Perelomov} for the unitary irreducible representation of dimension $2j+1$. The eigenmanifold is plotted fig.\ref{FuzSphPlot}.
  \begin{figure}
    \center
    \includegraphics[width=7cm]{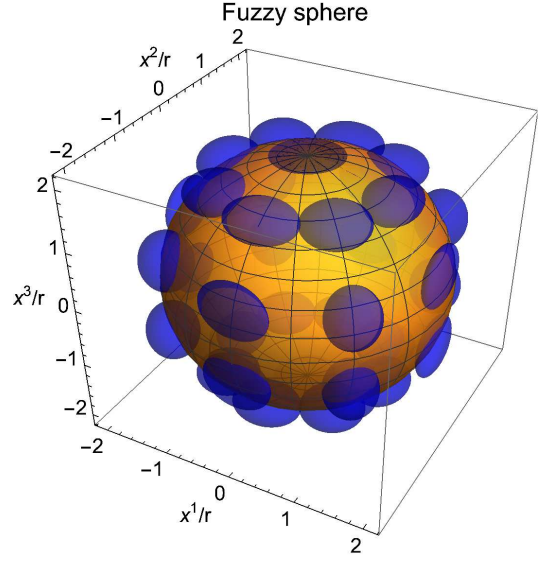}\\
    \caption{\label{FuzSphPlot} Eigenmanifold $M_\Lambda : \zeta\mapsto \llangle \Lambda(\zeta)|\vec X|\Lambda(\zeta)\rrangle$ (yellow surface) for the Fuzzy sphere (with $j=2$) with the quantum uncertainty clouds $\frac{1}{2}\Delta X^i = \frac{1}{2}\sqrt{\llangle \Lambda|(X^i)^2|\Lambda \rrangle - (\llangle \Lambda|X^i|\Lambda \rrangle)^2}$ (blue ellipsoids) of some points, in the target space $\mathbb R^3$ of the probe.}
  \end{figure}
\end{example}

We see with this example that the quasicoherent states of the Fuzzy spaces are strongly related to Perelomov coherent states of Lie algebras \cite{Perelomov}. This is confirmed by the following other examples.

\begin{example}{2}{Fuzzy surface plots}
  \noindent \textbf{Fuzzy plane:} The eigenmanifold of the Fuzzy plane is just the complex plane $M_\Lambda = \{(\Re(\alpha),\Im(\alpha),0\}_{\alpha \in \mathbb C}$ with the quasicoherent state in the canonical basis of $\mathbb C^2$ \cite{Viennot2}:
  \begin{equation}
    |\Lambda(\alpha) \rrangle = \left(\begin{array}{c} |\alpha \rangle \\ 0 \end{array} \right)
  \end{equation}
  where $|\alpha \rangle$ is the Perelomov coherent state of the quantum harmonic oscillator algebra. The quantum uncertainties are $\Delta X^1 = \Delta X^2 = 1/2$ ($\Delta X^3=0$). To treat the other examples of fuzzy surface plots, it is interesting to consider all eigenvectors of $\slashed D_\alpha$ \cite{Viennot2, Viennot4}:
  \begin{equation}
    \slashed D_\alpha |\lambda_{n\pm}(\alpha)\rrangle = \pm \sqrt{n} |\lambda_{n\pm}(\alpha)\rrangle
  \end{equation}
  with
  \begin{equation}
    |\lambda_{n\pm} (\alpha) \rrangle = \frac{1}{\sqrt 2} \left(\begin{array}{c} |n\rangle_\alpha \\ \pm |n-1\rangle_\alpha \end{array} \right)
  \end{equation}
  where $|0\rangle_\alpha \equiv |\alpha \rangle$ and
  \begin{equation}
    |n\rangle_\alpha = \frac{(a^+-\bar \alpha)^n}{\sqrt{n!}}|0\rangle_\alpha
  \end{equation}
  $(|n\rangle_\alpha)_{n\in \mathbb N}$ form a ``translated'' basis of $\mathscr F$: $a_\alpha|n\rangle_\alpha = \sqrt n |n-1\rangle_\alpha$ ($a_\alpha|0\rangle_\alpha = 0$) and $a_\alpha^+|n\rangle_\alpha = \sqrt{n+1}|n+1\rangle_\alpha$ with $a_\alpha \equiv a - \alpha$ and $a_\alpha^+ \equiv a^+-\bar \alpha$. In the sequel, the quasicoherent state of the Fuzzy plane will be denoted by $|\lambda_0\rrangle$. \\
  Let $U=e^{\imath \vartheta a^+a}$ be the rotation operator in the complex plane: $Ua^+U^\dagger = e^{\imath \vartheta} a^+$, $U|\alpha\rangle = |e^{\imath \vartheta} \alpha\rangle$. Since $U(a^+-\bar \alpha)U^\dagger = e^{\imath \vartheta}(a^+-e^{-\imath \vartheta} \bar \alpha)$, we have $U|n\rangle_{\alpha} = e^{\imath n \vartheta} |n\rangle_{e^{\imath \vartheta} \alpha}$, and then $|\llangle \lambda_{n\pm}(\alpha)|U|\lambda_{n\pm}(\alpha) \rrangle|^2 = \frac{1+\cos(\vartheta)}{2}$. For $n>0$, the eigenvectors of the fuzzy plane are not invariant by rotation around $0$. So the perturbed quasicoherent states of the fuzzy plane are not invariant by rotation around $0$ even if the perturbation operator is.\\

  \noindent \textbf{Fuzzy elliptic paraboloid:} For $\epsilon \ll 1$ and $|\alpha|^2$ not too large (in order to $\epsilon|\alpha|^2 \ll 1$) we can treat the system by using the perturbation theory \ref{perturbation}. The vertical deformation of the plane is then
  \begin{eqnarray}
    \delta x^3(\alpha) & = & \langle \alpha|X^3|\alpha \rangle \\
    & = & \epsilon (|\alpha|^2+1/2)
  \end{eqnarray}
  and so $M_\Lambda = \{(\Re(\alpha),\Im(\alpha),\epsilon (|\alpha|^2+1/2))\}_{\alpha \in \mathbb C}$. Since ${_\alpha}\langle n|(a^+a-|\alpha|^2)|0\rangle_\alpha = \alpha \delta_{n,1}$ we have
  \begin{eqnarray}
    |\Lambda(\alpha) \rrangle & = & |\lambda_0\rrangle - \sum_{n>0,\varsigma\in\{\pm\}} \frac{\llangle \lambda_{n\varsigma}|\sigma_3 \otimes(X^3-\delta x^3)|\lambda_0\rrangle}{\varsigma \sqrt{n}}|\lambda_{n\varsigma}\rrangle + \mathcal O(\epsilon^2) \\
    & = & \left(\begin{array}{c} 1 \\ -\epsilon \alpha \end{array} \right) \otimes |0\rangle_\alpha + \mathcal O(\epsilon^2)
  \end{eqnarray}
  The quantum uncertainties are $\Delta X^1=\Delta X^2=1/2$ and $\Delta X^3 = \epsilon |\alpha|$.\\

  \noindent \textbf{Fuzzy hyperbolic paraboloid:} for the same conditions ($\epsilon|\alpha|^2 \ll 1$), we have
  \begin{equation}
    \delta x^3(\alpha) = \epsilon \Re(\alpha^2)
  \end{equation}
  The eigenmanifold is then $M_\Lambda = \{(\Re(\alpha),\Im(\alpha), \epsilon (\Re(\alpha)^2-\Im(\alpha)^2))\}_{\alpha \in \mathbb C}$. Since ${_\alpha}\langle n|\frac{a^2+(a^+)^2}{2}|0\rangle_\alpha = \bar \alpha \delta_{n,1} + \frac{\sqrt 2}{2} \delta_{n,2}$ the quasi-coherent state is
  \begin{equation}
    |\Lambda(\alpha)\rrangle = \left(\begin{array}{c} |0\rangle_\alpha \\ -\epsilon \bar \alpha |0\rangle_\alpha - \frac{\epsilon}{2}|1\rangle_\alpha \end{array} \right) + \mathcal O(\epsilon^2)
  \end{equation}
  In contrast with the previous examples, the quasicoherent state is entangled. The model being as simple as the previous ones, this entanglement results probably from the hyperbolicity of the geometry. The quantum uncertainties are $\Delta X^1=\Delta X^2=1/2$ and $\Delta X^3 = \epsilon \sqrt{|\alpha|^2+1/2}$.\\

  \noindent \textbf{Fuzzy hyperboloid:} we have no a simple expression of the vertical deformation of the plane, but (for $\epsilon|\alpha|\ll 1$) we have
  \begin{equation}
    \delta x^3(\alpha) = \epsilon e^{-|\alpha|^2} \sum_{n=0}^{+\infty} \frac{|\alpha|^{2n}}{n!} \sqrt{r^2+1/2+n}
  \end{equation}
  $\delta x^3$ is in fact very close to the hyperboloid $\alpha \mapsto \sqrt{r^2+1/2+|\alpha|^2}$ as we can see this fig.\ref{CompaHyper}.
  \begin{figure}
    \center
    \includegraphics[width=7cm]{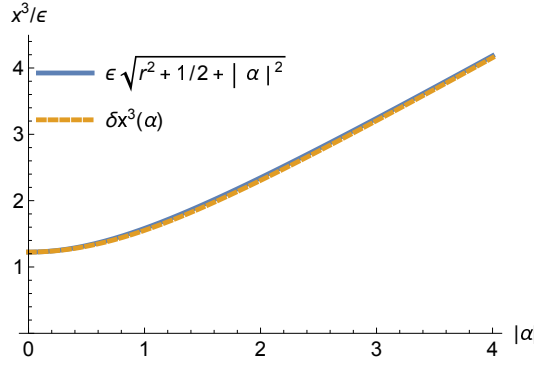}\\
    \caption{\label{CompaHyper} Comparison between $\alpha \mapsto \delta x^3(\alpha) = \llangle \Lambda(\alpha)|X^3|\Lambda(\alpha)\rrangle = \epsilon e^{-|\alpha|^2} \sum_{n=0}^{+\infty} \frac{|\alpha|^{2n}}{n!} \sqrt{r^2+1/2+n}$ for the fuzzy hyperboloid with the classical hyperboloid $\alpha \mapsto \sqrt{r^2+1/2+|\alpha|^2}$ with $r=1$.}
  \end{figure}
  The quasi-coherent state is
  \begin{equation}
    |\Lambda(\alpha)\rrangle = \left(\begin{array}{c} |0\rangle_\alpha \\ 0 \end{array}\right)-\sum_{n=1}^{\infty} \frac{{_\alpha}\langle n|X^3|0\rangle_\alpha}{\sqrt n} \left(\begin{array}{c} 0 \\ |n-1\rangle_\alpha \end{array} \right) + \mathcal O(\epsilon^2)
  \end{equation}
  with
  \begin{equation} \scriptstyle
    {_\alpha}\langle n|X^3|0\rangle_\alpha = \epsilon e^{-|\alpha|^2} \sum_{q=0}^{+\infty} \sum_{p=0}^{\max(q,n)} \frac{(-1)^{n-p}\alpha^{n-p+q} \bar \alpha^{q-p}}{(q-p)!} \sqrt{\frac{n!}{(n-p)!p!}} \sqrt{r^2+1/2+q}
  \end{equation}
  The quantum uncertainties are $\Delta X^1=\Delta X^2=1/2$ and $\Delta X^3= \sqrt{\epsilon^2(r^2+1/2+|\alpha|^2)-(\delta x^3(\alpha))^2}$.\\

  The different eigenmanifolds are plotted fig.\ref{FuzSurfPlot}.
    \begin{figure}
    \center
    \includegraphics[width=7cm]{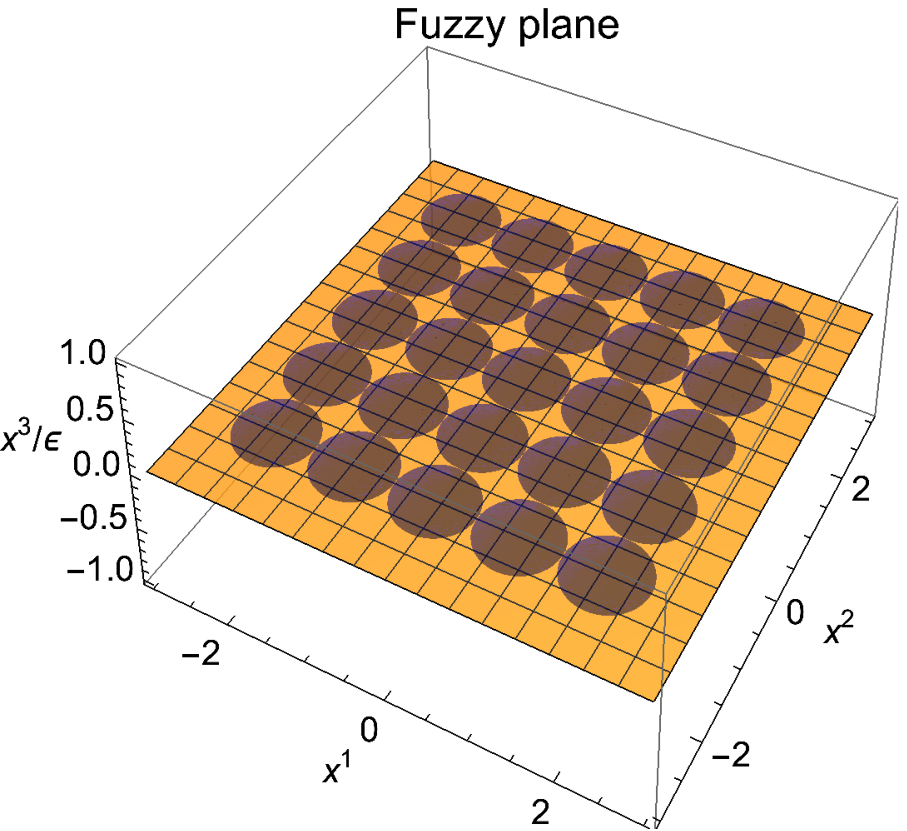}\includegraphics[width=7cm]{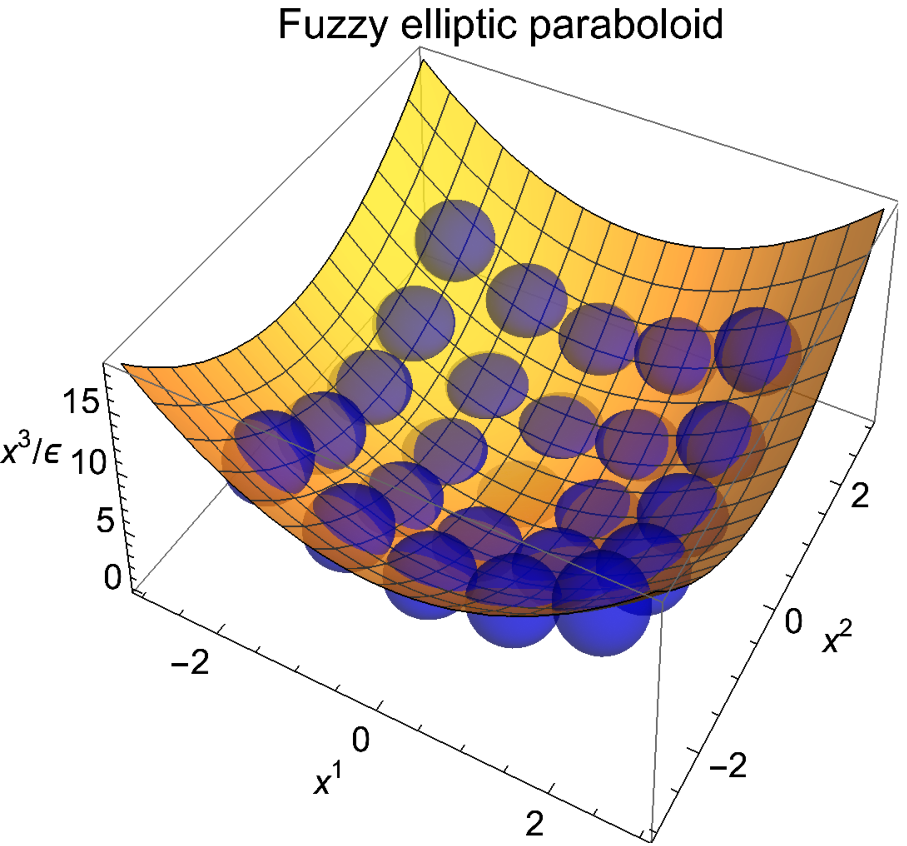}\\
    \includegraphics[width=7cm]{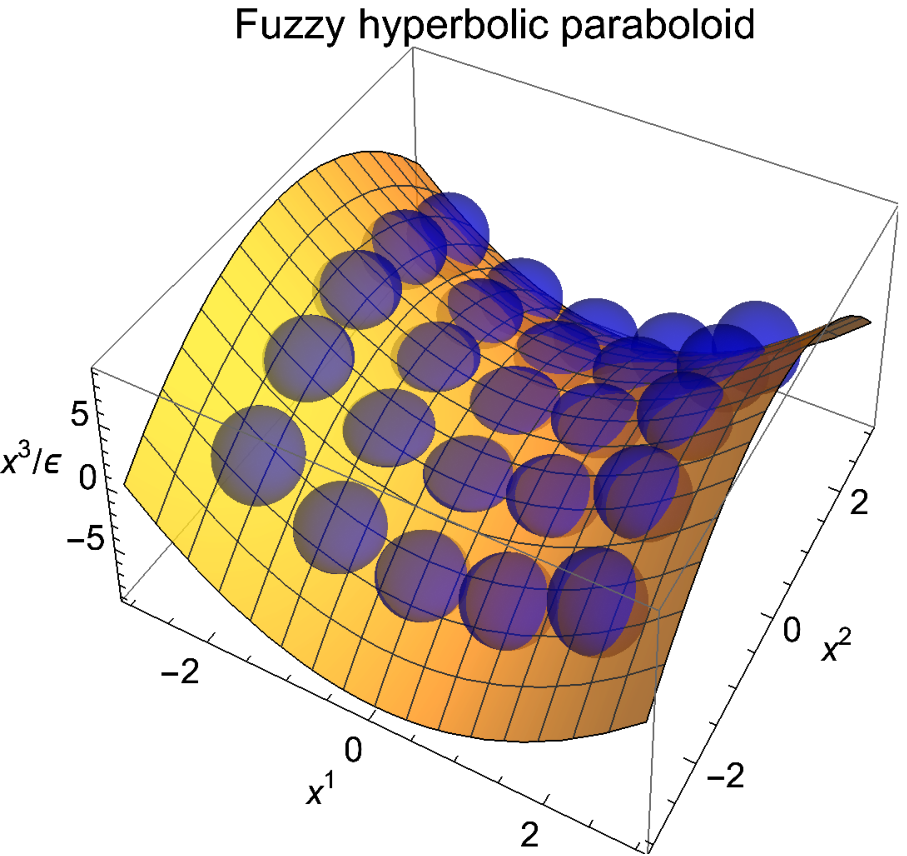}\includegraphics[width=7cm]{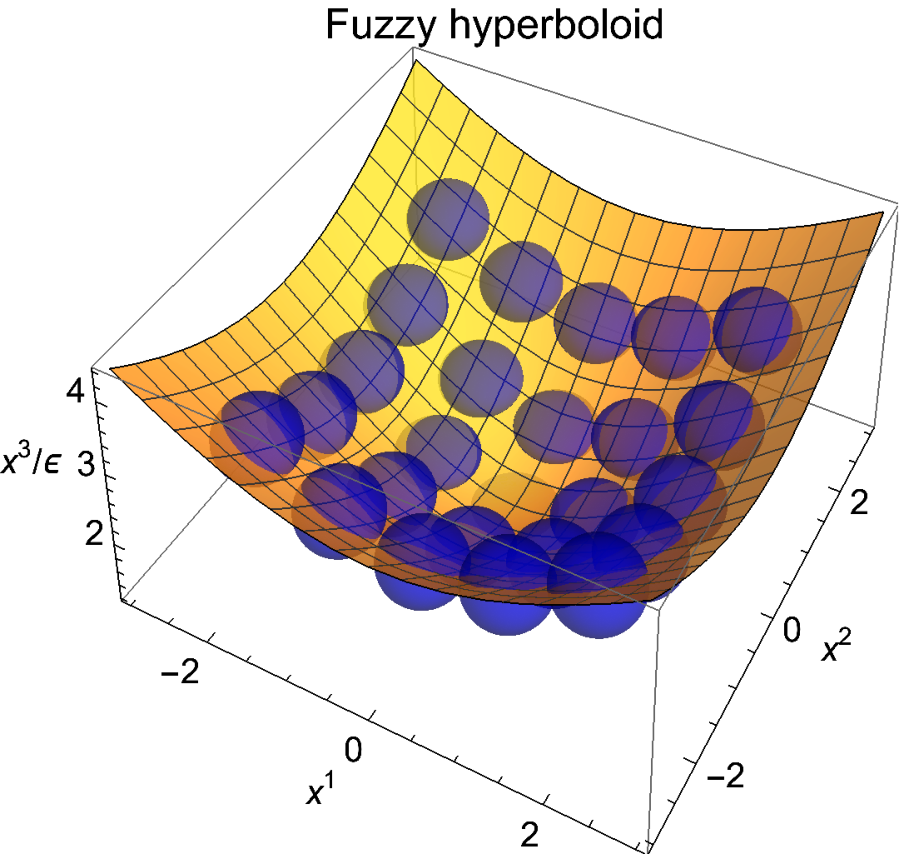}
    \caption{\label{FuzSurfPlot} Eigenmanifolds $M_\Lambda : \alpha\mapsto \llangle \Lambda(\alpha)|\vec X|\Lambda(\alpha)\rrangle$ (yellow surfaces) for the Fuzzy surface plots with the quantum uncertainty clouds $\Delta X^i = \sqrt{\llangle \Lambda|(X^i)^2|\Lambda \rrangle - (\llangle \Lambda|X^i|\Lambda \rrangle)^2}$ (blue ellipsoids) of some points, in the target space $\mathbb R^3$ of the probe. $r=1$ for the Fuzzy hyperboloid.}
  \end{figure}
  
\end{example}

\begin{prop}
  $\vec n_x = \tr(|\Lambda(x)\rrangle\llangle \Lambda(x)| \vec \sigma)$ is a normal vector at $x$ of $M_\Lambda$ in $\mathbb R^3$
\end{prop}

\begin{proof}
  We denote by $d$ the differential on $M_\Lambda$. $\slashed D_x |\Lambda(x)\rrangle = 0 \Rightarrow \sigma_i dx^i |\Lambda(x) \rrangle = \slashed D_x d|\Lambda(x)\rrangle$. So $\llangle \Lambda(x)|\sigma_i|\Lambda(x)\rrangle dx^i = 0$. $dx^i = \frac{\partial x^i}{\partial s^a} ds^a$ (for $(s^1,s^2$) a local curvilinear coordinate system on $M_\Lambda$). So $\vec n_x \cdot \frac{\partial \vec x}{\partial s^a} = 0$ ($\forall a$). $\vec n_x$ is then normal to any tangent vector of $M_\Lambda$ at $x$.
\end{proof}

As normal vector to the surface $M_\Lambda$, $\vec n_x$ defines a local orientation of this manifold by the right hand rule. This is the reason why we assimilate the spin degree of freedom ($\mathbb C^2$) as a quantum local orientation on $\mathfrak M$.\\

Now we want to interpret the norm of $\vec n_x$. Let $(|0\rangle,|1\rangle)$ be the canonical basis of $\mathbb C^2$: $|\Lambda(x) \rrangle = |0\rangle \otimes |\Lambda_0(x)\rangle + |1\rangle \otimes |\Lambda_1(x)\rangle$. By construction the normal vector is
\begin{equation}
  \vec n_x = \llangle \Lambda(x)|\vec \sigma|\Lambda(x) \rrangle = \left(\begin{array}{c} 2\Re \langle \Lambda_0|\Lambda_1\rangle \\ 2\Im \langle \Lambda_0|\Lambda_1\rangle \\ \langle \Lambda_0|\Lambda_0\rangle - \langle \Lambda_1 | \Lambda_1 \rangle) \end{array} \right)
\end{equation}
Moreover we have in the same basis $(|0\rangle,|1\rangle)$:
\begin{equation}
  \rho_\Lambda(x) = \tr_{\mathcal H}|\Lambda(x)\rrangle\llangle \Lambda(x)| = \left(\begin{array}{cc} \langle \Lambda_0|\Lambda_0\rangle & \langle \Lambda_1|\Lambda_0\rangle \\ \langle \Lambda_0|\Lambda_1\rangle & \langle \Lambda_1|\Lambda_1 \rangle \end{array} \right)
\end{equation}
$\langle \Lambda_0|\Lambda_0\rangle \equiv p_0$ and $\langle \Lambda_1|\Lambda_1\rangle \equiv p_1$ are the the populations of the states $(|0\rangle,|1\rangle)$ in the eigengeometry and $\langle \Lambda_0|\Lambda_1\rangle \equiv c$ is its quantum coherence. The normal vector is then
\begin{equation}
  \vec n_x = \left(\begin{array}{c} 2 \Re c \\ 2 \Im c \\ p_0-p_1 \end{array} \right)
\end{equation}
Finally we have
\begin{eqnarray}
  \|\vec n_x\| & = & \sqrt{(p_0-p_1)^2+4|c|^2} \\
  & = & \sqrt{(\tr \rho_\Lambda)^2 - 4 \det \rho_\Lambda} \\
  & = & s_+(x)-s_-(x)
\end{eqnarray}
where $s_\pm = \frac{\tr \rho_\Lambda \pm \sqrt{(\tr \rho_\Lambda)^2 - 4 \det \rho_\Lambda}}{2}$ are the eigenvalues of $\rho_\Lambda$. $\|\vec n_x\|$ is then a measure of the purity of $\rho_\Lambda$, since $\rho_\Lambda$ is pure state if $s_+=1$ and $s_-=0$ ($\|\vec n_x\|=1$) and is maximally mixed if $s_+=s_-=\frac{1}{2}$ ($\|\vec n_x\| = 0$, microcanonical distribution). In other words, $1-\|\vec n_x\|$ is a measure of the entanglement of $|\Lambda \rrangle$, this one is separable if $\|\vec n_x\| = 1$ and is maximally entangled if $\vec n_x = \vec 0$. In that case, in the eigenbasis of $\rho_\Lambda$ we have $|\Lambda(x) \rrangle = |s_+(x)\rangle \otimes |\Lambda_+(x)\rangle + |s_-(x)\rangle \otimes |\Lambda_-(x)\rangle$ with $\langle \Lambda_\alpha|\Lambda_\beta\rangle = \frac{1}{2} \delta_{\alpha \beta}$ ($|\Lambda \rrangle$ is a Bell state).\\
For these reason we can refer to $\vec n_x$ as being the ``purity normal vector'' at $M_\Lambda$. 

\subsection{Degeneracy of the eigenmanifold}
We say that $M_\Lambda$ is strongly non-degenerate at $x$ if $\dim \ker(\slashed D_x) = 1$. This is the usual definition of the non-degeneracy of an eigenvalue in the Hilbert space $\mathbb C^2 \otimes \mathcal H$. But it can be more natural to consider this space with its $C^*$-module structure. So if there exists two quasicoherent states at $x$ related by a $SU(2)$-phase change, we can consider that 0 is not degenerate in the $C^*$-module. Note that if $|\Lambda(x)\rrangle$ and $u|\Lambda(x)\rrangle$ are two orthogonal quasicoherent states of $\ker(\slashed D_x)$, $(\alpha + \beta u)|\Lambda(x)\rrangle$ (with $|\alpha|^2+|\beta|^2=1$) is also a quasicoherent state, but in general $(\alpha+\beta u)$ is not unitary: $(\alpha+\beta u)^\dagger (\alpha+\beta u) = |\alpha|^2+|\beta|^2u^\dagger u + \bar \alpha \beta u + \alpha \bar \beta u^\dagger = 1 +  \bar \alpha \beta u + \alpha \bar \beta u^\dagger$. But $\llangle \Lambda(x)|u|\Lambda(x)\rrangle = 0 \iff \tr(\rho_\Lambda(x)u)=0$ with $\rho_\Lambda = \tr_{\mathcal H}|\Lambda \rrangle \llangle \Lambda|$ the eigen density matrix. We have then $\tr(\rho_\Lambda(\alpha+\beta u)^\dagger (\alpha+\beta u))=0$ and so $(\alpha + \beta u) \in U(2)_{a.s.}$.
\begin{defi}[Weakly non-degenerate fuzzy space]
  A fuzzy space $\mathfrak M$ is said weakly non-degenerate at $x \in M_\Lambda$ if $\exists |\Lambda(x)\rrangle \in \mathbb C^2 \otimes \mathcal H$ and $\exists u_i \in U(2)_{a.s.}$ such that $\ker(\slashed D_x) = \Lin(|\Lambda(x)\rrangle, u_1|\Lambda(x)\rrangle,...,u_n|\Lambda(x)\rrangle)$. Moreover:
  \begin{itemize}
  \item $\mathfrak M$ is said strongly non-degenerate at $x$ if $\dim \ker(\slashed D_x) = 1$.
  \item if $\mathfrak M$ is not strongly non-degenerate at $x$ and if $|\Lambda(x)\rrangle$ is separable, then $\dim \ker(\slashed D_x) = 2$ and $\ker(\slashed D_x) = SU(2)\cdot \Lin(|\Lambda(x)\rrangle)$.
  \item if $\mathfrak M$ is not strongly non-degenerate at $x$ and if $|\Lambda(x)\rrangle$ is entangled, then $\dim \ker(\slashed D_x) \leq 4$.
  \end{itemize}
\end{defi}

\begin{proof}
  We suppose that the quasicoherent state is separable $|\Lambda(x)\rrangle = |O_x \rangle \otimes |\Omega_x\rangle$ and that $u|\Lambda(x)\rrangle = u|O_x\rangle \otimes |\Omega_x \rangle$ is another quasi-coherent state orthogonal to the first one ($u \in U(2)_{a.s.}$). We have then $\langle I_x|O_x\rangle = 0$ with $|I_x\rangle = u|O_x\rangle$. So $(|O_x\rangle,|I_x\rangle)$ is an orthonormal basis of $\mathbb C^2$, and finally $\ker (\slashed D_x) = \{(\alpha|O_x\rangle + \beta|I_x\rangle)\otimes |\Omega_x \rangle, \alpha,\beta \in \mathbb C\} = U(2)\cdot|\Lambda(x)\rrangle$ (we cannot find a third quasi-coherent state $u'|\Lambda(x)\rrangle$ orthogonal to the first two ones, since $u'|O_x\rangle$ cannot be orthogonal to $|O_x\rangle$ and $|I_x\rangle$ except if it is zero).\\
  We suppose  that the quasicoherent state is entangled $|\Lambda(x) \rrangle = |O_x \rangle \otimes |\Lambda^0_x\rangle + |I_x \rangle \otimes |\Lambda^1_x\rangle$ for some orthonormal basis $(|O_x\rangle,|I_x\rangle)$ of $\mathbb C^2$ ($|\Lambda^0_x\rangle \not= |\Lambda^1_x\rangle$). Let $u|\Lambda(x)\rrangle = |O_x' \rangle \otimes |\Lambda^0_x\rangle + |I_x' \rangle \otimes |\Lambda^1_x\rangle$ be another quasi-coherent state orthogonal to the first one ($u \in U(2)_{a.s.}$), with $|O'_x\rangle = u|O_x\rangle$ and $|I'_x\rangle = u|I_x\rangle$.
  \begin{eqnarray}
    \llangle \Lambda(x)|u|\Lambda(x)\rrangle = 0 & \iff & \langle O_x|O'_x\rangle \langle \Lambda^0_x|\Lambda^0_x \rangle + \langle O_x|I'_x\rangle \langle \Lambda^0_x|\Lambda^1_x \rangle \nonumber \\
    & & \quad + \langle I_x|O'_x\rangle \langle \Lambda^1_x|\Lambda^0_x \rangle + \langle I_x|I'_x\rangle \langle \Lambda^1_x|\Lambda^1_x \rangle = 0 \\
    & \iff & \left(\begin{array}{cc} \langle O_x| & \langle I_x| \end{array} \right) \rho^T_\Lambda \otimes \id \left(\begin{array}{c} |O'_x \rangle \\ |I'_x \rangle \end{array} \right) = 0
  \end{eqnarray}
  with $\rho_\Lambda = \tr_{\mathcal H}|\Lambda \rrangle \llangle \Lambda|$. $\rho_\Lambda^\dagger = \rho_\Lambda$, $\rho_\Lambda \geq 0$ and $\det \rho_\Lambda \not=0$ (otherwise $|\Lambda \rrangle$ would be separable), $\langle \bullet | \rho^T \otimes \id \bullet \rangle_{\mathbb C^4}$ is then a scalar product of $\mathbb C^4$ (a definite positive sesquilinear form).  Two couples $(|O_x\rangle, |I_x\rangle)$ and $(|O'_x\rangle, |I_x'\rangle)$ (viewed as vectors of $\mathbb C^4$) define two orthogonal quasicoherent states if they are orthogonal with respect to $\langle \bullet | \rho^T \otimes \id \bullet \rangle_{\mathbb C^4}$. We can then find at most four independent quasicoherent states.
\end{proof}

In the following we consider for the whole of this paper only the case of eigenmanifolds at least weakly non-degenerate at all points (the case of a strongly non-degenerate eigenmanifold being not excluded).\\

If $\mathfrak M$ is not strongly non-degenerate at $x$ with separable quasienergy states, then these ones have the form $|\zeta \rangle \otimes |\Omega_x\rangle$ for some $|\Omega_x\rangle \in \mathcal H$ and for any $|\zeta \rangle \in \mathbb C^2$. It follows that any vector $\langle \zeta|\vec \sigma|\zeta \rangle$ is normal to $M_\Lambda$ at $x$. $x$ is then an isolated point ($\dim M_\Lambda = 0$ at $x$). Moreover $\slashed D_x|\zeta \rangle \otimes |\Omega_x\rangle=0$ for any $|\zeta\rangle$ implies that $|\Omega_x\rangle$ is a common eigenvector of $X^1$, $X^2$ and $X^3$. $M_\Lambda$ is then a discrete set of points if $\mathfrak X$ is abelian.\\

\begin{example}{3}{``Fuzzy'' circle}
  Let $\mathcal H = \mathbb C^N$ be endowed with an orthonormal basis $(|n\rangle)_{n=0,...,N-1}$, $Z = \sum_{n=0}^{N-1} |n\rangle \langle n+1|$ with by convention $|N\rangle \equiv |0\rangle$) and $X^3=0$. We can view this model as the Fuzzy unit circle since $Z^\dagger Z = (X^1)^2+(X^2)^2 = \id$. $[Z^\dagger,Z]=0$, $\mathfrak X$ is then abelian.
  \begin{equation}
    \left(\begin{array}{cc} 0 & Z^\dagger - \bar z \\ Z-z & 0 \end{array}\right)\left(\begin{array}{c} |\Lambda^0\rangle \\ |\Lambda^1\rangle \end{array} \right) \iff \left\{ \begin{array}{c} Z|\Lambda^1\rangle = z|\Lambda^1 \rangle \\ Z^\dagger|\Lambda^0\rangle = \bar z|\Lambda^0 \rangle \end{array} \right.
  \end{equation}
  Since $Z^N=\id$, $z$ must be a $N$-th root of the unity, and then $M_\Lambda = \{(\cos(\frac{2\pi q}{N}),\sin(\frac{2\pi q}{N}),0)\}_{q \in \{0,...,N-1\}}$, the quasicoherent states being weakly degenerate of the form
  \begin{equation}
    |\Lambda_\psi(e^{\imath \frac{2\pi q}{N}})\rrangle = |\psi \rangle \otimes \frac{1}{\sqrt N} \sum_{n=0}^{N-1} e^{\imath \frac{2\pi q}{N} n}|n\rangle
  \end{equation}
  for all $|\psi\rangle \in \mathbb C^2$. For finite values of $N$, the model is a Fuzzy regular cyclic $N$-gone with the circumscribed unit circle. For $N \to +\infty$, $M_\Lambda$ becomes the unit circle $S^1$ and the quasicoherent states tend in the weak topology to the singular distributions $|\Lambda_\psi(e^{\imath \varphi}) \rrangle = |\psi \rangle \otimes |\delta_\varphi \rangle$ with $\langle \theta|\delta_\varphi \rangle = \delta(\theta-\varphi)$ (Dirac distribution), the Hilbert space being assimilated to $\mathcal H = L^2(S^1,\frac{d\theta}{2\pi})$. Due to the abelianity of $\mathfrak X$ and to the zero quantum uncertainty, the system is not really fuzzy, it is more a matrix embedding of the classical circle.
\end{example}

If $\mathfrak M$ is not strongly non-degenerate at $x$ with entangled quasicoherent states, several situations can occur depending on the number of independent non-trivial purity normal vectors in $\Lin(\vec n_x^\alpha)_{\alpha =0,...,n \leq 3}$ with $\vec n_x^\alpha = \llangle \Lambda(x)| u_\alpha^\dagger \vec \sigma u_\alpha |\Lambda(x) \rrangle$ ($u_\alpha \in U(2)_{a.s.}$). This number can be smaller than $n = \dim \ker(\slashed D_x)$. Indeed, some quasicoherent states can be maximally entangled (Bell states), and their associated purity normal vectors are $\vec 0$. If $u_\alpha$ belongs to the isotropy group of $\rho_\Lambda$: $u^\dagger_\alpha \rho_\Lambda u_\alpha = \rho_\Lambda$, then $\vec n_x^\alpha = \vec n_x^0$. In a same way, if $u_\alpha$ belongs to the isotropy group of $u_\beta^\dagger \rho_\Lambda u_\beta$, then $\vec n_x^\alpha = \vec n_x^\beta$. We then have four possible situations:
\begin{itemize}
\item All purity normal vectors at $x$ are zero ($\ker \slashed D_x$ is a subspace of maximally entangled states). In this case we cannot conclude anything about the geometry of $M_\Lambda$ at $x$.
\item We have only one independent non-trivial purity normal vector, then this one provides the normal direction to $M_\Lambda$ at $x$ as for the strongly non-degenerate case.
\item We have only two independent non-trivial purity normal vectors, $M_\Lambda$ has then a normal plane at $x$, then $x$ is a pinch point of $M_\Lambda$ ($\dim M_\Lambda = 1$ at $x$).
\item We have at most three independent non-trivial purity normal vectors, the situation is similar to the case of a separable not strongly non-degenerate quasi-coherent state, $x$ is an isolated point ($\dim M_\Lambda=0$ at $x$).
\end{itemize} 

\subsection{Internal and external gauge choices}
The quasicoherent states are defined up to local $C^*$-phase changes: $\forall x \in M_\Lambda$
\begin{equation}
  |\tilde \Lambda(x) \rrangle = h(x) u(x) |\Lambda(x) \rrangle \qquad h(x)\in U(1),\ u(x)\in U_x
\end{equation}
$U_x$ being the subgroup of $SU(2)_{a.s.}$ such that $U_x \cdot |\Lambda(x)\rrangle = \ker(\slashed D_x)$ (the quasicoherent state being supposed weakly nondegenerate). Note that even if the quasicoherent state is strongly nondegenerate ($\dim \ker \slashed D_x = 1$), $U_x$ can be nontrivial as being the isotropy subgroup of $|\Lambda(x)\rrangle$ in $SU(2)$. $h$ corresponds to an usual phase change (it is the single phase change existing for a strongly nondegenerate quasienergy state), whereas $u$ is a local ``orientation'' change (a noncommutative phase change). We call this phase choice the internal gauge choice. We have also an external gauge choice associated with the coordinate changes (the change of the basis generating $\mathfrak X$):
\begin{equation}
  Y^i-y^i = {J^i}_j (X^j-x^j) \qquad J \in SO(3)
\end{equation}
$(Y^i)$ (resp. $(y^i)$) constitute new quantum (resp. classical) coordinate observables of $\mathfrak X$ (resp. $\mathbb R^3$). The coordinate changes must leave invariant the geometry operator by inducing a spin rotation:
\begin{eqnarray}
  \hat \slashed D_y & = & \hat \sigma_i \otimes (Y^i-y^i) \\
  & = & {J^i}_j \hat \sigma_i \otimes (X^j-x^j) \\
  & = & u_J^{-1} \hat \sigma_j u_J \otimes (X^j-x^j) \\
  & = & \sigma_j \otimes (X^j-x^j) \\
  & = & \slashed D_x
\end{eqnarray}
$\hat \sigma_j = u_J \sigma_j u_J^{-1}$ being the Pauli matrices after the spin rotation $u_J \in SU(2)$ associated with the 3D-vector rotation $J$. But consider $\check \slashed D_y = \sigma_i \otimes (Y^i-y^i) = u_J^{-1} \slashed D_{J^{-1}y} u_J$  the geometry observable in the new coordinates without spin rotation and $|\hat \Lambda(y) \rrangle$ an associated quasicoherent states ($\check \slashed D_y|\hat \Lambda(y)\rrangle = 0$). $|\hat \Lambda(y)\rrangle$ is related to $|\Lambda(x)\rrangle$ by
\begin{equation}\label{extchangeLambda}
  |\hat \Lambda(y) \rrangle = u_J^{-1} |\Lambda(J^{-1}y) \rrangle
\end{equation}

From the point of view of quantum gravity matrix model, $J$ constitutes a frame change of the observer (it is the ``observer'' which is rotated by $J^{-1}$ and not the quantum spacetime!). This is the reason for which we consider this gauge change as being ``external''.

\section{Metrics on a Fuzzy space}
To translate the notions of the usual geometry into the fuzzy space theory, we want to endow fuzzy spaces with metrics compatible with the quasicoherent geometry. For applications to quantum gravity, the goal of introducing these metrics is obvious (general relativity dealing with manifolds endowed with metrics). But we want also understand these metrics in the abstract geometric context and find their interpretations in the applications to quantum information theory. In this section we are interested with metrics associated with the quasicoherent picture which is different from the original Connes metric (for this one applied on fuzzy spaces we can see ref. \cite{Rieffel,Dandrea,Dandrea2,Dandrea3,Dandrea4}).

\subsection{The natural metrics of $\mathfrak M$ and of $M_\Lambda$}
$M_\Lambda$ is naturally endowed with the metric induced by its immersion in $\mathbb R^3$:
\begin{equation}
  \gamma_{ab} = \delta_{ij} \frac{\partial x^i}{\partial s^a} \frac{\partial x^j}{\partial s^b}
\end{equation}
where $(s^1,s^2)$ is a local curvilinear coordinate system on $M_\Lambda$.
\begin{prop}
The natural metric of $M_\Lambda$ can be written as $\forall x \in M_\Lambda$:
\begin{equation}
  \gamma_{ab}(x) = \frac{1}{2} \llangle \partial_{(a} \Lambda(x)|\slashed D^2_x|\partial_{b)}\Lambda(x) \rrangle
\end{equation}
\end{prop}

\begin{proof}
  $\slashed D_x|\Lambda \rrangle = 0 \Rightarrow -\sigma_i \frac{\partial x^i}{\partial s^a} |\Lambda \rrangle + \slashed D_x \partial_a|\Lambda \rrangle=0$. It follows that
  \begin{eqnarray}
    \llangle \partial_a \Lambda|\slashed D_x^2|\partial_b \Lambda \rrangle ds^a ds^b & = & \llangle \Lambda|\sigma_i \sigma_j|\Lambda \rrangle \frac{\partial x^i}{\partial s^a} \frac{\partial x^j}{\partial s^b}ds^a ds^b \\
    & = & \delta_{ij} \frac{\partial x^i}{\partial s^a} \frac{\partial x^j}{\partial s^b}ds^a ds^b \nonumber \\
    & & \quad + \imath \underbrace{{\varepsilon_{ij}}^k \llangle \Lambda|\sigma_k|\Lambda \rrangle\frac{\partial x^i}{\partial s^a} \frac{\partial x^j}{\partial s^b}ds^a ds^b}_{=0} \\
    & = & \gamma_{ab} ds^a ds^b
  \end{eqnarray}
\end{proof}

The triads of $M_\Lambda$ $(e^i_a)$ are $e^i_a = \frac{\partial x^i}{\partial s^a}$ in order to have $\gamma_{ab} = \delta_{ij} e^i_a e^i_b$.\\
To be consistent with the quasicoherent metric, the metric of $\mathfrak M$ as a square length observable is defined by $\pmb{d\ell^2}_x = \slashed D_x^2$ (and not $\slashed D_x^{-2}$ as in usual noncommutative geometry \cite{Connes}, this is consistent with the fact that $\slashed D_x$ and $-\imath [\slashed D,\bullet] = \sigma_i \otimes L_{X^i}$ are homogeneous to a length and not to the inverse of a length). The definition of $\gamma$ can be then rewritten as
\begin{equation}
  {d\ell^2}_s = \gamma_{ab}(x(s))ds^ads^b = \omega_{x(s+ds)}(\pmb{d\ell^2}_{x(s)})
\end{equation}
We recall that the geometric operator $\slashed D_x$ can be interpreted as an energy observable (displacement energy in string theory, or energy of the qubit dressed by the environment in quantum information theory). By construction, the mean energy in a quasi-coherent state is zero with no uncertainty: $\omega_x(\slashed D_x)= 0$ and $\Delta_x \slashed D_x^2 = \omega_x(\slashed D_x^2) - \omega_x(\slashed D_x)^2=0$. But consider a ``slipping'' on $M_\Lambda$, for which the state of the quantum system is $\omega_{x(s+ds)}$ but the energy measurement is preformed with $\slashed D_{x(s)}$ (the measurement is performed with an infinitesimal ``misalignment'' $ds$). The mean energy is then $\omega_{x(s+ds)}(\slashed D_{x(s)}) = 0$ (at the first order of approximation), but the square deviation is $\Delta_{x(s+ds)} \slashed D_{x(s)}^2 = \omega_{x(s+ds)}(\slashed D_{x(s)}^2) - \omega_{x(s+ds)}(\slashed D_{x(s)})^2 = {d\ell^2}_s$. ${d\ell}_s$ is then the energy uncertainty. $\ell(\mathscr C) = \int_{\mathscr C} d\ell$ is the cumulated energy uncertainty during a slipping along the path $\mathscr C$, and so the geodesics of $\gamma$ are the path minimising the energy uncertainty due to the infinitesimal misalignment of the probe. \\
$\pmb{d\ell^2}_x$ as square length observable induces a quantisation of the measure of the infinitesimal lengths. In contrast with a classical manifold for which we have a continuum of infinitesimal lengths, onto a fuzzy space these ones belong to $\Sp(\sqrt{\pmb{d\ell^2}_x}) = \Sp(|\slashed D_x|)$. For example, for the fuzzy plane $\Sp(|\slashed D_x|) = \{\sqrt n, n \in \mathbb N\}$ which can be interpreted in quantum gravity by a quantisation of the infinitesimal lengths which are multiples of the Planck length (with factors being square roots of integers).\\ 

To continue the analysis of the relation between the geometry of $\mathfrak M$ and the one of $M_\Lambda$, we need maps permitting to transform the (co)tangent vectors of one into those of the other. Let $\omega_{x*}$ and $\omega_x^*$ be the push-forward (tangent map) and the pull-back (cotangent map) of $\omega_x$ (for $x \in M_\Lambda$)
\begin{equation}
  \begin{array}{rcl} \omega_{x*} : \Der(\mathfrak X) & \to & T_x\mathbb R^3 = T_x M_\Lambda \oplus N_x M_\Lambda \\ L & \mapsto & \omega_x(L(X^i)) \partial_i \end{array}
\end{equation}
$T_x M_\Lambda$ and $N_x M_\Lambda$ being the tangent and the normal vector spaces at $x$ on $M_\Lambda$.
\begin{equation}
  \begin{array}{rcl} \omega_x^* : \Omega^1_x \mathbb R^3 & \to & \Omega^1_{\Der} \mathfrak X \\ \eta_i(x) dx^i & \mapsto & \eta_i(x) \dnc X^i \end{array}
\end{equation}
These maps are defined to satisfy the following duality consistency condition: $\langle \eta(x),\omega_{x*}(L) \rangle_{\mathbb R^3} = \omega_x(\langle \omega_x^* \eta(x), L \rangle_{\mathfrak X})$ ($\langle \bullet , \bullet \rangle_{\mathbb R^3}$ standing for the duality bracket between the tangent vectors and the differential 1-forms of $\mathbb R^3$). Moreover let $\pi_x$ be the orthogonal projection onto $T_xM_\Lambda$:
\begin{equation}
  \begin{array}{rcl} \pi_x: T_x\mathbb R^3 & \to & T_x M_\Lambda \\ v^i \partial_i & \mapsto & v^i \delta_{ij} \gamma^{ab} \frac{\partial x^j}{\partial s^b} \partial_a \end{array}
\end{equation}
and $\pi_x^*$ be its dual map ($\langle \pi_x^* \bullet,\bullet\rangle_{\mathbb R^3} = \langle \bullet, \pi_x \bullet\rangle_{\mathbb R^3}$):
\begin{equation}
  \begin{array}{rcl} \pi_x^* : \Omega^1_x M_\Lambda & \to & \Omega^1_x \mathbb R^3 \\ \eta_a ds^a & \mapsto & \eta_a \gamma^{ab} \delta_{ij} \frac{\partial x^j}{\partial s^b} dx^i \end{array}
\end{equation}
$\omega_x^* \pi_x^* (\gamma_{ab}(x)ds^a ds^b) \in (\Omega^1_{\Der} \mathfrak X)^{\otimes 2}$ defines then a metric on $\mathfrak M$ as a symmetric bicovector.
\begin{prop}
  The metric of a fuzzy space $\mathfrak M$ as a bicovector is
  \begin{equation}
    \pmb \gamma = \omega_x^* \pi_x^* (\gamma_{ab}(x)ds^a ds^b) = \delta_{ij} \dnc X^i \otimes \dnc X^j
  \end{equation}
\end{prop}

\begin{proof}
  \begin{eqnarray}
    \omega_x^* \pi_x^* (\gamma_{ab}(x)ds^a ds^b) & = & \gamma_{ab} \gamma^{ac}\gamma^{bd} \delta_{ik} \delta_{jl} \frac{\partial x^k}{\partial s^c} \frac{\partial x^l}{\partial s^d} \dnc X^i \otimes \dnc X^j \\ & = &  \gamma^{bd} \delta_{ik} \delta_{jl} \frac{\partial x^k}{\partial s^b} \frac{\partial x^l}{\partial s^d} \dnc X^i \otimes \dnc X^j \\
    & = & \delta_{ij} \dnc X^i \otimes \dnc X^j
  \end{eqnarray}
because $\gamma_{bd} = \frac{\partial x^k}{\partial s^b} \frac{\partial x^l}{\partial s^d} \delta_{kl}$. 
\end{proof}
We have also $\pmb \gamma = \delta_{ij} \Theta^{ki} \Theta^{lj} \theta_k \otimes \theta_l$. $\pmb \gamma$ is well defined as metric of $\mathfrak M$ because it is independent of $x$ (of the probe).
\begin{prop}
  $\pmb \gamma: \Der(\mathfrak X)^2 \to \Env(\mathfrak X)$ is hermitian and positive (and non-degenerate if $\mathcal Z(\mathfrak X)$ is trivial).
\end{prop}
\begin{proof}
  Let $L_V = v_i L_{X^i}, L_W=w_i L_{X^i} \in \Der(\mathfrak X)$ be two tangent vectors of $\mathfrak M$ (with $v_i,w_i \in \mathcal Z(\Env(\mathfrak X))$ selfadjoint).
  \begin{eqnarray}
    \pmb \gamma(L_V,L_W) & = & \delta_{ij} \Theta^{ki} \Theta^{lj} v_k w_l \\
    & = & -\delta_{ij} [V,X^i][W,X^j] \\
    & = & -\delta_{ij} ([X^j,W][X^i,V])^\dagger \\
    & = & -\delta_{ij} ([W,X^j][V,X^i])^\dagger \\
    & = & \pmb \gamma(L_W,L_V)^\dagger
  \end{eqnarray}
  In particular $\pmb \gamma(L_V,L_V)$ is selfadjoint and has then a real spectrum. Let $\psi \in \mathcal H$.
  \begin{eqnarray}
    \langle \psi|\pmb \gamma(L_V,L_V)|\psi \rangle & = & - \delta_{ij} \langle \psi|[V,X^i][V,X^j]\psi \rangle \\
    & = & - \delta_{ij} \langle [X^i,V] \psi | [V,X^j] \psi \rangle \\
    & = & \delta_{ij} \langle [V,X^i] \psi | [V,X^j] \psi \rangle \\
    & = & \sum_i \|[V,X^i]\psi\|^2 \\
    & \geq & 0
  \end{eqnarray}
    Moreover $\langle \psi|\pmb \gamma(L_V,L_V)|\psi \rangle = 0$ for all $\psi$ if and only if $V=0$ except if $V$ commutes with all $X^i$. If the center of $\mathfrak X$ is trivial ($\mathcal Z(\mathfrak X)=\Lin(\id)$), this cannot arise.
\end{proof}

At the semi-classical limit where $\Theta^{ij} = -\imath [X^i,X^j] \sim \vartheta^{ij}$ defines a Poisson structure of $\mathbb R^3$ ($\{f,g\} = \vartheta^{ij} \partial_i f \partial_j g$, $\forall f,g \in \underline{\mathbb R}_{\mathbb R^3}$), $\pmb \gamma$ defines an effective metric $G^{kl} = \delta_{ij} \vartheta^{ki} \vartheta^{lj}$ which plays the role in quantum gravity matrix models of the metric of the emergent curved spacetime \cite{Klammer, Steinacker2} (the spacetime is flat and noncommutative at the microscopic scale but it emerges a curved commutative spacetime at the macroscopic scale with the semi-classical limit). In this meaning, $\pmb \gamma$ can be then viewed as the quantum metric of a quantum spacetime.\\

Since $\pmb \gamma_{kl} = \delta_{ij} {\Theta_k}^i {\Theta_l}^j$ (with respect to the representation with $\theta_i$), we can set $E^i_k = {\Theta_k}^i = [X_k,X^i]$ as triads of $\mathfrak M$. The analogue of the Weitzenb\"ock connection \cite{Aldrovandi} of $\mathfrak M$ is then
\begin{equation}
  \mathring \Gamma^i_{jk} =  L_{X^k}(E^i_j) = \imath [X_k,[X^i,X_j]]
\end{equation}
which defines an analogue of the Weitzenb\"ock torsion:
\begin{eqnarray}
  \mathring T^k_{ij} & = & L_{X^j}(E^k_i) - L_{X^i}(E^k_j) \\
  & = & \imath [X_j,[X^k,X_i]] - \imath [X_i,[X^k,X_j]] \\
  & = & \imath [X^k,[X_i,X_j]]
\end{eqnarray}
(where we have used the Jacobi identity).\\

Remark : $E^j_i$ induces the following dual diads on $M_\Lambda$:
\begin{eqnarray}
  \mathring e^a_i & = & \omega_x({[\pi_x]^a}_j E^j_i) \\
  & = & -\imath \gamma^{ab} \frac{\partial x^j}{\partial s^b} \llangle \Lambda(x)|[X_i,X_j]|\Lambda(x) \rrangle
\end{eqnarray}
the triads being by construction $e^i_a = \frac{\partial x^i}{\partial s^a}$.\\

In BFFS matrix model, $M_\Lambda$ is the classical space closest as possible to the physics of the quantum space $\mathfrak M$. $\mathfrak M$ is the quantisation of a flat space (as we can see it with its metric $\pmb \gamma = \delta_{ij} \dnc X^i \otimes \dnc X^j$). In this kind of models, the gravitational effects emerge at the macroscopic level from the noncommutativity at the microscopic level (with the thermodynamical limit where the number of strings tends to infinity with constant density). At the microscopic level, the gravitational effects emerge at the quasicoherent picture (which is the classical geometry closest to the quantum geometry) \cite{Viennot2}. These gravitational effects are induced, in accordance with general relativity, by the curvature of $M_\Lambda$ with its metric $\gamma_{ab}$.\\
For the applications to quantum information, $M_\Lambda$ is the constrained control parameter manifold of the qubit. With adiabatic control in the whole of $\mathbb R^3$, the adiabatic transport induces in addition to the geometric phase a dynamical phase which, in contrast with the geometric phase, cannot be fixed only by the control design (the shape of paths $\mathscr C$ in $M_\Lambda$ \cite{Shapere, Bohm}). In strong adiabatic transport, these dynamical phases modify the interferences, and in weak adiabatic transport, these operator valued dynamical phases modify the reached density matrix as a kind of effective Hamiltonian evolution in addition to the geometric effects \cite{Viennot1, Viennot3}. To have a pure geometric control, the adiabatic control must be constrained onto $M_\Lambda$. The geometry of this manifold is then characteristic of the controllability of the qubit in presence of the environment.

\subsection{Use of the noncommutative metric}
To use in practice the metric $\pmb \gamma$ we need to define the equivalents of paths on $\mathfrak M$. Due to the non-separability of the quasicoherent states and of  Heisenberg uncertainties concerning the coordinate observables, the usual notion of path cannot be obviously generalised. To solve this problem, we introduce (in accordance with the Perelomov coherent state theory \cite{Perelomov}) the notion of displacement operator. 

\subsubsection{The displacement operator}
\begin{defi}[Displacement operator]
  Let $M_\Lambda$ be the eigenmanifold of a fuzzy space $\mathfrak M$, supposed weakly non-degenerate. Let $x,y \in M_\Lambda$ be two distinct points. We call displacement operator $\Dis(y,x) \in e^{\imath \mathfrak X} \subset \Env(\mathfrak X)$ from $x$ to $y$, the invertible operator, if it exists, such that
  \begin{equation}
    \exists u_{yx} \in U(2)_{a.s.}, \quad u_{yx} \otimes \Dis(y,x)|\Lambda(x)\rrangle = |\Lambda(y) \rrangle
  \end{equation}
  and then
  \begin{equation}
    \omega_y = \omega_x \circ \Ad_{\Dis(y,x)}
  \end{equation}
  If $\Dis(y,x)$ exists, $x$ and $y$ are said linkable.
\end{defi}
Up to a local orientation and local phase changes, $\Dis(y,x)$ induces the move from the pseudo-point $\omega_x$ to the pseudo-point $\omega_y$. A similar definition of the displacement operator exists for the Perelomov coherent states (by definition the coherent states are the orbit of the vacuum state by the action of the group of displacement operators \cite{Perelomov, Puri}).\\
In the context of quantum information theory, the displacement operator models a ``control jump'', i.e. a local operation quantum protocol on the qubit ($u_{yx}$) and on its environment ($\Dis(y,x)$) permitting to abruptly change the control parameters from $x$ to $y$.\\

The displacement operator associated with $|\hat \Lambda \rrangle$ after an external gauge change eq. (\ref{extchangeLambda}) is related to the one associated with $|\Lambda \rrangle$ by
\begin{eqnarray}
    \hat \Dis(y_2,y_1) & = & \Dis(J^{-1}y_2,J^{-1}y_1) \\
    & = & \Dis(y_2,y_1) \mathbf{G}_J(y_2,y_1)
\end{eqnarray}
with $\mathbf G_J(y_2,y_1) = \Dis(y_2,y_1)^{-1}\Dis(J^{-1}y_2,J^{-1}y_1) \in \Env(\mathfrak X)$.\\

\begin{example}{2}{Fuzzy surface plots}
  The Perelomov coherent states of the harmonic oscillator algebra are defined by \cite{Perelomov}
  \begin{equation}
    \mathcal D(\alpha)|0 \rangle = |\alpha \rangle \qquad \text{with }\mathcal D(\alpha)=e^{\alpha a^+-\bar \alpha a} = e^{-|\alpha|^2/2} e^{\alpha a^+} e^{-\bar \alpha a}
  \end{equation}
  and since $\mathcal D(\alpha)a\mathcal D(\alpha)^\dagger = a-\alpha$, we have also $|n\rangle_\alpha = \mathcal D(\alpha)|n\rangle$. Moreover we have
  \begin{equation}
    \mathcal D(\alpha) \mathcal D(\beta) = e^{\imath \Im(\alpha \bar \beta)} \mathcal D(\alpha+\beta)
  \end{equation}
  and then
  \begin{equation}
    \mathcal D(\beta-\alpha)|\alpha \rangle = e^{-\imath \Im (\alpha \bar\beta)} |\beta \rangle
  \end{equation}
  This defines a displacement operator for the fuzzy spaces of the fuzzy plane family by
  \begin{equation}
    \Dis(\beta,\alpha) = \mathcal D(\beta-\alpha) = e^{(\beta-\alpha) a^+-(\bar \beta-\bar \alpha) a}
  \end{equation}
  The local gauge change associated with the displacement operator, $u_{\beta\alpha} \otimes \Dis(\beta,\alpha)$, are
  \begin{equation}
    u_{\beta \alpha} = \left\{ \begin{array}{ll}
      e^{\imath  \Im (\alpha \bar\beta)} & \text{\tiny for the fuzzy plane} \\
      e^{\imath  \Im (\alpha \bar\beta)} \left(\begin{array}{cc} 1 & 0 \\ \epsilon(\alpha-\beta) & 1 \end{array} \right) & \text{\tiny for the fuzzy elliptic paraboloid} \\
      e^{\imath  \Im (\alpha \bar\beta)} \left(\begin{array}{cc} 1 & 0 \\ \epsilon(\bar \alpha-\bar \beta) & 1 \end{array} \right) & \text{\tiny for the fuzzy hyperbolic paraboloid}
    \end{array} \right.
  \end{equation}  

  The fuzzy hyperboloid is totally unlinkable and no displacement operator exists. Nevertheless, the fuzzy hyperboloid is ``paralinkable'' and it is possible to generalise the discussion for it, see \ref{paralinkable}.
\end{example}

\begin{prop}
  If $(\mathfrak X,\mathcal H)$ is a selfadjoint irreducible representation $\iota$ of a Lie algebra $\mathfrak g$ (involving that $\mathfrak X$ is stable by the commutator $[\bullet,\bullet]$), and if $|\Lambda(x)\rrangle$ and $|\Lambda(y)\rrangle$ are separable, then it exists a displacement operator $\Dis(y,x)$.
\end{prop}
\begin{proof}
  By an abusive notation we denote also by $\iota$ the induced unitary irreducible representation of the Lie group $G$ of $\mathfrak g$ onto $\mathcal H$.
  \begin{lemma}
    Let $\iota$ be an unitary irreducible representation of the a Lie group $G$ onto $\mathcal H$. $\forall \psi,\phi \in \mathcal H$, $\|\psi\|=\|\phi\|=1$, $\exists g \in G$ such that $|\psi \rangle = \iota(g)|\phi \rangle$
  \end{lemma}
  By \textit{reductio ad absurdum} we suppose that $\exists \psi,\phi \in \mathcal H \setminus\{0\}$ with $|\psi \rangle \not\in \iota(G)\Lin|\phi \rangle$. By construction $\iota(G)\Lin|\phi \rangle$ is stable by $G$. $\iota$ being irreducible, this implies that $\iota(G)\Lin|\phi \rangle = \{0\}$ or $\iota(G)\Lin|\phi \rangle = \mathcal H$. But $\iota(G)\Lin|\phi \rangle = \mathcal H$ is in contradiction with $|\psi \rangle \not\in \iota(G)\Lin|\phi \rangle$, and $\iota(G)\Lin|\phi \rangle = \{0\}$ implies that $\phi = 0$ in contradiction with $\phi$ normalisable.\\

  $|\Lambda(x) \rrangle$ being separable it can be written as $|\Lambda(x) \rrangle = |O(x) \rangle \otimes |\Omega(x)\rangle$ with $|O(x)\rangle \in \mathbb C^2$ and $|\Omega(x)\rangle \in \mathcal H$. Two normalised vectors of $\mathbb C^2$ can be related by an unitary transformation: $|O(y)\rangle = u_{yx}|O(x)\rangle$ ($\mathbb C^2$ is an unitary irreducible representation of $U(2)$). Since $(\mathfrak X,\mathcal H)$ is an unitary irreducible representation of $G$, $\exists g \in G$ such that $|\Omega(y)\rangle = \iota(g)|\Omega(x)\rangle$. Let $(\mathcal X^\alpha, \id)$ be a set of generators of $\mathfrak g$, and let $\xi^i_\alpha \in \mathbb R$ be such that $X^i = \xi^i_\alpha \iota(\mathcal X^\alpha)$ (and $\iota(\id)=\id$). We have then $\exists \mu_\alpha \in \mathbb R$ such that $\iota(g) = e^{\imath \varphi} e^{\imath \mu_\alpha \iota(\mathcal X^\alpha)} =  e^{\imath \varphi} e^{\imath \mu_\alpha \xi^\alpha_i X^i} \equiv \Dis(y,x) \in e^{\imath \mathfrak X}$.
\end{proof}

\begin{example}{1}{Fuzzy sphere}
  The Fuzzy sphere being defined by the $(2j+1)$-dimensional unitary irreducible representation of $\mathfrak{su}(2)$ and its quasicoherent states being separable, it satisfies this property. More precisely let $\mathcal D_j(x) = e^{\imath \theta \vec m \cdot \vec J}$ with $x=rj(\sin \theta \cos \varphi,\sin \theta \sin \varphi,\cos \theta)$ and $\vec m = (\sin \varphi,-\cos \varphi,0)$. The quasicoherent state of the Fuzzy sphere is then 
  \begin{equation}
    |\zeta \rangle_{1/2}\otimes|\zeta\rangle_j = \mathcal D_{1/2}(x) \otimes \mathcal D_j(x) |1/2,1/2\rangle \otimes |j,j\rangle
  \end{equation}
$\zeta = e^{\imath \varphi} \tan \frac{\theta}{2}$. We have \cite{Perelomov}
\begin{equation}
  \mathcal D_j(x') \mathcal D_j(x) = \mathcal D_j(x'') e^{\imath \Phi(x',x) J^3}
\end{equation}
with $x''=rj(\sin (\theta+\theta') \cos (\varphi+\varphi'),\sin (\theta+\theta') \sin (\varphi+\varphi'),\cos (\theta+\theta'))$ and $\Phi(x',x)$ is the solid angle of the geodesic triangle (at the surface of the sphere) linking the north pole to $x$ and $x'$. It follows that we have the following displacement operator for the Fuzzy sphere:
\begin{eqnarray}
  u_{x'x} \otimes \Dis(x',x) & = & \mathcal D_{1/2}(x')e^{-\frac{\imath}{2} \Phi(\tilde x,x)\sigma_3} \mathcal D_{1/2}(x')^\dagger \mathcal D_{1/2}(\tilde x) \nonumber \\
  & & \quad \otimes \mathcal D_j(x')e^{-\imath \Phi(\tilde x,x)J^3} \mathcal D_j(x')^\dagger \mathcal D_j(\tilde x)
  \end{eqnarray}
with $\tilde x = rj(\sin (\theta'-\theta) \cos (\varphi'-\varphi),\sin (\theta'-\theta) \sin (\varphi'-\varphi),\cos (\theta'-\theta))$.
\end{example}

Note that in general situations, the existence of a displacement operator is not ensured.\\

We can see the displacement operators as replacing the notion of smooth paths on a classical manifold (due to the quantum inseparability of the pures states, the usual notion of path has no meaning). More precisely, the path category of a classical manifold $\mathscr PM$ ($\Obj \mathscr PM = M$ and $\Morph \mathscr PM$ is the set of directed smooth paths on $M$ with the source and the target maps returning the end points) is replaced by the category $\mathscr P \mathfrak X$ with $\Obj \mathscr P \mathfrak X = \mathcal P(\mathfrak M)$, $\Morph \mathscr P \mathfrak X = \InnAut\Env(\mathfrak X) \times \mathcal P(\mathfrak M)$ ($\InnAut\Env(\mathfrak X)$ denoting the set of inner automorphisms of $\Env(\mathfrak X)$) with source, target and identity maps defined by: $s(a,\omega) = \omega$, $t(a,\omega) = \omega \circ a$, $\id_\omega = (\id,\omega)$ and arrow composition defined by $(a',\omega \circ a) \circ (a,\omega) = (a \circ a',\omega)$.\\

Suppose that we proceed to a measurement with the probe at $x$, permitting to prepare $\mathfrak M$ in the quasicoherent state $|\Lambda(x)\rrangle$. After that, we abruptly proceed to a second measurement with the probe at $y \not= x$ (in contrast with classical geometry, we cannot assume in quantum physics a smooth continuous measurement from $x$ to $y$). This second measurement has a probability $|\llangle \Lambda(y)|\Lambda(x)\rrangle|^2 \not= 0$ to project $\mathfrak M$ in $|\Lambda(y) \rrangle$. In this case, the evolution induced by the two measurements is defined by $u_{yx} \otimes \Dis(y,x)$. We can then well view $\Dis(y,x)$ as a ``quantum path'' from $\omega_x$ to $\omega_y$.\\

Moreover, the group of the automorphisms of the $C^*$-algebra $\Env(\mathfrak X)$ plays the role of the diffeomorphism group: $\Diff_{n.c.}\mathfrak M \sim \Aut\Env(\mathfrak X)$. This suggests a relation between the displacement operators and the diffeomorphisms of $M_\Lambda$. Let $\mathscr M_\Lambda$ be the category such that $\Obj(\mathscr M_\Lambda) = M_\Lambda$, $\Morph(\mathscr M_\Lambda) = \Diff M_\Lambda \times M_\Lambda$, $s(\varphi,x) = x$, $t(\varphi,x)=\varphi(x)$, $\id_x = (\id,x)$ and $(\varphi',\varphi(x)) \circ (\varphi,x) = (\varphi' \circ \varphi, x)$. We call $\mathscr M_\Lambda$ the eigen category of $\mathfrak M$. $\mathscr M_\Lambda$ and $\mathscr P\mathfrak X$ are related by the functor $\pmb \omega$:
\begin{equation}
  \left\{\begin{array}{rcl} \pmb \omega (x) & = & \omega_x \\ \pmb \omega(\varphi,x) & = & (\Ad_{\Dis(\varphi(x),x)},\omega_x) \end{array} \right.
\end{equation}

The categorical aspects related to the noncommutative geometry are studied in details \ref{category}, especially \ref{NCvsCat} where we compare the categorical and the noncommutative geometries.\\

For the sake of the simplicity, we treat only the case of totally linkable fuzzy spaces. This assumption can be considered as too restrictive, but it is pertinent since it implies that the magnitude of the quantum entanglement be constant on $M_\Lambda$ (and so corresponds to fuzzy spaces homogeneous from the point of view of the quantum information properties). In \ref{paralinkable} we generalise the concept of displacement operator in order to treat fuzzy spaces without this assumption (but with a weaker assumption of homogeneity consisting to have a single class of entanglement (entangled or separable) on the whole of $M_\Lambda$ but not necessary with a constant magnitude of the entanglement).\\
Linkability is a noncommutative equivalent of the classical notion of connectivity. A classical manifold $M_\Lambda$ is connected if $\forall x,y \in M_\Lambda$, it exists a smooth continuous path from $x$ to $y$. A fuzzy space $\mathfrak M$ is totally linkable if $\forall x,y \in M_\Lambda$ (eigenmanifold of $\mathfrak M$), it exists a ``quantum path'' $\Dis(y,x)$ from $\omega_x$ to $\omega_y$. But note that $\mathfrak M$ linkable $\nLeftrightarrow$ $M_\Lambda$ connected as we will see in the next with examples.

\subsubsection{The linking vector observable}
By analogy with the commutative case, we can define the length of a path $(a,\omega_P) \in \InnAut\Env(\mathfrak X) \times \mathcal P(\mathfrak X)$ if $a=e^{v_i L_{X^i}} = \Ad_{e^{-\imath v_i X^i}}$ by
\begin{eqnarray}
  \pmb \ell(\Ad_{e^{-\imath V}},\omega_P) & = & \omega_P(\sqrt{\pmb \gamma(L_V,L_V)}) \\
  & = & \tr\left(P \sqrt{-\delta_{ij}[V,X^i][V,X^j]}\right)
\end{eqnarray}
$V \in \Der(\mathfrak X)$ being the tangent vector at $(Ad_{e^{-\imath V}},\omega_P)$, the formula is a noncommutative equivalent to $\ell(\mathscr C) = \int_{\mathscr C} \sqrt{\gamma(\partial_s,\partial_s)}ds$ with $\partial_s = \frac{\partial x^i}{\partial s} \partial_i \in T_{x(s)}\mathscr C$.\\

If $\Dis(y,x) = e^{\imath \mu_i(y,x) X^i} \in e^{\imath \mathfrak X}$, $\pmb \gamma$ induces the following distance between two quasicoherent states:
\begin{eqnarray}
  \dist_{\pmb \gamma}(\omega_x,\omega_y) &  = & \pmb \ell(\Ad_{\Dis(y,x)},\omega_x) \\
  & = & \llangle \Lambda(x)|\sqrt{-\delta_{ij} [\mu_kX^k,X^i][\mu_l X^l,X^j]}|\Lambda(x)\rrangle
\end{eqnarray}

To interpret $\pmb \gamma$, consider two linkable infinitely close points $\vec x$ and $\vec y=\vec x+\delta \vec x$. The probe is moved from $x$ to $y$, measurements being performed at these two end points but not between them. In contrast with classical geometry where we can assume a continuously performed measurement along a smooth path, in a quantum context, due to the Born projection rule, this cannot be possible. We must then only consider two end points where measurements are performed, the ``path'' between them being not physically defined (since no measure occurs). This is the reason for which the equivalent of a smooth path in $\mathfrak M$ is an operator depending only on two end points (or is a succession of this kind of operators). The distance between the two points in $\mathbb R^3$ is then $\|\delta \vec x\|$ and along $M_\Lambda$ is $\dist_\gamma(x,y)=\gamma_{ab}(x)\delta x^a \delta x^b$. But $\vec x = \omega_x(\vec X) = \llangle \Lambda(x)|\vec X|\Lambda(x)\rrangle$ is only the mean value of the starting point, which is subject to standard deviation $\Delta_x X^i = \sqrt{\omega_x([X^i]^2)-[\omega_x(X^i)]^2}$. In the reality the measurement of the starting point obeys to a random process, and the coordinates of the starting point are random variables described by $\vec X$ for a probability law described by $\omega_x$. In the same manner, the coordinates of the final point are  random variables described also by $\vec X$ with the probability law $\omega_y$. Since $|\Lambda(y)\rrangle = u_{yx} \otimes \Dis(y,x)|\Lambda(x)\rrangle$ and $\omega_y(\vec X) = \omega_x(\Dis(y,x)^\dagger \vec X \Dis(y,x))$ we can state that the coordinates of the final point are random variables described by $\Dis(y,x)^\dagger\vec X \Dis(y,x)$ with the probability law $\omega_x$. The vector linking the two end points is then a random variable of probability law $\omega_x$ described by the following observable:
\begin{eqnarray}
  \pmb{\delta \vec \ell}_{yx} & = & \Dis(y,x)^\dagger \vec X \Dis(y,x) - \vec X \label{linkvect}\\
  & = & -\imath [\delta \Pi_x, \vec X] + \mathcal O(\|\delta \vec x\|^2)
\end{eqnarray}
where we have set $\Dis(y,x) = \Dis(x+\delta x,x) = e^{\imath \delta \Pi_x}$. And so, the distance between the two measured end points is a random variable of probability low $\omega_x$ described by the following observable:
\begin{eqnarray}
  |\pmb{\delta \ell}_{yx}| & = & \sqrt{(\pmb{\delta \vec  \ell}_{yx})^2} \\
  & = & \sqrt{-\delta_{ij}[\delta \Pi_x,X^i][\delta \Pi_x,X^j]} \\
  & = & \sqrt{\pmb \gamma(L_{\delta \Pi_x},L_{\delta \Pi_x})}
\end{eqnarray}
$\omega_x(|\pmb{\delta \ell}_{yx}|) = \dist_{\pmb \gamma}(\omega_x,\omega_y)$ is the mean value of the distance observable.\\

We can note that $\Dis(y,x)^\dagger \slashed D_x \Dis(y,x) - \slashed D_x = \sigma_i \otimes \pmb{\delta \ell}^i_{yx}$. It follows that
\begin{equation}
  (\pmb{\delta \vec \ell}_{yx})^2 = \frac{1}{2} \tr_{\mathbb C^2} (\Dis(y,x)^\dagger \slashed D_x \Dis(y,x) - \slashed D_x)^2
\end{equation}
Let $\Delta\slashed D_{yx} = \Dis(y,x)^\dagger \slashed D_x \Dis(y,x) - \slashed D_x$ be the energy gap observable between $x$ and $y$:
\begin{equation}
  \omega_x(\Delta\slashed D_{x+\delta x,x}) = \omega_{x+\delta x}(\slashed D_x)-\omega_x(\slashed D_x) = \omega_{x+\delta x}(\slashed D_x)
\end{equation}
we have
\begin{equation}
  \dist_{\pmb \gamma}(\omega_x,\omega_{x+\delta x}) = \omega_x\left(\sqrt{\frac{1}{2} \tr_{\mathbb C^2} (\Delta\slashed D^2_{x+\delta x,x})}\right)
\end{equation}
$\frac{1}{2} \tr_{\mathbb C^2} (\Delta\slashed D^2_{x+\delta x,x})$ is the average square energy gap observable for the microcanonical distribution $\frac{1}{2}\id$ of the spin degree of freedom. Since $\tr_{\mathbb C^2} \Delta\slashed D_{x+\delta x,x}=0$, $\sqrt{\frac{1}{2} \tr_{\mathbb C^2} (\Delta\slashed D^2_{x+\delta x,x})}$ is the uncertainty of the energy gap for the microcanonical distribution, and so $\dist_{\pmb \gamma}(\omega_x,\omega_{x+\delta x})$ is the mean value of the  energy gap ``isotropic'' uncertainty (``isotropic'' since the uncertainty is computed for the equiprobable direction state distribution). A discrete path $\mathscr C = \{x_0,x_1,...,x_N\}$ on $M_\Lambda$ such that $\omega_{x_n}(|\pmb{\delta \ell}_{x_{n+1},x_n}|)$ is a small constant value minimises then the quantity $\sum_{n=0}^{N-1} \omega_{x_n}\left(\sqrt{\frac{1}{2} \tr_{\mathbb C^2} (\Delta\slashed D^2_{x_{n+1},x_n})}\right)$.\\

\begin{example}{3}{``Fuzzy'' circle}
  The definition of the noncommutative metric can be extended to fuzzy spaces with disconnected eigenmanifold $M_\Lambda$. We can consider for example the case of the perturbed fuzzy circle with $Z = \sum_{n=0}^{N-1} |n\rangle \langle n+1|$ and $X^3 = \epsilon \hat N$ with $\hat N = \sum_{n=0}^{N-1} n|n\rangle \langle n|$ and $\epsilon \simeq 0$. By applying the perturbation theory for a separable weakly degenerate quasicoherent state (\ref{perturbWD}), we have for the vertical deformation:
  \begin{eqnarray}
    \delta x^3(q) & = & \epsilon \llangle \Lambda_\psi(e^{\imath \frac{2\pi q}{N}})|\hat N|\Lambda_\psi(e^{\imath \frac{2\pi q}{N}}) \rrangle \\
    & = & \frac{\epsilon}{N} \sum_{n=0}^{N-1} n \\
    & = & \epsilon \frac{N-1}{2}
  \end{eqnarray}
  The perturbed fuzzy circle is just translated in the $x^3$-direction. Since $|\Lambda_\psi(e^{\imath \frac{2\pi q}{N}})\rrangle = |\psi \rangle \otimes \frac{1}{\sqrt N} \sum_{n=0}^{N-1} e^{\imath \frac{2\pi q}{N} n}|n\rangle$ we have (at the zero order of perturbation)
  \begin{equation}
    |\Lambda_\psi(e^{\imath \frac{2\pi (q+p)}{N}})\rrangle = \Dis^p |\Lambda_\psi(e^{\imath \frac{2\pi q}{N}})\rrangle
  \end{equation}
  with
  \begin{equation}
    \Dis = \sum_{n=0}^{N-1} e^{\imath \frac{2\pi}{N} n} |n\rangle \langle n| = e^{\imath \frac{2\pi}{N} \hat N}
  \end{equation}
  The points labelled $q$ and $q+p$ are then linkable ($\forall p \in \mathbb N$), with a displacement operator $\Dis(q+p,q) = \Dis^p$. $\delta \Pi  = \frac{2\pi}{N} \hat N$ and the linking vector observable is
  \begin{eqnarray}
    \pmb{\delta \vec \ell}_{q+1,q} & = & -\imath \frac{2\pi}{N} [\hat N, \vec X] \\
    & = & -\imath \frac{2\pi}{N} \left( \frac{Z^\dagger-Z}{2} \partial_1 - \frac{Z+Z^\dagger}{2\imath} \partial_2 \right)
  \end{eqnarray}
  since $[\hat N,Z] = -Z$ and $[\hat N,Z^\dagger] = Z^\dagger$. It follows that
  \begin{equation}
    \dist_{\pmb \gamma}(\omega_{q+1},\omega_q) = \sqrt{\pmb \gamma(L_{\delta \Pi},L_{\delta \Pi})} = \frac{2\pi}{N}
  \end{equation}
  The noncommutative distance between the points labelled by $q$ and $q+1$ is then the distance along the unit circle between the two points. While $M_\Lambda$ is disconnected, it is totally linkable permitting to compute the noncommutative distance (in contrast with the fuzzy hyperboloid for which $M_\Lambda$ is connected permitting to compute $\gamma$ but is totally unlinkable).
\end{example}

The linking vector observable can also be interpreted as a noncommutative gauge field onto $\mathfrak M$ as we can see this in the following section.

\subsubsection{The noncommutative gauge potential}
The displacement operator defines a noncommutative gauge potential on $\mathfrak M$, $\forall \Dis \in \Env(\mathfrak X)$:
\begin{equation}
  \mathbf A_{\Dis} = \imath \Dis^\dagger \dnc \Dis \in \Omega^1_{\Der}(\mathfrak X)
\end{equation}
which is directly associated with the displacement by
\begin{eqnarray}
  \omega_x(\mathbf A_{\Dis(y,x)}(L_{X^i})) & = & \llangle \Lambda(x)|(\Dis(y,x)^\dagger X^i \Dis(y,x)-X^i)|\Lambda(x)\rrangle \\
  & = & y^i-x^i
\end{eqnarray}
$ \omega_x(\mathbf A_{\Dis(y,x)}(L_{X^i}))$ is then the distance between $x$ and $y$ in the $i$-direction of $\mathbb R^3$. In other words, the linking vector observable previously introduced eq.(\ref{linkvect}) is
\begin{equation}
  \pmb{\delta \vec \ell}_{yx} = \mathbf A_{\Dis(y,x)} (L_{\vec X})
\end{equation}
and so
\begin{equation}
  \mathbf A_{\Dis(x+\delta x,x)} = \dnc \delta \Pi_x + \mathcal O(\|\delta \vec x\|^2)
\end{equation}
where $y=x+\delta x$ and $\Dis(y,x) = e^{\imath \delta \Pi_x}$. Since by construction $\dnc \Dis = -\imath \Dis \mathbf A_\Dis$, we can see $\Dis(y,x)$ as the noncommutative path-ordered exponential of $\mathbf A_\Dis$ along the noncommutative path $(\Ad_{\Dis(y,x)},\omega_x)$: $\forall \varphi \in \Diff M_\Lambda$ and $\forall x \in M_\Lambda$:
\begin{equation}
  \Dis(\varphi(x),x) \equiv \Ped_{n.c.}^{-\imath \int_{\pmb \omega(\varphi,x)} \mathbf A_\Dis}
\end{equation}
$\Dis(\varphi(x),x)$ can be then viewed as the noncommutative geometric phase accumulated along the arrow $(\varphi,x) \in \Morph(\mathscr M_\Lambda)$.\\

Under an external gauge change $J \in SO(3)$, the noncommutative gauge potential becomes 
\begin{equation}
  \hat \mathbf A_{\Dis} = \mathbf G_J^\dagger \mathbf A_{\Dis} \mathbf G_J + \imath \mathbf G_J^\dagger \dnc \mathbf G_J
\end{equation}
Indeed $\Dis(J^{-1}x_2,J^{-1}x_1) = \Dis(x_2,x_1) \mathbf G_J(x_2,x_1) \Rightarrow \dnc \Dis(J^{-1}x_2,J^{-1}x_1) = (\dnc \Dis(x_2,x_1)) \mathbf G_J(x_2,x_1) + \Dis(x_2,x_1) \dnc \mathbf G_J(x_2,x_1)$. Then $\imath \hat \Dis^\dagger \dnc \hat \Dis =  \imath \mathbf G_J^{-1} (\Dis^{-1}\dnc \Dis) \mathbf G_J +\imath \mathbf G_J^{-1}\dnc \mathbf G_J$, proving the formula for $\hat \mathbf A_\Dis$.

\begin{prop}
  The noncommutative curvature is zero.
\end{prop}
\begin{proof}
  $\forall X,Y \in \mathfrak X$,
  \begin{eqnarray}
    \mathbf F_\Dis(L_X,L_Y) & = & \dnc \mathbf A_\Dis(L_X,L_Y) - \imath[\mathbf A_\Dis(L_X),\mathbf A_\Dis(L_Y)] \\
    & = & [X,\mathbf A_\Dis(L_Y)]-[Y,\mathbf A_\Dis(L_X)] - \mathbf A_\Dis(L_{[X,Y]}) \nonumber \\
    & & \qquad \qquad - \imath[\mathbf A_\Dis(L_X),\mathbf A_\Dis(L_Y)] \\
    & = & \imath[X,\Dis^\dagger[Y,\Dis]]-\imath[Y,\Dis^\dagger[X,\Dis]]-\imath\Dis^\dagger[[X,Y],\Dis] \nonumber \\
    & & \qquad \qquad +\imath[\Dis^\dagger[X,\Dis],\Dis^\dagger[Y,\Dis]] \\
    & = & \imath [X,\Dis^\dagger Y\Dis] - \imath[X,Y] - \imath[Y,\Dis^\dagger X\Dis]+\imath [Y,X] \nonumber \\
    & & \quad -\imath \Dis^\dagger[X,Y]\Dis + \imath[X,Y] \nonumber \\
    & & \quad +\imath \underbrace{[\Dis^\dagger X\Dis-X,\Dis^\dagger Y \Dis-Y]}_{\Dis^\dagger[X,Y]\Dis-[\Dis^\dagger X\Dis,Y]-[X,\Dis^\dagger Y \Dis]+[X,Y]} \\
    & = & 0
  \end{eqnarray}
\end{proof}

\subsection{About the two distances between quasicoherent states\\}
Consider a path $\mathscr C$ on $M_\Lambda$ followed by the probe and suppose that along this path measurements are performed on a sequence of infinitely close points $x_0,x_1,...,x_N$. We can make an analogy between the quantum random processes with classical stochastic processes by considering that the position of $n$-th point of the sequence is a classical random variable of gaussian law with mean value $\vec x_n$ and standard deviation $\Delta_{x_n} \vec X$ (this is just an analogy, the quantum process cannot be identified as this because of the noncommutativity of the observables $X^i$ implying that the result of measurements depends on the ordering of the different measures, and so the measurement of the three coordinates is a counterfactual definiteness). In this analogy, the observed trajectory is a Brownian motion. By repeating the experiment, we have then a beam of random paths of average path $\mathscr C$. But the length of the average path is not the average length of the random paths. In this analogy, $\ell(\mathscr C) = \int_{\mathscr C} \sqrt{\gamma(\partial_s,\partial_s)} ds$ is the length of the average path whereas $\pmb \ell ((\omega_{x_n})_n) = \sum_{n=0}^{N-1} |\pmb{\delta \ell}_{x_{n+1}x_n}|= \sum_{n=0}^{N-1} \sqrt{\pmb \gamma(L_{\delta \Pi_{x_n}},L_{\delta \Pi_{x_n}})}$ is the average length of the random paths. The classical metric $\gamma$ (which results from $\pmb{d\ell^2}_x = \slashed D_x^2$) is then the square length measure onto the average manifold $M_\Lambda$ whereas the quantum metric $\pmb \gamma$ is the average square length measure onto ``quantum fluctuations'' of the fuzzy space. $\dist_\gamma(x+\delta x,x) = \sqrt{\gamma_{ab}(x)\delta x^a \delta x^b}$ is the distance between the average positions of the states $\omega_x$ and $\omega_{x+\delta x}$ (which are the normal states closest to the notion of point in $\mathfrak M$), whereas $\dist_{\pmb \gamma}(\omega_{x+\delta x},\omega_x) = \omega_x(|\pmb{\delta \ell}_{x+\delta x,x}|)$ is the mean value in the state $\omega_x$ of the distance observable between the two probe positions $x$ and $x+\delta x$.\\
In another viewpoint, suppose that $\delta \Pi_x$ commutes with two coordinate observables but not with $X^1$. In this case, $\sqrt{\pmb \gamma(L_{\delta \Pi_x},L_{\delta \Pi_x})} = -\imath [\delta \Pi_x,X^1]$ and then we have the following Heisenberg uncertainty relation $\Delta_x X^1 \Delta_x \delta \Pi_x \geq \frac{1}{2} |\omega_x(\sqrt{\pmb \gamma(L_{\delta \Pi_x},L_{\delta \Pi_x})})| = \frac{1}{2} \dist_{\pmb \gamma}(\omega_{x+\delta x},\omega_x)$. $L_{\delta \Pi_x}$ being the generator of the displacement operator it can be viewed as an infinitesimal translation operator, and so as a linear momentum observable (conjugated to $X^1$ here) in the fuzzy space. $\dist_{\pmb \gamma}(\omega_{x+\delta x},\omega_x)$ is then position-momentum uncertainty product in the fuzzy space (we recall that the quasi-coherent states minimising the quantum uncertainties). But in the fuzzy space, the position observables are not independent $[X^i,X^j] \not=0$, and so in general a linear momentum observable $\delta \Pi_x$ does not commute with any position observable. So, $\sqrt{\pmb \gamma(L_{\delta \Pi_x},L_{\delta \Pi_x})}$ is a kind of average commutator of the linear momentum observable with the position observables. We could then define $\Delta_x X_x \equiv \frac{\sqrt{\pmb \gamma(L_{\delta \Pi_x},L_{\delta \Pi_x})}}{2 \Delta_x V_x}$ as the quantum uncertainty of the location in the direction defined by $L_{\delta \Pi_x}$ (even if we cannot define a position observable $X_x$ conjugated to $L_{\delta \Pi_x}$: in general $\nexists X_x$ such that $[\delta \Pi_x,X_x]^2 = \delta_{ij}[\delta \Pi_x,X^i][\delta \Pi_x,X^j]$).\\
We can also relate the double definition of the metric with the category structure of $\mathscr M_\Lambda$. For $x,x+dx \in \Obj(\mathscr M_\Lambda)=M_\Lambda$, $\dist_\gamma(x,x+dx) = \sqrt{\gamma_{ab}(x)dx^adx^b}$ is the distance between two objects. For $(\varphi,x) \in \Morph(\mathscr M_\Lambda) = \Diff(M_\Lambda) \times M_\Lambda$ with $\varphi(x)-x \equiv \delta x$ (infinitesimal diffeomorphism), $\dist_{\pmb \gamma}(\omega_x,\omega_{\varphi(x)}) = \omega_x(\sqrt{\pmb \gamma(L_{\delta \Pi_x},L_{\delta \Pi_x})})$ is the length of the arrow $\pmb \omega(\varphi,x) = (\Ad_{\Dis(\varphi(x),x)},\omega_x)$, whereas $\dist_\gamma(x,\varphi(x))$ is the distance between its source and target. \ref{NCvsCat} summarises the relations between the different metrics.\\

\begin{example}{2}{Fuzzy surface plots}
 From the displacement operator we find that the linear momentum observable is
  \begin{equation}
    \delta \Pi_\alpha = -\imath(\delta \alpha a^+-\delta \bar \alpha a)
  \end{equation}
  In physical models, we can see $\delta \alpha$ as the complex representation of an electric field, if we interpret $a$ and $a^+$ as creation/annihilation operator of photons, the linear momentum observable $\delta \Pi_\alpha$ can be identified with the second quantised electric field operator.  
  \begin{eqnarray}
    \left[\delta \Pi_\alpha,X^1\right] & = & \imath \Re(\delta \alpha) \\
    \left[\delta \Pi_\alpha,X^2\right] & = & \imath \Im(\delta \alpha)
  \end{eqnarray}
  It follows that for the fuzzy plane $\pmb \gamma(L_{\delta \Pi_\alpha},L_{\delta \Pi_\alpha}) = |\delta \alpha|^2$ and then
  \begin{equation}
    \dist_\gamma(\alpha,\alpha+\delta \alpha) = \dist_{\pmb \gamma}(\omega_\alpha,\omega_{\alpha+\delta \alpha}) = |\delta \alpha|
  \end{equation}
  For the fuzzy plane, the length of the average path is equal to the average length of the random paths.\\
  For the fuzzy elliptic paraboloid, we have
  \begin{eqnarray}
    \left[\delta \Pi_\alpha,X^3\right] & = & -\imath \epsilon[\delta \alpha a^+-\delta \bar \alpha a,a^+a] \\
    & = & \imath \epsilon (\delta \alpha a^+ + \delta \bar \alpha a)
  \end{eqnarray}
  It follows that $\pmb \gamma(L_{\delta \Pi_\alpha},L_{\delta \Pi_\alpha}) = |\delta \alpha|^2+\epsilon^2(\delta \alpha a^++\delta \bar \alpha a)^2$ and so
  \begin{equation}
    \sqrt{\pmb \gamma(L_{\delta \Pi_\alpha},L_{\delta \Pi_\alpha})} = |\delta \alpha|+\frac{\epsilon^2}{2}|\delta \alpha|\left(\frac{\delta \alpha}{\delta \bar \alpha} (a^+)^2+ \frac{\delta \bar \alpha}{\delta \alpha}a^2+2a^+a+1\right) + \mathcal O(\epsilon^3)
  \end{equation}
  and then
  \begin{equation}
    \dist_{\pmb \gamma}(\omega_\alpha,\omega_{\alpha+\delta \alpha}) = |\delta \alpha|\left(1+\epsilon^2\left(\Re\left(\alpha^2\frac{\delta \bar \alpha}{\delta \alpha}\right)+|\alpha|^2+1/2\right)\right) + \mathcal O(\epsilon^3)
  \end{equation}
  whereas
  \begin{eqnarray}
    \dist_\gamma(\alpha,\alpha+\delta \alpha) & = & (1+4\epsilon^2 \Re(\alpha)^2)\delta \Re(\alpha)^2+(1+4\epsilon^2 \Im(\alpha)^2)\delta \Im(\alpha)^2 \nonumber \\
    & & \quad +8\epsilon^2 \Re(\alpha)\Im(\alpha) \delta \Re(\alpha) \delta \Im(\alpha)
  \end{eqnarray}
  To help the comparison, we can consider the case $\alpha=x \in \mathbb R$ and $\delta \alpha=\delta x \in \mathbb R$, for which
  \begin{eqnarray}
    \dist_{\pmb \gamma}(\omega_x,\omega_{x+\delta x}) & = & \delta x(1+2\epsilon^2x^2+\epsilon^2/2) + \mathcal O(\epsilon^3)\\
    \dist_{\gamma}(x,x+\delta x) & = & \delta x(1+2\epsilon^2x^2)
  \end{eqnarray}
  As expected, the average length of the random paths is larger than the length of the average path.\\
  For the fuzzy hyperbolic paraboloid, we have the same thing:
  \begin{eqnarray}
    \left[\delta \Pi_\alpha,X^3\right] & = & -\frac{\imath}{2} \epsilon[\delta \alpha a^+-\delta \bar \alpha a,a^2+(a^+)^2] \\
    & = & \imath \epsilon (\delta \alpha a + \delta \bar \alpha a^+)
  \end{eqnarray}
  and then
    \begin{equation}
    \sqrt{\pmb \gamma(L_{\delta \Pi_\alpha},L_{\delta \Pi_\alpha})} = |\delta \alpha|+\frac{\epsilon^2}{2}|\delta \alpha|\left(\frac{\delta \alpha}{\delta \bar \alpha} a^2+ \frac{\delta \bar \alpha}{\delta \alpha}(a^+)^2+2a^+a+1\right) + \mathcal O(\epsilon^3)
  \end{equation}
    Finally we have (by taken into account the contribution of the entanglement of $|\Lambda(\alpha)\rrangle$):
    \begin{equation}
      \dist_{\pmb \gamma}(\omega_\alpha,\omega_{\alpha+\delta \alpha}) = |\delta \alpha|\left(1+\epsilon^2\left(\Re\left(\alpha^2\frac{\delta \alpha}{\delta \bar \alpha}\right)+2|\alpha|^2+3/4\right)\right) + \mathcal O(\epsilon^3)
    \end{equation}
    whereas
    \begin{eqnarray}
      \dist_\gamma(\alpha,\alpha+\delta \alpha) & = & (1+4\epsilon^2 \Re(\alpha)^2)\delta \Re(\alpha)^2+(1+4\epsilon^2 \Im(\alpha)^2)\delta \Im(\alpha)^2 \nonumber \\
    & & \quad -8\epsilon^2 \Re(\alpha)\Im(\alpha) \delta \Re(\alpha) \delta \Im(\alpha)
    \end{eqnarray}
    Anew for $\alpha=x \in \mathbb R$ and $\delta \alpha=\delta x \in \mathbb R$ we have
    \begin{eqnarray}
      \dist_{\pmb \gamma}(\omega_x,\omega_{x+\delta x}) & = & \delta x(1+3\epsilon^2x^2+\frac{3}{4}\epsilon^2) + \mathcal O(\epsilon^3)\\
      \dist_{\gamma}(x,x+\delta x) & = & \delta x(1+2\epsilon^2x^2)
    \end{eqnarray}
    or for $\alpha=(1+\imath)x$ and $\delta \alpha(1+\imath)\delta x$,
    \begin{eqnarray}
      \dist_{\pmb \gamma}(\omega_{(1+\imath)x},\omega_{(1+\imath)(x+\delta x)}) & = & \sqrt 2 \delta x(1+2\epsilon^2x^2+\frac{3}{4}\epsilon^2) + \mathcal O(\epsilon^3)\\
      \dist_{\gamma}((1+\imath)x,(1+\imath)(x+\delta x)) & = & \sqrt 2 \delta x
    \end{eqnarray}

\end{example}

Table \ref{dynbeh} summarises the interpretations concerning the energy of the paths chosen on $M_\Lambda$.
\begin{table} \small
  \begin{tabular}{l|c|c}
    & energy mean value & energy uncertainty \\
    \hline
    \begin{minipage}{3cm} strong adiabatic transport along any path \end{minipage} & $\llangle \slashed D_{x(t)} \rrangle_{\Psi(t)} = 0$ & $\Delta_{\Psi(t)} \slashed D_{x(t)}^2 = 0$ \\
    \hline
    \begin{minipage}{3cm} strong adiabatic transport along a minimising geodesic path \end{minipage} & $\llangle \slashed D_{x(t)} \rrangle_{\Psi(t+\Delta t)} = 0$ & $\Delta_{\Psi(t+\Delta t)} \slashed D_{x(t)}^2$ minimised \\
    \hline
    \begin{minipage}{3cm} strong adiabatic transport along a path such that $ \omega_{x(t)}(|\pmb{\delta \ell}_{x(t+\Delta t),x(t)}|)$ is a small constant \end{minipage} & $\llangle \slashed D_{x(t)} \rrangle_{\Psi(t+\Delta t)} = 0$ & $\llangle \sqrt{\frac{1}{2} \tr_{\mathbb C^2} \Delta\slashed D^2_{x(t+\Delta t),x(t)}} \rrangle_{\Psi(t)}$ minimised \\
    \hline
  \end{tabular}
  \caption{\label{dynbeh} Behaviours of the energy with respect to the different paths for $\Delta t$ a fixed small value. $\Psi(t)$ being defined eq.(\ref{StrAdiabTransp})}
\end{table}

\subsection{The quantum geometric tensor}
To characterise the geometry of quantum state spaces, Provost and Valle introduced the concept of quantum geometric tensor $\mathcal Q_{ab} = \gamma_{ab} + \imath \mathscr S_{ab}$ in \cite{Provost} (see also \cite{Shapere, Cheng}) which is defined such that $\mathcal Q_{ab}$ is invariant by phase change, the symmetric part $\gamma_{ab}$ is the component of a relevant metric of the manifold of quantum states, and the antisymmetric part $\mathscr S_{ab}$ is the component of a relevant symplectic form of the manifold of quantum states. For example, in adiabatic approximations where the manifold is the set of parameters of a continuous eigenvector, or when the manifold is the projected Hilbert space of all quantum states, $\mathcal Q_{ab} = \langle \partial_a \psi| (1-|\psi\rangle\langle \psi|)|\partial_b \psi \rangle$ (with $\|\psi\|=1$). $\gamma$ is the Fubini-Study metric (see \cite{Shapere}) measuring the distance between two quantum states by $|\langle \psi_1|\psi_2\rangle|^2$  and $\mathscr S$ is the Berry curvature measuring the non-adiabatic factors.\\

In accordance with the definition of the metric of $M_\Lambda$, the quantum geometric tensor of the quasicoherent geometry is defined by
\begin{eqnarray}
  \mathcal Q_{ab} & = & \llangle \partial_a \Lambda|\slashed D_x^2|\partial_b \Lambda \rrangle \\
  & = & \gamma_{ab} + \imath \mathscr S_{ab}
\end{eqnarray}
where $\gamma_{ab}$ is the metric of $M_\Lambda$ and where
\begin{eqnarray}
  \mathscr S_{ab} & = & {\varepsilon_{ij}}^k \llangle \Lambda|\sigma_k|\Lambda \rrangle \frac{\partial x^i}{\partial s^a} \frac{\partial x^j}{\partial s^b} \\
  & = & \frac{1}{2} {\varepsilon_{ij}}^k \llangle \Lambda|\sigma_k|\Lambda \rrangle \frac{\partial x^i}{\partial s^{[a}} \frac{\partial x^j}{\partial s^{b]}} 
\end{eqnarray}
The symplectic 2-form is then
\begin{eqnarray}
  \mathscr S & = & {\varepsilon_{ij}}^k \llangle \Lambda|\sigma_k|\Lambda \rrangle \frac{\partial x^i}{\partial s^1} \frac{\partial x^j}{\partial s^2} ds^1 \wedge ds^2 \\
  & = & \vec n_x \cdot (\vec t_{x1} \times \vec t_{x2}) ds^1 \wedge ds^2 \\
  & = & \|\vec n_x\| ds^1 \wedge ds^2
\end{eqnarray}
where $\vec t_{xa} = \frac{\partial x^i}{\partial s^a} \partial_i$ are two orthonormal tangent vectors of $M_\Lambda$ at $x$ (with $\{s^1,s^2\}$ local curvilinear coordinates on $M_\Lambda$) and so $\vec t_{x1} \times \vec t_{x2}$ is the unitary normal vector at $x$. $\mathscr S$ inherits then its interpretation from $\|\vec n_x\|$. It measures the quantum purity of the quasicoherent geometry at $x$ (inverse measure of the entanglement of the quasicoherent state). For this reason, we call $\mathscr S$ the purity form of $M_\Lambda$. Note that in accordance with this interpretation, $\mathscr S$ can be degenerate at $x$ ($\mathscr S(x) =0$) if $|\Lambda(x)\rrangle$ is maximally entangled.\\ 

To finish, it is necessary to verify that $\mathcal Q_{ab}$ is independant of the gauge choice:
\begin{prop}
  $\mathcal Q_{ab}$ is invariant by internal gauge changes. 
\end{prop}

\begin{proof}
  Firstly, let $|\tilde \Lambda \rrangle = h|\Lambda \rrangle$ with $h \in U(1)$. We have then $|\partial_a \tilde \Lambda \rrangle = \frac{\partial h}{\partial s^a} |\Lambda \rrangle + h |\partial_a \Lambda \rrangle$ and then $\slashed D_x |\partial_a \tilde \Lambda \rrangle = h \slashed D_x |\partial_a \Lambda \rrangle$. It follows that $\tilde {\mathcal Q}_{ab} = \llangle \partial_a \tilde \Lambda | \slashed D_x |\partial_b \tilde \Lambda \rrangle = \mathcal Q_{ab}$.\\
  Now let $|\tilde \Lambda \rrangle = u |\Lambda \rrangle$ with $u \in U_x \subset SU(2)_{a.s.}$. $\slashed D_x u|\Lambda\rrangle = 0 \Rightarrow \slashed D_x |\partial_a \tilde \Lambda \rrangle = \sigma_i |\tilde \Lambda \rrangle \frac{\partial x^i}{\partial s^a}$. It follows that $\tilde {\mathcal Q}_{ab} = \gamma_{ab} + \imath {\varepsilon_{ij}}^k \llangle \Lambda|u^\dagger \sigma_k u|\Lambda \rrangle \frac{\partial x^i}{\partial s^a} \frac{\partial x^j}{\partial s^b}$.
  \begin{itemize}
  \item If $\mathfrak M$ is strongly non-degenerate at $x$, then $u=\id$ and $\tilde {\mathcal Q}_{ab} = \mathcal Q_{ab}$.
  \item If $\mathfrak M$ is not strongly non-degenerate and presents at least two independent non-trivial normal vectors at $x$, then we can set $\tilde {\mathscr S} = \mathscr S = 0$ since $\dim M_\Lambda \leq 1$ at $x$ (and then $ds^1\wedge ds^2=0$).
  \item In the other cases, $\llangle \Lambda |u^\dagger \vec \sigma u|\Lambda \rrangle = \llangle \Lambda |\vec \sigma |\Lambda \rrangle$ since $u$ belongs to the isotropy group of $\rho_\Lambda$ (or the all purity normal vectors are zero), and then $\tilde {\mathcal Q}_{ab} = \mathcal Q_{ab}$.
  \end{itemize}
\end{proof}

We can remark that $\omega_x^* \pi_x^* \mathscr S = \frac{1}{2} {\varepsilon_{ij}}^k \llangle \Lambda|\sigma_k|\Lambda \rrangle \dnc X^i \wedge \dnc X^j$ (with $\dnc X^i \wedge \dnc X^j = \dnc X^i \otimes \dnc X^j - \dnc X^j \otimes \dnc X^i$). We can then set the spin noncommutative 2-form $\pmb {\mathscr S} = \frac{1}{2} {\varepsilon_{ij}}^k \sigma_k \otimes \dnc X^i \wedge \dnc X^j$ to replace in the noncommutative geometry of $\mathfrak M$ the symplectic form (the dimension being odd no symplectic form can be defined); and since $\pmb \gamma = \delta_{ij} \dnc X^i \otimes \dnc X^j$, we set as quantum noncommutative geometric tensor:
\begin{equation}
  \pmb{\mathcal Q}_{ij} = \delta_{ij} + \imath {\varepsilon_{ij}}^k \sigma_k
\end{equation}
which satisfies
\begin{eqnarray}
  \pmb{\mathcal Q}_{ij} \otimes (X^i-x^i)(X^j-x^j) & = & \|\vec X - \vec x\|^2 + \imath {\varepsilon_{ij}}^k \sigma_k \otimes X^i X^j \\
  & = & \slashed D_x^2 \\
  & = & \pmb{d\ell^2}_x
\end{eqnarray}

The meaning of $\pmb{\mathscr S}$ can be understand by considering two displacement operators $\Dis(y,x) = e^{\imath \delta \Pi_x^1}$ and $\Dis(z,x) = e^{\imath \delta \Pi_x^2}$ with $y=x+\delta x^1$ and $z=x+\delta x^2$ ($\delta x^1 \not= \delta x^2$).
\begin{eqnarray}
  \pmb{\mathscr S}(L_{\delta \Pi_x^1},L_{\delta \Pi_x^2}) & = & - {\varepsilon_{ij}}^k \sigma_k \otimes [\delta \Pi_x^1,X^i][\delta \Pi_x^2,X^j] \\
  & = & {\varepsilon_{ij}}^k \sigma_k \otimes \pmb{\delta \ell}_{yx}^i \pmb{\delta \ell}_{zx}^j \\
  & = & \vec \sigma \odot \pmb{\delta \vec \ell}_{yx} \times \pmb{\delta \vec \ell}_{zx}
\end{eqnarray}
$\pmb{\delta \vec \ell}_{yx}$ and $\pmb{\delta \vec \ell}_{zx}$ being two linking vector observables at $M_\Lambda$ and $\vec \sigma$ being the normal vector observable at $M_\Lambda$, $\pmb{\mathscr S}(L_{\delta \Pi_x^1},L_{\delta \Pi_x^2})$ is the volume observable of the noncommutative parallelepiped defined by $(\vec \sigma,\pmb{\delta \vec \ell}_{yx},\pmb{\delta \vec \ell}_{zx})$. We can then interpete $\pmb{\mathscr S}$ as the non-commutative volume form of $\mathfrak M$. The quantum noncommutative geometric tensor $\pmb{\mathcal Q}_{ij}$ provides then the bicovector metric $\pmb{\gamma}$ and the volume form $\pmb{\mathscr S}$ of the fuzzy space $\mathfrak M$ (and also by contraction with $X^i-x^i$ the square lenght observable metric $\pmb{d\ell^2}_x$).

\section{Adiabatic regimes and geodesics}
The quasicoherent geometry is related to the eigenvector $|\Lambda(x)\rrangle$ of $\slashed D_x$. We can then consider for slow moves of the probe $t \mapsto x(t)$ onto $M_\Lambda$ the adiabatic regimes supported by $|\Lambda \rrangle$. This could permit to define the classical spacetime closest to $\mathfrak M$, in the meaning that the space part is related to quasicoherent states (states minimising the Heisenberg uncertainties) and the time part is related to the adiabatic approximation (quantum dynamics closest to a semi-classical approximation).

\subsection{The strong adiabatic regime}
\subsubsection{Berry gauge potential and Berry curvature}
The inner product $\llangle \bullet | \bullet \rrangle$ defines a natural gauge potential on $M_\Lambda$:
\begin{equation}
  A  =  -\imath \llangle \Lambda|d|\Lambda \rrangle \in \Omega^1(M_\Lambda,\mathbb R)
\end{equation}

$A$ is the generator of the geometric phase accumulated during a strong adiabatic transport of $|\Lambda \rrangle$ along a path $\mathscr C$ on $M_\Lambda$ \cite{Shapere, Bohm}. If $|\Lambda \rrangle$ is strongly nondegenerate along the path $\mathscr C$ linking $x$ and $y$, by using the standard adiabatic theorem \cite{Teufel} (strong adiabatic approximation), we have
\begin{equation}
  |\Psi(T) \rrangle = e^{-\imath \int_{\mathscr C} A} |\Lambda(y)\rrangle + \mathcal O(1/T) \qquad \text{if } |\Psi(0)\rrangle = |\Lambda(x)\rrangle
\end{equation}
$T$ being the transport duration along $\mathscr C$ and $|\Psi\rrangle$ being the solution of the Schr\"odinger equation eq.(\ref{SchroEq}). The condition of parallel transport associated with the strong adiabatic approximation is just the usual one introduced by Simon \cite{Shapere}:
\begin{equation}
  \llangle \Psi | \dot \Psi \rrangle = 0
\end{equation}

\begin{prop}\label{gauge1}
    Under an inner gauge transformation $(h,u) \in U(1) \times U_x \subset U(1) \times SU(2)_{a.s.}$, the gauge potential $A$ becomes:
  \begin{equation}
    \tilde A  =  A - \imath d\ln h -\imath \omega_x(u^{-1}du)
  \end{equation}
    Under a local external gauge transformation $J \in \underline{SO(3)}_{M_\Lambda}$, the gauge potential $A$ becomes:
  \begin{equation}
    \hat A = J^* A + \imath \omega_{J^{-1}x}(du_{J} u_J^{-1})
  \end{equation}
    where $u_J \in SU(2)$ is the transformations associated with $J$ viewed as a linear map of $\mathbb R^3$: $(Jy)^i = {J^i}_j y^j$. $J^* : \Omega^1\mathbb R^3 \to \Omega^1\mathbb R^3$ is defined by $(J^* k)(x) = k_a(J^{-1}x) ({(J^{-1})^i}_j dx^j - {(J^{-1}dJJ^{-1})^i}_j x^j)$, $\forall k \in \Omega^1\mathbb R^3$.
\end{prop}

\begin{proof}
  $|\tilde \Lambda \rrangle = hu|\Lambda \rrangle \Rightarrow d|\tilde \Lambda \rrangle = hud|\Lambda \rrangle + dh \times u|\Lambda \rrangle + du \times h|\Lambda \rrangle$. So $-\imath |d\tilde \Lambda \rrangle \llangle \tilde \Lambda| = u (-\imath |d\Lambda\rrangle \llangle \Lambda|) u^{-1} -\imath h^{-1}dh \times u|\Lambda\rrangle\llangle \Lambda|u^{-1} -\imath du u^{-1} \times u|\Lambda\rrangle \llangle \Lambda|u^{-1}$. By taking the total trace of this last equation we find the formula for $\tilde A$.\\
   We recall that we have $|\hat \Lambda(x) \rrangle = u_J^{-1} |\Lambda(J^{-1}x)\rrangle$. Then $-\imath |d\hat \Lambda \rrangle \llangle \hat \Lambda| = \imath u_J^{-1} du_J u_J^{-1} |\Lambda \rrangle \llangle \Lambda|u_J -\imath u_J^{-1} J^* |d\Lambda \rrangle \llangle \Lambda|u_J$ since by construction $df(J^{-1}x) = (J^* df)(x)$ for any function $f$. By taking the total trace of this last equation we find the formula for $\hat A$
\end{proof}

The Berry gauge potential defines the Berry curvature $F = dA \in \Omega^2(M_\Lambda,\mathbb R)$ \cite{Shapere, Bohm}. In accordance with usual interpretations of the Berry phases, the gauge fields can be assimilated to virtual electromagnetic fields, as we can see this on the example of the fuzzy surface plots:

\begin{example}{2}{Fuzzy surface plots}
  By using the following derivative rules:
  \begin{eqnarray}
    \frac{\partial}{\partial \alpha} |n\rangle_\alpha & = & \frac{\bar \alpha}{2}|n\rangle_\alpha + \sqrt{n+1}|n+1\rangle_\alpha \\
    \frac{\partial}{\partial \bar \alpha} |n\rangle_\alpha & = & - \frac{\alpha}{2}|n\rangle_\alpha - \sqrt{n}|n-1\rangle_\alpha
  \end{eqnarray}
  ($\partial_{\bar \alpha}|0\rangle_\alpha = - \frac{\alpha}{2}|n\rangle_\alpha$), we have for the fuzzy plane and the fuzzy paraboloids, the following Berry potential:
  \begin{eqnarray}
    A & = & -\imath \llangle \Lambda|d|\Lambda \rrangle \\
    & = & -\imath \frac{\bar \alpha d\alpha - \alpha d\bar \alpha}{2} + \mathcal O(\epsilon^2) \\
    & = & x^1 dx^2 - x^2dx^1 + \mathcal O(\epsilon^2) \\
    & = & r^2 d\theta + \mathcal O(\epsilon^2) \\
  \end{eqnarray}
  with $r=|\alpha| = \sqrt{(x^1)^2+(x^2)^2}$ and $\theta = \arg \alpha = \arctan \frac{x^2}{x^1}$. The rest $\mathcal O(\epsilon^2)$ is exactly $0$ for the fuzzy plane. The associated Berry curvature is then
  \begin{eqnarray}
    F & = & dA \\
    & = & \imath d\alpha \wedge d\bar \alpha + \mathcal O(\epsilon^2) \\
    & = & 2 dx^1 \wedge dx^2 + \mathcal O(\epsilon^2) \\
    & = & 2r dr \wedge d\theta + \mathcal O(\epsilon^2)
  \end{eqnarray}
  Following the usual interpretation of the Berry curvature, $A \in \Omega^1 (M_\Lambda,\mathbb R)$ and $F \in \Omega^2 (M_\Lambda,\mathbb R)$ are similar to a classical magnetic potential and to a classical magnetic field of a magnetic monopole at $\alpha = 0$, the magnetic field being normal to $M_\Lambda$ (even for the curved cases, this is due to the perturbative analysis at first order). Since $\Dis(\beta,\alpha) =e^{(\beta-\alpha)a^+-(\bar \beta - \bar \alpha)a}$ for the three models, we have
  \begin{eqnarray}
    \delta \Pi_\alpha & = & -\imath (\delta \alpha a^+ - \delta \bar \alpha a) \\
    & = & 2(\delta x^2 X^1 - \delta x^1 X^2)
  \end{eqnarray}
  It follows that the noncommutative gauge potential is
  \begin{eqnarray}
    \mathbf A_{\Dis(\alpha+d\alpha,\alpha)} & = & \dnc \delta \Pi_\alpha + \mathcal O(|\delta \alpha|^2)\\
    & = & -\imath (\delta \alpha \dnc a^+ - \delta \bar \alpha \dnc a) + \mathcal O(|\delta \alpha|^2) \\
    & = & 2(\delta x^2 \dnc X^1 - \delta x^1 \dnc X^2) + \mathcal O(\|\delta \vec x\|^2)
  \end{eqnarray}
  Since $\delta \alpha$ can be assimilated to the complex representation of an electric field, and if we interpret $a$ and $a^+$ as creation/annihilation operators of photons, we can see $\mathbf A_{\Dis(\alpha+d\alpha,\alpha)} \in \Omega^1_\Der(\mathfrak X)$ as similar to a second quantised electric field in $\mathfrak M$. The zero noncommutative curvature $\mathbf F_\Dis = 0$ is the noncommutative equivalent to the classical zero curl of the electric field (in the stationary regime). The electric field like appears as a quantised field whereas the magnetic field like appears as a classical field.
\end{example}

The interpretation of the gauge fields as virtual electromagnetic fields in $M_\Lambda/\mathfrak M$ is just a useful analogy.

\subsubsection{The foliated eigen spacetime of $\mathfrak M$}\label{foliation}
By invoking the adiabatic approximation in the previous section, we have considered the Schr\"odinger equation $\imath|\dot \Psi \rrangle = \slashed D_{x(t)} |\Psi \rrangle$. The time $t$ here is the time corresponding to slow variations of the probe $x$ onto $M_\Lambda$. It is then the time of the clock of the observer which studies the fuzzy space $\mathfrak M$ by measurements on the probe. In BFSS matrix models, $t$ is the proper time of the classical probe D0-brane, or in other words, the proper time of an observer comoving with the test particle revealing the gravitational effects. In quantum information theory, $t$ is the control time which runs during the realisation of a logical gate. It is then natural to think the emergent classical geometry as a spacetime $\mathcal M_\Lambda = \bigsqcup_{t\in \mathbb R}(t,M_\Lambda^{(t)})$ defined as a time foliation of leafs $M_\Lambda^{(t)}$ diffeomorphic to $M_\Lambda$. There is an infinite number of different foliations of $\mathcal M_\Lambda$ and there is no reason to think that the pertinent foliation is the trivial one (with null shift vector). We argue that the shift vector consistent with the geometry of $\mathfrak M$ is $A^a = \gamma^{ab} A_b$ where $A \in \Omega^1(M_\Lambda,\mathbb R)$ is the $U(1)$-gauge potential (providing the geometric interpretation of $A$). Indeed, suppose that the leaf $M^{(t)}_\Lambda$ is associated with $|\Lambda(x)\rrangle$, then $M_\Lambda^{(t+dt)}$ is associated with $e^{-\imath A_a \dot s^a dt} |\Lambda(x)\rrangle$ by applying the strong adiabatic assumption. Starting at $|\Lambda(x)\rrangle$ at $t$, we arrive at $e^{-\imath A_a \dot s_a dt} |\Lambda(x)\rrangle$ at $t+dt$. $A$ (by definition of a gauge potential on a principal bundle) defines then the shift between the two states. Since the quasicoherent states are the quantum states closest to the notion of points of a classical manifold, $A$ is then equivalent to the shift between the two pseudo-points. This corresponds to the definition of a shift vector of a foliated classical manifold \cite{Nakamura}: let $x$ be the point on $M_\Lambda^{(t)}$ of curvilinear coordinates $s$, $x'$ be the point at the intersection of the normal vector to $M_\Lambda^{(t)}$ in $\mathcal M_\Lambda$ at $x$ and $M_\Lambda^{(t+dt)}$, and $x''$ be the point of $M_\Lambda^{(t+dt)}$ of curvilinear coordinates $s$; then the shift vector of the foliation is $\overrightarrow{x'x''} = \vec A dt$. This choice of shift vector is moreover consistent with the interpretation in the adiabatic picture of $\slashed D_x$ as the Dirac operator of a fermionic string in the BFSS model \cite{Viennot2}.\\

The eigen spacetime $\mathcal M_\Lambda$ of $\mathfrak M$ is then endowed with the following metric:
\begin{equation}
  d\tau^2 = dt^2 - (\gamma^{ac}A_c dt+ds^a)(\gamma^{bd} A_d dt + ds^b) \gamma_{ab}
\end{equation}
(the laps function being supposed equal to 1). 
In BFSS matrix model, $\tau$ is the proper time of the fermionic string end moving onto the noncommutative D2-brane defined by $\mathfrak M$.\\
Note that $d\tau^2 = dt^2 - \omega_{s+ds}\left(\slashed D_{x(s^a-\gamma^{ab}A_bdt)}^2\right)$ (by writing that $\slashed D_{x(s^a-\gamma^{ab}A_bdt)} = \slashed D_x + \gamma^{ab}A_b \frac{\partial x^i}{\partial s^b} \sigma_i dt + \mathcal O(dt^2)$). $\sqrt{\omega_{s+ds}\left(\slashed D_{x(s^a-\gamma^{ab}A_bdt)}^2\right)}$ is the energy uncertainty with a double ``slipping'', the measurement is realised with a misalignment $ds$ onto the trajectory of the shifted point $x(s^a-\gamma^{ab}A_bdt)$ corresponding to the leaf defined with a time lag $dt$. For quantum information theory, it is then $\sqrt{dt^2-d\tau^2} = \sqrt{\omega_{s+ds}\left(\slashed D_{x(s^a-\gamma^{ab}A_bdt)}^2\right)}$ which is physically meaningful (and not directly $d\tau$).

\begin{prop}
  Let $h \in U(1)$ be an abelian inner gauge change. $d\tau^2$ is invariant under the gauge transformation $\tilde A = A - \imath d\ln h$ if and only if this one goes with $\varphi_t \in \Diff M_\Lambda$ such that $\varphi_t(x(s)) = x(\tilde s)$ with $\tilde s = f(s,t)$ ($f \in \mathcal C^\infty(\mathbb R^2\times \mathbb R,\mathbb R^2))$ satisfying the following equation
  \begin{equation}\label{diffgaugeEq}
    \partial_0 f^{\tilde a} = \imath \gamma^{ac} \frac{\partial f^{\tilde a}}{\partial s^a} \partial_c \ln h
  \end{equation}
  $\partial_0$ standing for derivative with respect to $t$.
\end{prop}

\begin{proof}
  \begin{eqnarray}
    d\tilde \tau^2 & = & dt^2-(\tilde \gamma^{\tilde a \tilde c} \tilde A_{\tilde c} dt+d\tilde s^{\tilde a})(\blacksquare^{\tilde b}) \tilde \gamma_{\tilde a \tilde b} \\
    & = & dt^2 - (\tilde \gamma^{\tilde a \tilde c}(A_{\tilde c}-\imath\partial_{\tilde c} \ln h)dt+d\tilde s^{\tilde a}) (\blacksquare^{\tilde b}) \tilde \gamma_{\tilde a \tilde b}
  \end{eqnarray}
  where $\blacksquare^{\tilde b}$ means the ``same expression with $\tilde b$ in place of $\tilde a$''. $\tilde s^{\tilde a} = f^{\tilde a}(s,t) \Rightarrow d\tilde s^{\tilde a} = \frac{\partial f^{\tilde a}}{\partial s^a} ds^a + \partial_0 f^{\tilde a} dt$.
  \begin{eqnarray}
    d\tilde \tau^2 & = & dt^2-\left(\gamma^{ac}\frac{\partial f^{\tilde a}}{\partial s^a}\frac{\partial f^{\tilde c}}{\partial s^c}(A_{\tilde c}-\imath \partial_{\tilde c} \ln h)dt + \frac{\partial f^{\tilde a}}{\partial s^a} ds^a + \partial_0 f^{\tilde a}dt \right)(\blacksquare^{\tilde b}) \tilde \gamma_{\tilde a \tilde b} \\
    & = & dt^2 - \left(\frac{\partial f^{\tilde a}}{ds^a}(\gamma^{ac}A_cdt+ds^a) +\underbrace{(\partial_0f^{\tilde a}-\imath \gamma^{ac} \frac{\partial f^{\tilde a}}{\partial s^a} \partial_c \ln h)}_{\equiv 0} \right) (\blacksquare^{\tilde b}) \tilde \gamma_{\tilde a \tilde b} \\
    & = & dt^2 - (\gamma^{ac} A_c dt+ds^a)(\blacksquare^b) \gamma_{ab} \\
    & = & d\tau^2
  \end{eqnarray}
\end{proof}

$\Diff \mathcal M_\Lambda$ must be restricted to diffeomorphisms for which it exists $h$ such that eq.(\ref{diffgaugeEq}) holds. This includes $\Diff M_\Lambda$ (time independent diffeomorphisms) which are associated with $h=1$. Moreover the $U(1)$-gauge changes must be restricted to the ones for which eq.(\ref{diffgaugeEq}) has solutions.\\
The couple $(h,\varphi_t) \in U(1) \times \Diff \mathcal M_\Lambda$ constitutes a change of foliation of $\mathcal M_\Lambda$ (the leaf at time $t$ is modified by $\varphi_t$ and the shift vector is modified by $-\imath d\ln h$). The $U(1)$-symmetry in the quantum space $\mathfrak M$ is then associated with the foliation changes of its eigen spacetime.\\

\begin{example}{4}{Fuzzy plane and fuzzy Painlev\'e-Gullstrand spacetime}
  For the fuzzy plane, the spacetime metric associated with the quasicoherent state $|\Lambda(\alpha)\rrangle = |0\rangle \otimes |\alpha \rangle$ is
  \begin{eqnarray}
    d\tau^2 & = & dt^2 - (A^a dt + ds^a)(A^b dt + ds^b) \delta_{ab} \\
    & = & dt^2 - (A^r dt + dr)^2 - r^2(A^\theta dt + d\theta)^2
  \end{eqnarray}
  in polar coordinates ($\alpha = re^{\imath \theta}$), with $A_r = 0$ and $A_\theta =r^2$. So $d\tau^2 = dt^2 - dr^2 - r^2(dt+d\theta)^2$. This is the metric of a flat spacetime with the effect of $A$ which will be discussed in the next section. But consider the following inner gauge change
  \begin{eqnarray}
    |\tilde \Lambda(\alpha) \rrangle & = & h(\alpha)|\Lambda(\alpha)\rrangle = e^{2\imath \sqrt{r_S |\alpha|}} |\Lambda(\alpha) \rrangle \\
    \tilde A & = & A - \imath d \ln h = A + \sqrt{\frac{r_S}{r}} dr
  \end{eqnarray}
  with $r_S >0$ a constant parameter. With this gauge choice, the metric becomes:
  \begin{equation}
    d\tilde \tau^2  =dt^2 - (dr + \sqrt{\frac{r_S}{r}} dt)^2 - r^2(dt+d\theta)^2
  \end{equation}
  In place of the metric a flat spacetime with the effect of the Berry potential, we have the metric of a Painlev\'e-Gullstrand spacetime (reduced of one dimension) with the effect of $A$. So the metric of a Schwarzschild black hole of Schwarzchild radius $r_S$ in the Painlev\'e-Gullstrand coordinates which are the coordinates such that $t$ is the proper time of a free-falling observer who starts from far away with zero velocity. This in accordance to the fact that, in quantum gravity matrix models, fuzzy spaces are the quantisation of the space from the point of view of an ideal Galilean observer for which $t$ is its classical clock (and so the proper time of this free-falling observer). The link between the plane and the Painlev\'e-Gullstrand spacetime is classically well-known, the space slices of the Painlev\'e-Gullstrand spacetime are flat. At the quantum level, the difference between the flat spacetime and the Painlev\'e-Gullstrand black-hole spacetime is just the gauge choice for $|\Lambda \rrangle$. But note that the two gauge choices are not physically equivalent because the gauge change $h$ is realised without its associated diffeomorphism. Let $\tilde r = f(r,t)$ be the generator of the diffeomorphism associated with $h$:
  \begin{equation}
    \frac{\partial f}{\partial t} = \imath \gamma^{rr} \frac{\partial f}{\partial r} \frac{\partial \ln h}{\partial r} \iff \frac{\partial f}{\partial t} = - \sqrt{\frac{r_S}{r}} \frac{\partial f}{\partial r}
  \end{equation}
  \begin{equation}
    \tilde r = f(r,t) = \left(r^{3/2} - \frac{3}{2} \sqrt{r_S} t \right)^{2/3}
  \end{equation}
  (with the initial condition $\tilde r = r$ at $t=0$). With this change of foliation in addition to the gauge change $h$, $d\tau^2$ is invariant and the flat spacetime remains the flat spacetime with just another coordinate system.\\
  Clearly, the fuzzyfication (of a 2D slice) of the Schwarzschild black hole is not the same for the Schwarzschild coordinates and for the Painlev\'e-Gullstrand coordinates: 
  \begin{center} \scriptsize
    \begin{tabular}{ccccc}
      \boxed{\begin{minipage}{3cm} classical Schwarzschild spacetime, Schwarzschild coordinates $(t',r,\theta,\varphi)$ \end{minipage}}  & $\xrightarrow{t=g(t',r)}$  & \boxed{\begin{minipage}{3cm} classical Schwarzschild spacetime, Painlev\'e-Gullstrand coordinates $(t,r,\theta,\varphi)$ \end{minipage}}  & $\times$  & \boxed{\begin{minipage}{3cm} classical Minkowski spacetime, polar coordinates $(t,\tilde r,\theta,\varphi)$ \end{minipage}} \\
    $\downarrow$ space slice & & $\downarrow$ space slice & & $\downarrow$ space slice \\
      \boxed{\begin{minipage}{3cm} Flamm's paraboloid, $t'$: clock of an observer at infinity comoving with the black hole \end{minipage}} & & \boxed{\begin{minipage}{3cm} plane, $t$: clock of a free falling observer \end{minipage}} & & \boxed{\begin{minipage}{3cm} plane, $t$: clock of a Galilean observer \end{minipage}} \\
      $\downarrow$ fuzzyfication & & $\downarrow$ fuzzyfication & & $\downarrow$ fuzzyfication \\
      \boxed{\begin{minipage}{3cm} fuzzy Flamm's paraboloid, $|\Lambda'(\alpha)\rrangle = \sum_{n,\varsigma} c_{n\varsigma}(\alpha)|\lambda_{n\varsigma}(\alpha)\rrangle$ \cite{Viennot4}, $t'$: clock of an observer at infinity comoving with the black hole \end{minipage}} & $\times$ & \boxed{\begin{minipage}{3cm} fuzzy plane, $|\Lambda(\alpha)\rrangle = e^{2\imath \sqrt{r_Sr}}|0\rangle \otimes |re^{\imath \theta} \rangle$, $t$: clock of a free falling observer \end{minipage}} & $\xrightarrow{\tilde r=f(r,t)}$ & \boxed{\begin{minipage}{3cm} fuzzy plane, $|\Lambda(\alpha)\rrangle = |0\rangle \otimes |\tilde re^{\imath \theta}\rangle$, $t$: clock of a Galilean observer \end{minipage}}
  \end{tabular}
  \end{center}
  with $g(t',r)=t'-r_S(-2\sqrt{r/r_S} + \ln\frac{\sqrt{r/r_S}+1}{\sqrt{r/r_S}-1})$. This is due to the fact that the only possible changes of coordinates in $\mathfrak M$ are the external gauge changes $\tilde X^i = {J^i}_j X^j$ to have a transformation inner to $\mathfrak X$.
\end{example}

\subsubsection{Geodesics}
The geodesics for $\gamma$ ($\ddot s^a + \breve \Gamma^a_{bc} \dot s^b \dot s^s = 0$, with $\breve \Gamma^a_{bc} = e^a_i \frac{\partial^2 x^i}{\partial s^b \partial s^c}$), i.e. the geodesics minimising the length on $M_\Lambda$ measured with $\gamma$, are the path minimising the cumulated energy uncertainty since $d\ell^2_s = \omega_{x(s+ds)}(\pmb{d\ell^2}_{x(s)})$ is the instantaneous energy uncertainty when the measure presents a small misalignment $ds$ with respect to the probe. The strongly adiabatic dynamics following minimising geodesics on $M_\Lambda$ is then such that the energy uncertainty due to the slipping is minimal.\\

Now we want consider the geodesics of the eigen spacetime $\mathcal M_\Lambda$. The tetrads of $\mathcal M_\Lambda$ are
\begin{eqnarray}
  e^i_a & = & \frac{\partial x^i}{\partial s^a} \label{tetrad1} \\
  e^i_0 & = & \gamma^{ab}A_a \frac{\partial x^i}{\partial s^b} \\
  e^0_a & = & 0 \\
  e^0_0 & = & 1 \label{tetrad4}
\end{eqnarray}
(the inverse tetrads being $e^a_i = \gamma^{ad} \delta_{ij} \frac{\partial x^j}{\partial s^d}$, $e^0_i=0$, $e^0_0=1$ and $e^a_0 = - \gamma^{ab} A_b$). This defines the following Christoffel symbols:
\begin{eqnarray}
  \Gamma^\alpha_{\beta \gamma} & = & 0 \\
  \Gamma^a_{b0} & = & e^a_i \partial_b e^i_0 \\
  \Gamma^0_{bc} & = & 0 \\
  \Gamma^a_{bc} & = & e^a_i \partial_b e^i_c = \breve \Gamma^a_{bc}
\end{eqnarray}
The geometry of $\mathcal M_\Lambda$ is not torsion free due to $\Gamma^a_{b0} \not= \Gamma^a_{0b} = 0$. Let $\breve \Gamma^a_{bc} = \frac{\partial^2 x^i}{\partial s^b \partial s^c} \delta_{ij}\gamma^{ad}\frac{\partial x^j}{\partial s^d}$ and $\breve \Gamma^0_{b0}= \breve \Gamma^a_{b0} = \breve \Gamma^0_{bc} = \breve \Gamma^0_{00} = 0$ be the Christoffel symbols induced by $dt^2-\gamma_{ab}ds^ads^b$ (the geometry without the geometric phase generator effects). The contorsion of $\mathcal M_\Lambda$ is $\kappa^a_{b0}  =  \Gamma^a_{b0} - \breve \Gamma^a_{b0} =  e^a_i \partial_b e^i_0$. The autoparallel geodesics are then defined by
\begin{eqnarray}
  & & \ddot s^{\alpha} + \Gamma^\alpha_{\beta \gamma} \dot s^\beta \dot s^\gamma = 0 \\
  & \iff & \ddot s^\alpha + \breve \Gamma^\alpha_{\beta\gamma} \dot s^\beta \dot s^\gamma = -\kappa^\alpha_{\beta\gamma} \dot s^\beta \dot s^\gamma
\end{eqnarray}
We can then say that the probe following an autoparallel geodesic is deviated from the extremal geodesics by a contorsional force. $A = \delta_{ij} e^i_0 dx^i$ (with $dx^i = \frac{\partial x^i}{\partial s^a} ds^a$, we recall that $d$ stands for the derivative of $M_\Lambda$) and then $F = \delta_{ij} de^i_0 \wedge dx^j$, more precisely $F_{ab} = \delta_{ij} \partial_{[a} e^i_0 \frac{\partial x^j}{\partial s^{b]}}$. But $\gamma_{ac} \kappa^c_{b0} = \delta_{ij} \partial_b e^i_0 \frac{\partial x^j}{\partial s^a}$. It follows that:
  \begin{equation}
    F_{ab} = \gamma_{[bc} \kappa^c_{a]0}
  \end{equation}
  The usual interpretation of the strong adiabatic transport is that this one is equivalent to the transport of a charged particle in $M_\Lambda$ where lives a magnetic field $F$ of potential $A$ \cite{Shapere, Bohm}. The strong adiabatic transport formula being $|\Psi \rrangle = e^{-\imath \int_{\mathscr C} A} |\Lambda \rrangle = e^{\imath S_A}|\Lambda \rrangle$ where $S_A = -\int_0^t A_a \cdot \frac{ds^a}{dt}$ is the classical action of the interaction of a particle (of negative unit charge) with a magnetic field, the adiabatic approximation can be seen as a kind of a semi-classical approximation (by analogy with the WKB ansatz for a wave function $\psi = e^{\imath S} R$). We can then interpret $F$ as a magnetic field in $M_\Lambda$, and so the contorsional force $f^a_\kappa = - \kappa^a_{b0} \dot s^b \dot t$ is a Laplace-like force acting on the virtual particle representing the probe. The contorsion deviates the geodesics as the magnetic Laplace force deviates the charged particle trajectories.

\begin{prop}
  If $t\mapsto x(s(t)) \in M_\Lambda$ is a geodesic for the strong adiabatic transport on $\mathcal M_\Lambda$ ($\ddot s^a + \breve \Gamma^a_{bc} \dot s^b \dot s^c + q \kappa^a_{b0} \dot s^b=0$), then the linear momentum $p = (\dot x^i + q \gamma^{ab} \frac{\partial x^i}{\partial s^a} A_b) \partial_i \in T_xM_\Lambda$ is conserved (with $q = \frac{dt}{d\tau}$ a constant).
\end{prop}
\begin{proof}
  The geodesic equation for $t$ being $\ddot t = 0$, we can set $t=q \tau$.
  \begin{equation}
    \ddot x^i = \frac{d}{d\tau}\left(\frac{\partial x^i}{\partial s^b} \dot s^b\right) = \frac{\partial x^i}{\partial s^b} \ddot s^b + \frac{\partial^2 x^i}{\partial s^b \partial s^c} \dot s^b \dot s^c
  \end{equation}
  It follows that
  \begin{eqnarray}
    \delta_{ij} \gamma^{ad} \frac{\partial x^j}{\partial s^d} \ddot x^i & = & \delta_{ij}\gamma^{ad}\frac{\partial x^i}{\partial s^b}\frac{\partial x^j}{\partial s^d} \ddot s^b + \delta_{ij} \gamma^{ad} \frac{\partial x^j}{\partial s^d} \frac{\partial^2 x^i}{\partial s^b \partial s^c} \dot s^b \dot s^c \\
      & = & \ddot s^a + \breve \Gamma^a_{bc} \dot s^b \dot s^c
  \end{eqnarray}
  We have then
  \begin{eqnarray}
    & & \ddot s^a + \breve \Gamma^a_{bc} \dot s^b \dot s^c + \kappa^a_{b0} \dot s^b q = 0 \\
    & \iff & \delta_{ij} \gamma^{ad} \frac{\partial x^i}{\partial s^d} \left(\ddot x^i + \frac{\partial}{\partial s^b}\left(\gamma^{ce} A_c \frac{\partial x^i}{\partial s^e} \right) \dot s^b q\right) = 0 \\
    & \iff & \pi_{xi} \frac{d}{d\tau}\left(\dot x^i + q \gamma^{ab}A_a \frac{\partial x^i}{\partial s^b}\right) = 0
  \end{eqnarray}
  Let $p^i = \dot x^i + q \gamma^{ab}A_a \frac{\partial x^i}{\partial s^b}$, we have $- p^i\sigma_i = \frac{d\slashed D_x}{dt} + q \gamma^{ab}A_a \frac{\partial \slashed D_x}{\partial s^b}$. But $\slashed D_x|\Lambda\rrangle = 0 \Rightarrow (\partial_b \slashed D_x)|\Lambda \rrangle = \slashed D_x \partial_b|\Lambda \rrangle = 0$ and then $\llangle \Lambda|(\partial_b \slashed D_x)|\Lambda \rrangle=0$. It follows that $p^i \llangle \Lambda|\sigma_i|\Lambda \rrangle = 0 \iff \vec p \cdot \vec n_x = 0$, and so $p^i \partial_i \in T_xM_\Lambda$, inducing than $\pi_x p = p$ ($\pi_x$ being the orthogonal projection onto $T_xM_\Lambda$).
\end{proof}
The conservation of $p^i = \dot x^i +qA^i$ is in accordance with the interpretation of the contorsional force as a Laplace-like force. It is equivalent to the conservation of $\frac{d\slashed D_x}{d\tau} + q \gamma^{ab}A_a \frac{\partial \slashed D_x}{\partial s^b}$.

\subsection{The weak adiabatic regime}
\subsubsection{Lorentz connection}
The inner product $\langle \bullet|\bullet \rangle_*$ defines another natural gauge potential on $M_\Lambda$:
\begin{equation}
  \mathfrak A \rho_\Lambda =  -\imath \langle \Lambda|d|\Lambda \rangle_* = -\imath \tr_{\mathcal H} |d\Lambda\rrangle\llangle \Lambda| \label{nonabAeq}
\end{equation}
$\rho_\Lambda = \tr_{\mathcal H}|\Lambda\rrangle \llangle \Lambda|$ being the $C^*$-norm of the quasicoherent state, and $d$ denoting the differential of $M_\Lambda$. Note that $\tr(\rho_\Lambda \mathfrak A) = A$, $A$ is the statistical mean value of $\mathfrak A$ in the mixed state $\rho_\Lambda$. We can then see $\mathfrak A$ as a gauge potential random variable depending on the statistical distribution of quantum states of $\mathbb C^2$ defined by the density matrix $\rho_\Lambda$.\\
$\mathfrak A$ is in general not self-adjoint but $\tr(\rho_\Lambda(\mathfrak A^\dagger-\mathfrak A))=0$, $\mathfrak A$ is then almost surely self-adjoint for the probability law defined by $\rho_\Lambda$. We write then $\mathfrak A \in \Omega^1(M_\Lambda,\mathfrak u(2)_{a.s.})$ where $\mathfrak u(2)$ denotes the set of self-adjoint operators of $\mathbb C^2$.\\

$\mathfrak A$ is the generator of the geometric phases accumulated during a weak adiabatic transport of $|\Lambda \rrangle$ along a path $\mathscr C$ on $M_\Lambda$, i.e. a transport adiabatic with respect to the quantum degree of freedom associated with $\mathcal H$ but not with the one associated with $\mathbb C^2$. If $|\Lambda \rrangle$ is weakly nondegenerate along the path $\mathscr C$ we have (see \cite{Viennot3} and appendix A in \cite{Viennot2})
\begin{eqnarray}
  |\Psi(T) \rrangle & = & e^{-\imath \int_0^T \lambda_0(t)dt} \Ped^{-\imath \int_{\mathscr C} \mathfrak A} |\Lambda(y)\rrangle + \mathcal O(\max(\epsilon,1/T)) \\
  & & \qquad \qquad \qquad \text{if } |\Psi(0)\rrangle = |\Lambda(x)\rrangle \nonumber
\end{eqnarray}
where $T$ is the total duration and $\epsilon$ is a perturbative magnitude of the transitions from the quasicoherent state to the other eigenstates induced by $\left(\Ped^{-\imath \int_{\mathscr C} \mathfrak A}\right)^{-1} \slashed D_x \Ped^{-\imath \int_{\mathscr C} \mathfrak A}$. $\Ped$ denotes the path anti-ordered exponential ($\frac{d}{ds} \Ped^{-\imath \int \mathfrak A} = -\imath \Ped^{-\imath \int \mathfrak A} \mathfrak A_a \dot s^a$ for $s$ a curvilinear coordinate along $\mathscr C$, $\dot s^a \equiv \frac{ds^a}{ds}$). The dynamical phase $\lambda_0(t) = \llangle \Lambda(x(t))|\left(\Ted^{-\imath \int_0^t \mathfrak A_a \dot s^a dt}\right)^{-1} \slashed D_{x(t)} \Ted^{-\imath \int_0^t \mathfrak A_a \dot s^a dt} |\Lambda(x(t))\rrangle$ is zero if $|\Lambda \rrangle$ is separable \cite{Viennot2}. The condition of parallel transport associated with the weak adiabatic approximation introduced in \cite{Viennot1} can be rewritten as
\begin{equation}
  \langle \Psi | \dot \Psi \rangle_* = 0
\end{equation}

\begin{prop}
  $\Ped^{-\imath \int_{\mathscr C} \mathfrak A} \in U(2)_{a.s.}$
\end{prop}
\begin{proof}
  Let $u \equiv \Ped^{-\imath \int_{\mathscr C} \mathfrak A}$.
  \begin{eqnarray}
    \frac{d}{ds} (u\rho_\Lambda u^\dagger) & = & \dot u \rho_\Lambda u^\dagger + u \dot \rho_\Lambda u^\dagger + u \rho_\Lambda \dot u^\dagger\\
    & = & u (-\imath \mathfrak A_a \dot s^a \rho_\Lambda + \dot \rho_\Lambda + \imath \rho_\Lambda \mathfrak A^\dagger_a \dot s^a) u^\dagger
  \end{eqnarray}
  But
  \begin{eqnarray}
    \dot \rho_\Lambda & = & \tr_{\mathcal H} |\dot \Lambda \rrangle \llangle \Lambda| + \tr_{\mathcal H} |\Lambda \rrangle \llangle \dot \Lambda| \\
    & = & \imath \mathfrak A_a \dot s^a \rho_\Lambda - \imath \rho_\Lambda \mathfrak A^\dagger_a \dot s^a
  \end{eqnarray}
  and then
  \begin{eqnarray}
    \frac{d}{ds} (u\rho_\Lambda u^\dagger) = 0 & \Rightarrow & u(s)\rho_\Lambda(s) u^\dagger(s) = \rho_\Lambda(0) \\
    & \Rightarrow & \tr \left(\rho_\Lambda(s) u^\dagger(s) u(s) \right) = 1 \\
    & \Rightarrow & u(s) \in U(2)_{a.s.}
  \end{eqnarray}
\end{proof}

\begin{prop}\label{gauge2}
  Under an inner gauge transformation $(h,u) \in U(1) \times U_x \subset U(1) \times SU(2)_{a.s.}$, the gauge potential becomes:
  \begin{equation}
    \tilde \mathfrak A  =  u \mathfrak A u^{-1} -\imath d\ln h - \imath du u^{-1} 
  \end{equation}
  Under a local external gauge transformation $J \in \underline{SO(3)}_{M_\Lambda}$, the gauge potential becomes:
  \begin{equation}
    \hat \mathfrak A  =  u_J^{-1} J^* \mathfrak A u_J + \imath u_J^{-1}du_J
  \end{equation}
  where $u_J \in SU(2)$ is the transformations associated with $J$ viewed as a linear map of $\mathbb R^3$: $(Jy)^i = {J^i}_j y^j$. $J^* : \Omega^1\mathbb R^3 \to \Omega^1\mathbb R^3$ is defined by $(J^* k)(x) = k_a(J^{-1}x) ({(J^{-1})^i}_j dx^j - {(J^{-1}dJJ^{-1})^i}_j x^j)$, $\forall k \in \Omega^1\mathbb R^3$.
\end{prop}
\begin{proof}
  $|\tilde \Lambda \rrangle = hu|\Lambda \rrangle \Rightarrow d|\tilde \Lambda \rrangle = hud|\Lambda \rrangle + dh \times u|\Lambda \rrangle + du \times h|\Lambda \rrangle$. So $-\imath |d\tilde \Lambda \rrangle \llangle \tilde \Lambda| = u (-\imath |d\Lambda\rrangle \llangle \Lambda|) u^{-1} -\imath h^{-1}dh \times u|\Lambda\rrangle\llangle \Lambda|u^{-1} -\imath du u^{-1} \times u|\Lambda\rrangle \llangle \Lambda|u^{-1}$. By taking the partial trace over $\mathcal H$ we find the formula for $\tilde \mathfrak A$ multiplied on the right by $\tilde \rho_\Lambda = u\rho_\Lambda u^{-1} = \tr_{\mathcal H}|\tilde \Lambda \rrangle \llangle \tilde \Lambda|$.\\
  We recall that we have $|\hat \Lambda(x) \rrangle = u_J^{-1} |\Lambda(J^{-1}x)\rrangle$. Then $-\imath |d\hat \Lambda \rrangle \llangle \hat \Lambda| = \imath u_J^{-1} du_J u_J^{-1}|\Lambda \rrangle \llangle \Lambda|u_J -\imath u_J^{-1} J^* |d\Lambda \rrangle \llangle \Lambda| u_J$ since by construction $df(J^{-1}x) = (J^* df)(x)$ for any function $f$. By taking the partial trace on $\mathcal H$, we find the formula for $\hat \mathfrak A$ multiplied on right by $\hat \rho_\Lambda = u_J^{-1} \rho_\Lambda u_J$.
\end{proof}

$\mathfrak A$ defines the adiabatic fake curvature $\mathfrak F = d\mathfrak A - \imath \mathfrak A \wedge \mathfrak A \in \Omega^2(M_\Lambda,\mathfrak{u}(2)_{a.s.})$ \cite{Viennot1}.\\

\begin{example}{2}{Fuzzy surface plots}
 $A$ and $F$ are fields onto the eigenmanifold $M_\Lambda$ spanned by the mean values of the coordinate observables $X^i$. These fields onto the mean value of $\mathfrak M$ are also mean values of the quantum observables $\mathfrak A$ and $\mathfrak F$ in the eigen density matrix $\rho_\Lambda$ ($A = \tr(\rho_\Lambda \mathfrak A)$ and $F = \tr(\rho_\Lambda \mathfrak F)$). These nonabelian gauge fields are observables for the spin degree of freedom. By definition $\mathfrak A \rho_\Lambda = -\imath \tr_{\mathcal H}|d\Lambda \rrangle \llangle \Lambda|$, and so for the fuzzy plane we have:
  \begin{equation}
    \mathfrak A = A \left(\begin{array}{cc} 1 & 0 \\ 0 & 0 \end{array} \right)
  \end{equation}
  where $\rho_\Lambda = \left(\begin{array}{cc} 1 & 0 \\ 0 & 0 \end{array} \right) = |0\rangle \langle 0|$ is the projection onto the first state of the canonical basis of $\mathbb C^2$ ($\rho_\Lambda$ is a pure state). $A = -\imath \frac{\bar \alpha d\alpha - \alpha d\bar \alpha}{2}$. Note that $\mathfrak A$ is not single defined, since $\mathfrak A' = \mathfrak A + \xi$ is also solution of eq. (\ref{nonabAeq}) if $|0\rangle \in \ker \xi$. We can then think $\xi$ as a gauge change of a third kind (the two first ones being the internal and external gauge changes). But by definition $\xi$ is almost surely zero, $\tr(\rho_\Lambda \xi)=0$ and plays no fundamental role. We can then set $\xi=0$ without lost of generality. The interpretation of the observable $\mathfrak A$ is obvious, the magnetic potential is $A$ for the spin state up $|0\rangle$ (eigenvector of $\mathfrak A$) and the magnetic potential is set to zero for the spin state down $|1\rangle$, since $|0\rangle$ is the only one spin state appearing in the quasicoherent state of the fuzzy plane. For the fuzzy elliptic paraboloid we have
  \begin{equation}
    \mathfrak A \rho_\Lambda = A \left(\begin{array}{cc} 1 & -\epsilon \bar \alpha \\ -\epsilon \alpha & 0 \end{array} \right) + \imath \epsilon d\alpha \left(\begin{array}{cc} 0 & 0 \\ 1 & 0 \end{array} \right) + \mathcal O(\epsilon^2)
  \end{equation}
  where the pure eigenstate $\rho_\Lambda = \left(\begin{array}{cc} 1 & -\epsilon \bar \alpha \\ -\epsilon \alpha & 0 \end{array} \right) + \mathcal O(\epsilon^2) = |O_\alpha \rangle \langle O_\alpha|$ is the projection onto $|O_\alpha \rangle = |0\rangle -\epsilon \alpha |1\rangle$. It follows that we can set
  \begin{equation}
    \mathfrak A = A \left(\begin{array}{cc} 1 & -\epsilon \bar \alpha \\ -\epsilon \alpha & 0 \end{array} \right) -\imath \frac{d\alpha}{\alpha} \left(\begin{array}{cc} 0 & 0 \\ 0 & 1 \end{array} \right) + \mathcal O(\epsilon^2)
  \end{equation}
  Note that the second member is almost surely zero at the perturbative first order $\tr(\rho_\Lambda |1\rangle \langle 1|) = \mathcal O(\epsilon^2)$. For the fuzzy hyperbolic paraboloid we have
    \begin{equation}
    \mathfrak A \rho_\Lambda = A \left(\begin{array}{cc} 1 & -\epsilon \alpha \\ -\epsilon \bar \alpha & 0 \end{array} \right) + \imath \frac{\epsilon}{2} \left(\begin{array}{cc} 0 & d\alpha \\ d\bar \alpha & 0 \end{array} \right) + \mathcal O(\epsilon^2)
  \end{equation}
    with $\rho_\Lambda = |O_{\bar \alpha}\rangle \langle O_{\bar \alpha}|$ the projection onto $|O_{\bar \alpha} \rangle = |0\rangle -\epsilon \bar \alpha |1\rangle$. It follows that we can set
    \begin{equation}
    \mathfrak A = A \left(\begin{array}{cc} 1 & -\epsilon \alpha \\ -\epsilon \bar \alpha & 0 \end{array} \right) - \frac{\imath}{2} \left(\begin{array}{cc} \frac{d\alpha}{\alpha} & \frac{d\alpha}{\epsilon|\alpha|^2} \\ 0 & \frac{d\bar \alpha}{\bar \alpha} \end{array} \right) + \mathcal O(\epsilon^2)
  \end{equation}
    Anew the second member is almost surely zero.
\end{example}

For the geometric viewpoint, $\mathfrak A$ defines a Lorentz connection onto the classical spacetime $\mathcal M_\Lambda$:
\begin{propo}
  Let $\Omega^{\mu \nu} \in \Omega^1(M_\Lambda,\mathbb R)$ be the 1-form valued antisymmetric tensor defined by
 \begin{eqnarray}
      \Omega^{jk} & = & -{\varepsilon^{jk}}_i \Re \tr(\sigma^i \mathfrak A) \\
      \Omega^{i0} & = & \Im \tr(\sigma^i \mathfrak A)
 \end{eqnarray}
then $\Omega$ is a Lorentz connection onto $\mathbb R \times M_\Lambda$. It follows that ${R^{\mu \nu}}_{ab}$ is the Riemann curvature tensor of $\mathbb R \times M_\Lambda$ associated with $\Omega$ ($a,b$ standing for indices of local curvilinear coordinates on $M_\Lambda$), with ${R^{jk}}_{ab} = -{\varepsilon^{jk}}_i \Re\tr(\sigma^i \mathfrak F_{ab})$ and ${R^{i0}}_{ab} = \Im\tr(\sigma^i \mathfrak F_{ab})$. The Ricci curvature tensor ${R^\mu}_b = {R^{a\mu}}_{ab} = \gamma^{ac} \delta_{ik} \frac{\partial x^k}{\partial s^c} {R^{i\mu}}_{ab}$ is then ${R^j}_b = -{\varepsilon^{jk}}_i \Re\tr(\sigma_k {\mathfrak F^a}_b) \frac{\partial x^i}{\partial s^a}$ and ${R^0}_b = \Im\tr(\sigma_i {\mathfrak F^a}_b) \frac{\partial x^i}{\partial s^a}$. 
\end{propo}
We can see that $\Omega$ is well a Lorentz connection by considering its transformation under frame changes (see \ref{frameOmega}). The definition of the Lorentz connection $\Omega$ by $\mathfrak A$ is consistent with the BFFS Dirac equation $\imath |\dot \Psi \rrangle = \slashed D_x |\Psi \rrangle$ at the weak adiabatic limit \cite{Viennot2}. We can show that operator valued geometric phases generated by $\mathfrak A$ induces spin rotations which are the Einstein-de Sitter spin precessions induced by $\Omega$. The definition of $\Omega$ is then consistent with the physical interpretation of $\mathfrak M$ as a string theory matrix model.\\
In quantum information theory, $\Omega$ permits to represent the qubit logical gate generated by the operator valued geometric phases induced by $\mathfrak A$ by a precession effect induced by $\Omega$. The adiabatic quantum control of the qubit in presence of the environment is then completely translated in the geometry of $M_\Lambda$.\\

As for usual gauge fields, $(A,F,\mathfrak A, \mathfrak F)$ are local data of a connection onto a complicated fibre bundle which is studied \ref{fibre}, the adiabatic regimes appearing as horizontal lifts in this one (\ref{pseudosurface}).

\subsubsection{Torsion and geodesics}
In weak adiabatic transport, the time dependent state is not proportional to $|\Lambda \rrangle$ but is $\Ped^{-\imath \int \mathfrak A} |\Lambda \rrangle$, so the energy uncertainty is not minimised during the dynamics, and the relevant geodesics are not the ones associated with $\gamma_{ab}$ (minimising geodesics) but the ones associated with the Lorentz connection $\Omega$ (autoparallel geodesics). Moreover the energy mean value $\llangle \Psi|\slashed D_x|\Psi \rrangle = \llangle \Lambda | (\Ped^{-\imath \int \mathfrak A})^\dagger \slashed D_x \Ped^{-\imath \int \mathfrak A}|\Lambda \rrangle$ is not constant except if $|\Lambda \rrangle$ is separable. The Lorentz connection with the triads define the following Christoffel symbols on $M_\Lambda$:
\begin{equation}
  \Gamma^a_{bc} = e^a_i \partial_b e^i_c + e^a_i \Omega^i_{bj} e^j_c
\end{equation}
where $e^a_i = {\Xi^a}_c \mathring e^c_i = \gamma^{ad} \delta_{ij} \frac{\partial x^j}{\partial s^d}$ where $\mathring e^c_i$ are the dual triads induced by $E^c_i$ and $\Xi$ the inverse matrix of $(\mathring e^a_i e^i_b)$. Since these Christoffel symbols do not derive directly from $\gamma_{ab}$, the geometry of $M_\Lambda$ is not torsion free. We can show after some algebra that this torsion is
\begin{eqnarray}
  T^a_{bc} & = & \frac{1}{2} \gamma^{ad} \Im \tr(\tau_d [\tau_{[c},\mathfrak A_{b]}]) \\
  & = & \gamma^{ad} {\varepsilon_{ij}}^k \Re \tr(\sigma_k \mathfrak A_{[b}) \frac{\partial x^j}{\partial s^{c]}} \frac{\partial x^i}{\partial s^d}
\end{eqnarray}
with $\tau_a \equiv \sigma_i \frac{\partial x^i}{\partial s^a}$. This connection is associated with the covariant derivatives in $\mathfrak M$, see \ref{coderiv}.\\
If $\mathfrak A$ is self-adjoint, a weak adiabatic transport along an autoparallel geodesic provides $\llangle \Lambda(x(t))|(\Ped^{-\imath \int \mathfrak A})^{-1}\sigma_i \Ped^{-\imath \int \mathfrak A}|\Lambda(x(t))\rrangle \dot x^i(0) = 0$ (see \ref{coderiv}): the rotated (by $\Ped^{-\imath \int \mathfrak A}$) normal vector remains normal to the initial tangent vector. The adiabatic transport is well autoparallel (in the general case, it is just almost surely autoparallel). Table \ref{geobeh} summarises the physical meaning of the different regimes.
\begin{table} \small
  \begin{tabular}{l|l}
    \begin{minipage}{5cm} strong adiabatic transport along any path \end{minipage} & \begin{minipage}{7cm} quantum transport minimising the uncertainties $\Delta_{\Psi(t)} X^i$ without spin precession \end{minipage} \\
    \hline
    \begin{minipage}{5cm} strong adiabatic transport along a minimising geodesic path \end{minipage} & \begin{minipage}{7cm} the length of the average path is minimised  \end{minipage}  \\
    \hline
    \begin{minipage}{5cm} strong adiabatic transport along a path such that $\omega_{x(t)}(|\pmb{\delta \ell}_{x(t+\Delta t),x(t)}|)$ is small constant \end{minipage} & \begin{minipage}{7cm} the average of the path length observable is minimised  \end{minipage} \\
    \hline
    \begin{minipage}{5cm} weak adiabatic transport along any path \end{minipage} &  \begin{minipage}{7cm} quantum transport minimising the uncertainties $\Delta_{\Psi(t)} X^i$ with spin precession \end{minipage}\\
    \hline
    \begin{minipage}{5cm} weak adiabatic transport along an autoparallel geodesic path \end{minipage} & \begin{minipage}{7cm} the spin precession is autoparallel \end{minipage}
  \end{tabular}
  \caption{\label{geobeh} Geometric behaviours with respect to the different dynamical regimes.}
\end{table}
In another viewpoint, let $\Delta_0 \slashed D = \left. \frac{d\slashed D_x}{dt} \right|_{t=0}  \Delta t$ be the energy uncertainty observable for an initial time uncertainty $\Delta t$. The autoparallel geodesics minimises the mean value of this energy uncertainty at each time (it is zero if $\mathfrak A$ is self-adjoint), see \ref{coderiv}.\\

The Lorentz connection $\Omega^{\mu \nu}_a = -{\varepsilon^{0\mu \nu}}_i \Re \tr(\sigma^i \mathfrak A_a) + (\delta^{\nu 0}\delta_i^\mu - \delta^{\mu 0}\delta_i^\nu) \Im \tr(\sigma^i \mathfrak A_a)$ both with the tetrads eq.(\ref{tetrad1}-\ref{tetrad4}) define the following Christoffel symbols on $\mathcal M_\Lambda$:
\begin{eqnarray}
  \Gamma^\alpha_{0\beta} & = & 0 \\
  \Gamma^0_{b0} & = & \Omega^0_{bj} e^j_0 \\
  \Gamma^a_{b0} & = & e^a_i \partial_b e^i_0 + e^a_0\Omega^0_{bj}e^j_0  +e^a_i \Omega^i_{b0} + e^a_i \Omega^i_{bj} e^j_0\\
  \Gamma^0_{bc} & = & \Omega^0_{bj} e^j_c \\
  \Gamma^a_{bc} & = & e^a_i \partial_b e^i_c + e^a_i \Omega^i_{bj} e^j_c + e^a_0 \Omega^0_{bj} e^j_c
\end{eqnarray}
The geometry of $\mathcal M_\Lambda$ is not torsion free, and we have
\begin{eqnarray}
  T^0_{ab} & = & \Im \tr(\sigma_i \mathfrak A_{[a}) \frac{\partial x^i}{\partial s^{b]}} \\
  T^0_{a0} & = & \Im \tr(\sigma_i \mathfrak A_a) \gamma^{bc} A_b \frac{\partial x^i}{\partial s^c} \\
  T^a_{b0} & = & \gamma^{ac} \frac{\partial x_i}{\partial s^c} (\partial_b e^i_0  + \Im \tr(\sigma^i \mathfrak A_b) + {\varepsilon^i}_{jk} \Re \tr(\sigma^k \mathfrak A_b) \gamma^{cd} A_c \frac{\partial x^j}{\partial s^d}) \nonumber \\
  & & - \gamma^{ac} \Im \tr(\sigma_j \mathfrak A_b) \gamma^{de} A_c A_d \frac{\partial x^j}{\partial s^e} \\
T^a_{bc} & = & \gamma^{ad} {\varepsilon_{ij}}^k \Re \tr(\sigma_k \mathfrak A_{[b}) \frac{\partial x^j}{\partial s^{c]}} \frac{\partial x^i}{\partial s^d} - \gamma^{ad} \Im\tr(\sigma_j \mathfrak A_{[b}) \frac{\partial x^j}{\partial s^{c]}} A_d
\end{eqnarray}

Let $\breve \Gamma^a_{bc} = \frac{\partial^2 x^i}{\partial s^b \partial s^c} \delta_{ij}\gamma^{ad}\frac{\partial x^j}{\partial s^d}$ and $\breve \Gamma^0_{b0}= \breve \Gamma^a_{b0} = \breve \Gamma^0_{bc} = \breve \Gamma^0_{00} = 0$ be the Christoffel symbols induced by $dt^2-\gamma_{ab}ds^ads^b$ (the geometry without the geometric phase generator effects), and $K^\alpha_{\beta \gamma} = \Gamma^\alpha_{\beta \gamma} - \breve \Gamma^\alpha_{\beta \gamma}$ be the contorsion of  $\mathcal M_\Lambda$. This one is
\begin{eqnarray}
  K^0_{b0} & = & \Omega^0_{bj} e^j_0 \\
  K^a_{b0} & = & e^a_i \partial_b e^i_0 + e^a_0\Omega^0_{bj}e^j_0 + e^a_i \Omega^i_{b0} + e^a_i \Omega^i_{bj} e^j_0 \\
  K^0_{bc} & = & \Omega^0_{bj} e^j_c \\
  K^a_{bc} & = & e^a_i \Omega^i_{bj} e^j_c+ e^a_0 \Omega^0_{bj} e^j_c
\end{eqnarray}

The autoparallel geodesics are defined by
\begin{eqnarray}
  & & \ddot s^{\alpha} + \Gamma^\alpha_{\beta \gamma} \dot s^\beta \dot s^\gamma = 0 \\
  & \iff & \ddot s^\alpha + \breve \Gamma^\alpha_{\beta\gamma} \dot s^\beta \dot s^\gamma = -K^\alpha_{\beta\gamma} \dot s^\beta \dot s^\gamma
\end{eqnarray}
The contorsion can be split in two terms, $\kappa^a_{b0} = e^a_i \partial_b e^i_0$ the part associated with the magnetic field $F$ and $\mathcal K^\alpha_{b \gamma} = e^\alpha_\mu \Omega^\mu_{b \nu} e^\nu_\gamma$. We recall that $\Ped^{- \imath \int_{\mathscr C} \mathfrak A^{off}} = \Ped^{ \int_0^t \Omega^{ij}_a \frac{1}{2} [\sigma_i,\sigma_j] \dot s^a dt}$ ($\mathfrak A^{off} = \mathfrak A - \tr \mathfrak A$) is responsible of a spin precession in the weak adiabatic transport:
\begin{equation}
  \rho_\Psi(t) = \Ped^{ \int_0^t \Omega^{ij}_a \frac{1}{2} [\sigma_i,\sigma_j] \dot s^a dt} \rho_\Lambda(x(t)) \left(\Ped^{ \int_0^t \Omega^{ij}_a \frac{1}{2} [\sigma_i,\sigma_j] \dot s^a dt}\right)^\dagger
\end{equation}
where $\rho_\Psi = \tr_{\mathcal H} |\Psi \rrangle \llangle \Psi|$ is the spin density matrix and $\rho_\Lambda = \tr_{\mathcal H}|\Lambda \rrangle \llangle \Lambda|$ is the eigen density matrix. The contorsional force $f_{\mathcal K}^\alpha= -\mathcal K ^\alpha_{b \gamma} \dot s^b \dot s^\gamma$ as the same interpretation than in Einstein-Cartan theory \cite{Trautman}: as the spacetime curvature and energy momentum tensor influence each other (Einstein equation), spacetime torsion and spin influence each other (Sciama-Kibble equation \cite{Trautman}). So the geometry of $\mathcal M_\Lambda$ induces a spin precession, but in return the spin modifies the geometry of $\mathcal M_\Lambda$ inducing a deviation from the extremal geodesics by the force $f_{\mathcal K}$.\\

 To summarise, onto $M_\Lambda$ we have four notions of geodesic:
  \begin{itemize}
  \item \textit{Geodesics minimising the length of the mean path (GMLM):} geodesics for the metric $\omega_{x(s+ds)}(\pmb{d\ell^2}_{x(s)})= \gamma_{ab}ds^ads^b$ for which the energy uncertainty with probe misalignment is minimised;
  \item \textit{Geodesics minimising the average length of the quantum paths (GMAL):} geodesics for the metric $\left(\omega_{x(s)}(\sqrt{\pmb{\gamma}(L_{\delta \Pi},L_{\delta \Pi})})\right)^2 = \tilde \gamma_{ab} \delta s^a \delta s^b$ (this ones being applied onto a discrete set of points along the geodesic corresponding to measures onto the quantum system), for which the average energy gap isotropic uncertainty is minimised;
  \item \textit{Strongly adiabatic autoparallel geodesics (SAAG):} geodesics including the contorsion $\kappa^a_{b0}$ corresponding to the effect of the Laplace like force induced by the Berry curvature $F$ (viewed as a magnetic field), for which energy uncertainty with probe misalignment and time lag is minimised;
  \item \textit{Weakly adiabatic autoparallel geodesics (WAAG):} geodesics including the contorsions $\kappa^a_{b0}$ and $\mathcal K^\alpha_{\beta \gamma}$ corresponding to the Einstein-Cartan spinning effect, for which the spin precession is autoparallel.
  \end{itemize}
  
\begin{example}{2}{Fuzzy surface plots}
  For the fuzzy plane, the minimising geodesics are straight lines whereas the autoparallel geodesics are circles (see \cite{Viennot2}). Under the effect of the magnetic field $F$ (normal to the plane) the geodesics are curved to induce a cyclic rotation of the charged particle. GMLM and GMAL are the same since the two distances are the same, and SAAG and WAAG are the same since $\mathfrak A = A |0\rangle \langle 0|$ (implying than $\mathcal K^\alpha_{\beta \gamma}=0$).\\
  For the fuzzy elliptic paraboloid, GMAL are slightly deviated from GMLM because the average length of quantum paths between two measures is slightly larger than the length of the mean path. SAAG are not circles but epicycles around deferents corresponding to circles of the rotational symmetry (the epicyclic moves being induced by the magnetic field $F$ and the curvature of $M_\Lambda$). WAAG are very complicated trajectories on $M_\Lambda$ difficult to interpret which seem to trap the particle in the neighbourhood of $\alpha=0$. Fig. \ref{geodesicsfig} presents numerical integrations of geodesic equations for the fuzzy elliptic paraboloid.
  \begin{figure}
    \includegraphics[width=7cm]{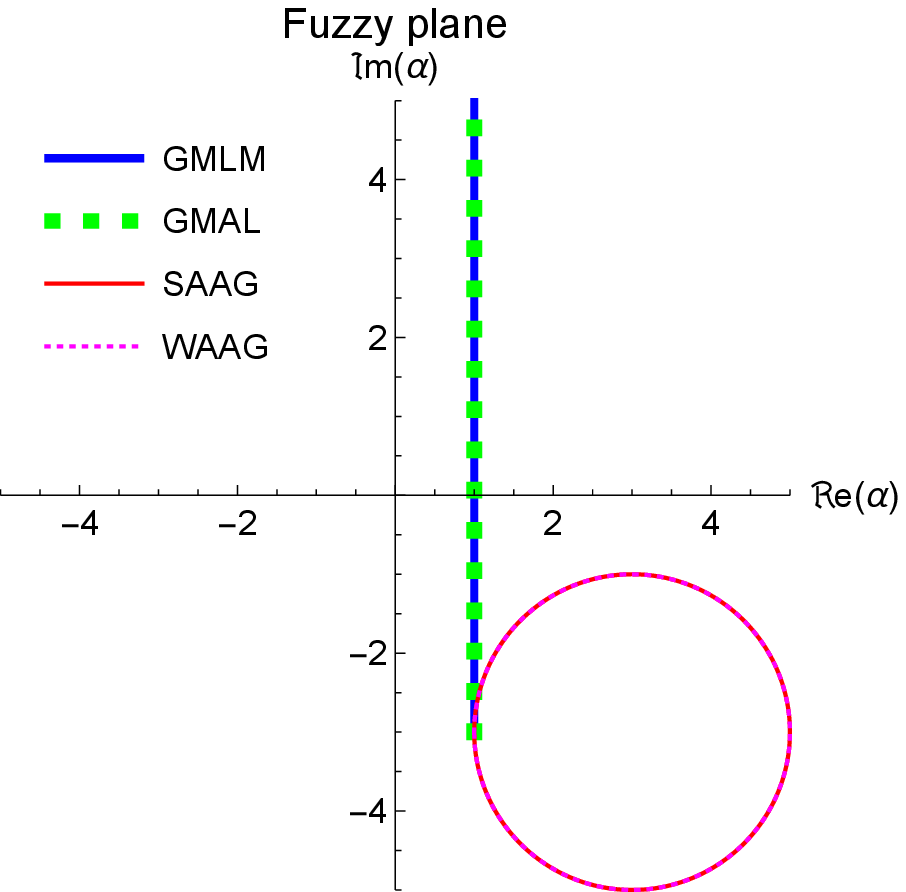} \includegraphics[width=7cm]{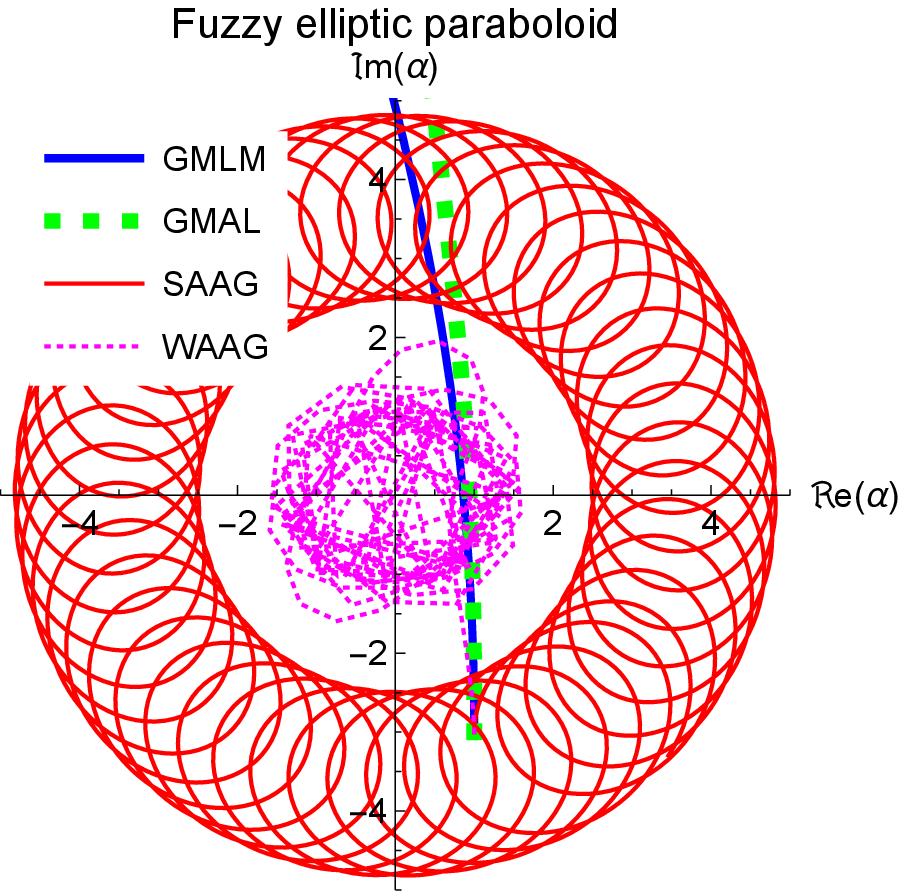} \includegraphics[width=7cm]{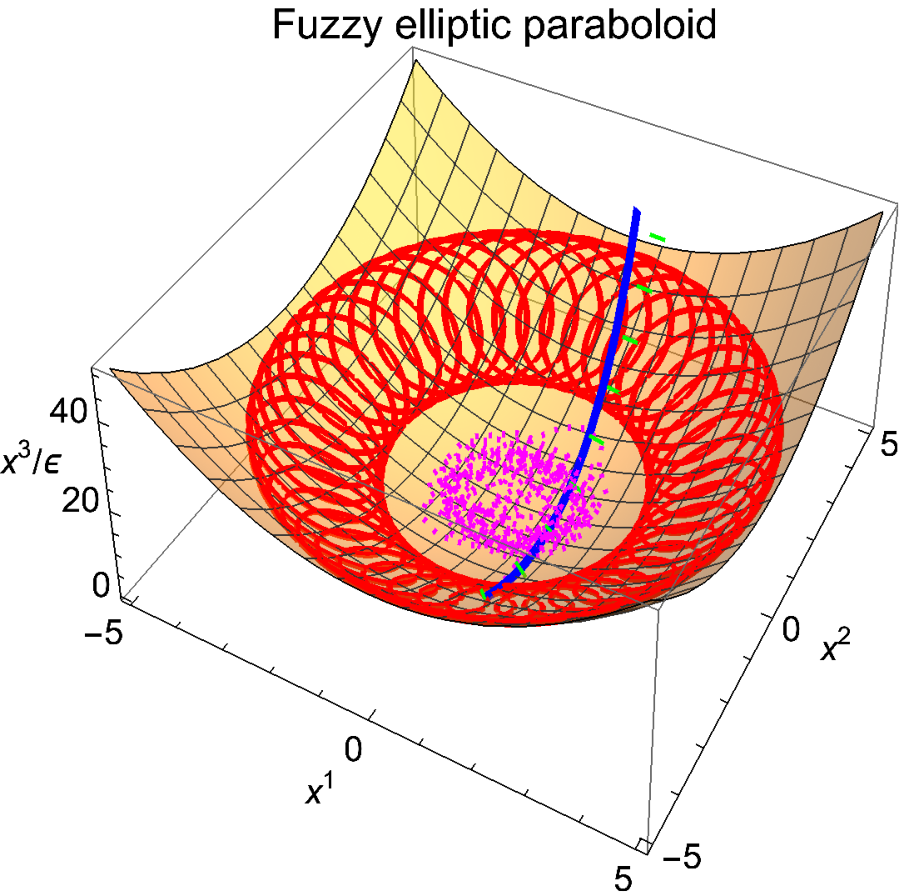}
    \caption{\label{geodesicsfig} In the parameter space spanned by $\alpha$, the geodesics for the fuzzy plane and the fuzzy elliptic paraboloid (and these ones onto $M_\Lambda$). GMLM is a geodesic minimising the length of the mean path (associated with $\gamma)$; GMAL is  geodesic minimising the average length of quantum paths between two measurements (this one is then discrete, and is associated with $\dist_{\pmb \gamma}$); SAAG is a strongly adiabatic autoparallel geodesic (including the effect of the magnetic field $F$); and WAAG is a weakly adiabatic autoparallel geodesic (including the effects of $F$ and of the torsion associated with $\Omega$). In all cases the initial conditions are $\alpha(0) = 1-3\imath$ and $\dot \alpha(0) = \imath$, with $\epsilon=0.1$.}
  \end{figure}
\end{example}

\section{Time-dependent Fuzzy spaces}\label{TDFS}
In the previous sections, we have supposed that $\mathfrak M$ was time-independent and that only the probe $x$ was moved. But the physical situations induce that the coordinate observables $X^i$ can be time-dependent. Indeed, in string theory BFSS matrix models, the D2-brane obeys to the following equation of motion (without vacuum fluctuation) \cite{BFSS}:
\begin{equation}\label{D2branemotion}
  \ddot X^i - [X_j,[X^i,X^j]] = 0
\end{equation}
with the Gauss constraint $[\dot X_i,X^i]=0$. This equation is the noncommutative version of the Einstein equation in the vacuum ($\Box_g g_{\mu \nu} + ... = 0$ in harmonic local coordinates with $\Box_g$ the spacetime Laplace-Beltrami operator -- the spacetime D'Alembertian --; $[X_j,[\bullet,X^j]]$ being the noncommutative Laplacian of $\mathfrak M$).\\
Moreover, in quantum information theory, the environmental noise operators satisfy the following Heisenberg equation:
\begin{equation}
  \dot X^i = \imath [\tilde H_{\mathcal E},X^i]
\end{equation}
where $H_{\mathcal E} \in \mathcal L(\mathcal H)$ is the Hamiltonian of the environment (and $\tilde H_{\mathcal E} = U_{\mathcal E}^\dagger H_{\mathcal E} U_{\mathcal E}$ with $\imath \dot U_{\mathcal E} = H_{\mathcal E} U_{\mathcal E}$; $\tilde H_{\mathcal E} = H_{\mathcal E}$ if it is time-independent).\\

If the time-dependence of $X^i$ is fast, we cannot apply directly the adiabatic approximations and it is needed to treat the problem with a Schr\"odinger-Floquet picture (for the periodic cases) \cite{SF} or a Schr\"odinger-Koopman picture (for the quasi-periodic and aperiodic cases) \cite{SK}. These pictures are equivalent to add a fourth coordinate observable $X^0$ (as a quantum time coordinate observable) with a probe space $\mathbb R^4$ endowed with the Minkowski metric \cite{Viennot2}. This method has been discussed in a previous paper. In this paper we want to focus on the case where $X^i$ depends slowly on the time permitting to use directly the adiabatic approximations.\\

With a time-dependent fuzzy space, the eigenequation becomes:
\begin{equation}
  \sigma_i \otimes (X^i(t)-x^i)|\Lambda(x,t)\rrangle = 0
\end{equation}
for $x^i \in M_\Lambda(t)$, the eigenmanifold being time dependent. By using the perturbation theory \ref{perturbation} with $\delta X^i = \dot X^i dt$ we have $M_\Lambda(t) = \{x(t) \in \mathbb R^3 \}$ where $t \mapsto x(t)$ are the solutions of the following differential equation:
\begin{equation}\label{eqFlow1}
  \dot x^i = \langle \Omega_{x(t),t}|\dot X^i|\Omega_{x(t),t}\rangle
\end{equation}
if $|\Lambda(x,t)\rrangle = |O_{x,t} \rangle \otimes |\Omega_{x,t} \rangle$ is separable, or else
\begin{equation}\label{eqFlow2}
 \dot x^i = \llangle \Lambda(x(t),t)|\sigma_j \otimes \dot X^j|\Lambda(x(t),t)\rrangle \frac{p_{x(t)}^i}{\vec n_{x(t)} \cdot \vec p_{x(t)}}
\end{equation}
with $\vec n_{x(t)} = \llangle \Lambda(x(t),t)|\sigma^i|\Lambda(x(t),t \rrangle \partial_i \in N_{x(t)}M_\Lambda(t)$ and $\vec p_{x(t)} = \frac{\llangle \Lambda(x(t),t)|\dot X^i|\Lambda(x(t),t)\rrangle}{\|\llangle \vec {\dot X} \rrangle(t)\|} \partial_i$.  The set of initial conditions $M_\Lambda(0)$ is such that $\ker (\sigma_i \otimes (X^i(0)-x^i(0))) \not= \{0\}$. Moreover we have $\forall x(t) \in M_\Lambda(t)$:
\begin{eqnarray}
  & & \left. \frac{\partial}{\partial t} |\Lambda(x,t)\rrangle \right|_{x=x(t)} \nonumber \\
  & = & - \sum_{n>0} \frac{\llangle \lambda_n(x(t),t)|\sigma_i \otimes (\dot X^i-\dot x^i)|\Lambda(x(t),t)\rrangle}{\lambda_n(x(t),t)} |\lambda_n(x(t),t) \rrangle
\end{eqnarray}
where $\slashed D_x(t)|\lambda_n(x,t)\rrangle = \lambda_n(x,t)|\lambda_n(x,t)\rrangle$ (with $\lambda_0(x,t)=0$).\\
Let $(t,s^1,s^1)$ be a local coordinate system onto $\mathcal M_\Lambda = \bigsqcup_t (t,M_\Lambda(t))$. The metric of the leaf $M_\Lambda(t)$ satisfies
\begin{equation}
  \dot \gamma_{ab} = \delta_{ij} \frac{\partial x^i}{s^{(a}} \frac{\partial \dot x^j}{\partial s^{b)}}
\end{equation}

The main difference with the previous sections is that now the gauge potentials are forms of $\mathcal M_\Lambda$
\begin{eqnarray}
  A & = & -\imath (\llangle \Lambda (x,t)|d|\Lambda(x,t) \rrangle \nonumber \\
  & & \quad + \llangle \Lambda(x,t)|\partial_t|\Lambda(x,t) \rrangle dt) \in \Omega^1(\mathcal M_\Lambda,\mathbb R) \\
  \mathfrak A \rho_\Lambda & = & -\imath (\langle \Lambda (x,t)|d|\Lambda(x,t) \rangle_* \nonumber \\
  & & \quad + \langle \Lambda(x,t)|\partial_t|\Lambda(x,t) \rangle_* dt) \in \Omega^1(\mathcal M_\Lambda,\mathfrak{u}(2)_{a.s.})
\end{eqnarray}
and we must redefine one tetrad by
\begin{equation}
  e^0_0 = 1 + A_0
\end{equation}
The metric of $\mathcal M_\Lambda$ is then now
\begin{equation}
  d\tau^2 = (1+A_0)^2 dt^2 - (\gamma^{ac}A_c dt+ds^a)(\gamma^{bd}A_d dt+ds^b) \gamma_{ad}
\end{equation}
where $A$ and $\gamma$ depends on $s$ and $t$. $\vec A = \gamma^{ab} A_b \partial_a$ is still the shift vector of the foliation of $\mathcal M_\Lambda$, but now $1+A_0(x,t)$ is its laps function.\\
The Lorentz connection has now new components $\Omega^{\mu \nu}_0 = {\varepsilon^{0\mu \nu}}_{i} \Re\tr(\sigma^i \mathfrak A_0) + (\delta^{\nu 0}\delta^\mu_i - \delta^{\mu0}\delta^\nu_i) \Im \tr (\sigma^i \mathfrak A_0)$. With the modification of the inverse tetrads $e^a_0 = -\frac{\gamma^{ab}A_b}{1+A_0}$ and $(e^{-1})^0_0 = (e^0_0)^{-1} = \frac{1}{1+A_0}$, this induces these new Christoffel symbols:
\begin{eqnarray}
  \Gamma^0_{b0} & = & \partial_b \ln e^0_0 + (e^0_0)^{-1} \Omega^0_{bj} e^j_0 \\
  \Gamma^a_{b0} & = & e^a_0 \partial_b e^0_0 + e^a_i \partial_b e^i_0 + e^a_0\Omega^0_{bj}e^j_0  +e^a_i \Omega^i_{b0} e^0_0 + e^a_i \Omega^i_{bj} e^j_0\\
  \Gamma^0_{bc} & = & (e^0_0)^{-1} \Omega^0_{bj} e^j_c \\
  \Gamma^a_{bc} & = & e^a_i \partial_b e^i_c + e^a_i \Omega^i_{bj} e^j_c + e^a_0 \Omega^0_{bj} e^j_c \\
  \Gamma^a_{0c} & = & e^a_i \partial_0 e^i_c + e^a_0 \Omega^0_{0j} e^j_c + e^a_i \Omega^i_{0j} e^j_c \\
  \Gamma^0_{0c} & = & (e^0_0)^{-1} \Omega^0_{0j} e^j_c \\
  \Gamma^0_{00} & = & \partial_0 \ln e^0_0 + (e^0_0)^{-1} \Omega^0_{0j} e^j_0
\end{eqnarray}

\begin{prop}\label{dynFuzzy}
  In the case where $\dot X^i = \imath[\tilde H_{\mathcal E},X^i] \iff X^i(t) = U_{\mathcal E}^\dagger (t) X^i(0) U_{\mathcal E}(t)$, we have:
  \begin{eqnarray}
    M_\Lambda(t) & = & M_\Lambda(0) \\
    |\Lambda(x,t) \rrangle & = & U_{\mathcal E}^\dagger(t)|\Lambda(x,0)\rrangle \\
    A & = & A_a(x,0) ds^a + \llangle \Lambda(x,0)|H_{\mathcal E}|\Lambda(x,0)\rrangle dt \\
    \mathfrak A & = & \mathfrak A_a(x,0) ds^a + \tr_{\mathcal H} (|\Lambda(x,0)\rrangle \llangle \Lambda(x,0)| H_{\mathcal E}) \rho_\Lambda(x,0)^{-1} dt \\
    \mathbf A_{\Dis_t} & = & U_{\mathcal E}^\dagger(t) \mathbf A_{\Dis_0} U_{\mathcal E}(t) + \imath U_{\mathcal E}^\dagger(t) \dnc U_{\mathcal E}(t) +\imath \Dis_t \dnc U_{\mathcal E}^\dagger(t) U_{\mathcal E}(t) \Dis_t
  \end{eqnarray}
  ($\rho_\Lambda^{-1}$ standing for a pseudo-inverse), where $\Dis_t = U_{\mathcal E}^\dagger(t) \Dis_0 U_{\mathcal E}(t)$. The flow $\varphi^t$ defined by eq.(\ref{eqFlow1},\ref{eqFlow2}) is just a time-dependent diffeomorphism of $M_\Lambda(0)$ ($\varphi^t \in \Diff M_\Lambda(0)$).
\end{prop}

\begin{proof}
  $\slashed D_x(t) = U_{\mathcal E}^\dagger(t) \slashed D_x(0) U_{\mathcal E}(t) \Rightarrow \slashed D_x(0) U_{\mathcal E}(t) |\Lambda(x,t) \rrangle = 0 \Rightarrow |\Lambda(x,0) \rrangle = U_{\mathcal E}(t)|\Lambda(x,t) \rrangle$. It follows that $\partial_t|\Lambda(x,t)\rrangle = -U_{\mathcal E}^\dagger(t) \dot U_{\mathcal E}(t) U_{\mathcal E}^\dagger(t)|\Lambda(x,0) \rrangle = - U_{\mathcal E}^\dagger(t) H_{\mathcal E}|\Lambda(x,0)\rrangle$.\\
  $u_{yx}\otimes \Dis_0(y,x) |\Lambda(x,0)\rrangle = |\Lambda(y,0)\rrangle \Rightarrow u_{yx} \otimes U_{\mathcal E}^\dagger \Dis_0(y,x) U_{\mathcal E}|\Lambda(x,t)\rrangle = |\Lambda(y,t)\rrangle$. It follows $\Dis_t(y,x) = U_{\mathcal E}^\dagger \Dis_0(y,x) U_{\mathcal E}$ and then $\dnc \Dis_t = \dnc U_{\mathcal E}^\dagger \Dis_0 U_{\mathcal E} + U_{\mathcal E}^\dagger \dnc \Dis_0 U_{\mathcal E} + U_{\mathcal E}^\dagger D_0 \dnc U_{\mathcal E}$.
\end{proof}

In the case where the quasicoherent state is separable, the flow is similar to the Ehrenfest theorem: $\frac{d}{dt} \langle \Omega_{x(t),t}|X^i|\Omega_{x(t),t}\rangle= \imath \langle \Omega_{x(t),t}|[\tilde H_{\mathcal E},X^i]|\Omega_{x(t),t}\rangle$.\\

Another simple time-dependent case arises with time-dependent external gauge changes:
\begin{prop}
  In the case where $X^i(t) = {J(t)^i}_j X^j(0)$ with $J \in SO(3)$ (independent of $x$), we have:
  \begin{eqnarray}
    M_\Lambda(t) & = & \{{J(t)^i}_j x^j(0), x(0) \in M_\Lambda(0)\} \\
    |\Lambda(x,t) \rrangle & = & u_J(t)^{-1} |\Lambda(J(t)^{-1}x,0) \rrangle \\
    A & = & A_i(J^{-1}(t)x,0){J^{-1}(t)^i}_j \frac{\partial x^j}{\partial s^a} ds^a \nonumber \\
    & & \quad + \imath \omega_{J^{-1}x}(\dot u_J(t)u_J(t)^{-1}) dt \\
    \mathfrak A & = & u_J(t)^{-1} \mathfrak A_i(J^{-1}(t)x,0)u_J(t){J^{-1}(t)^i}_j \frac{\partial x^j}{\partial s^a} ds^a \nonumber \\
    & & \quad + \imath u_J(t)^{-1} \dot u_J(t) dt  \\
    \mathbf A_{\Dis(y,x)}(t) & = & \mathbf G_J(y,x,t)^\dagger \mathbf A_{\Dis(y,x)}(0) \mathbf G_J(y,x,t) \nonumber \\
    & & \quad +\imath \mathbf G_J(y,x,t)^\dagger \dnc G_J(y,x,t)
  \end{eqnarray}
  where $u_J \in SU(2)$ is such that $u_J^{-1} \sigma_j u_J = {J^i}_j \sigma_i$ and $\mathbf G_J(y,x,t) = \Dis(y,x)^{-1} \Dis(J(t)^{-1}y,J(t)^{-1}x)$.\\
  The flow $\varphi^t$ defined by eq.(\ref{eqFlow1},\ref{eqFlow2}) is just the rotation of $M_\Lambda(0)$ in $\mathbb R^3$ defined by $J(t)$.
\end{prop}

\begin{proof}
  We apply simply the properties \ref{gauge1} and \ref{gauge2} concerning the external gauge changes with a time-dependent $J$ in place of a $x$-dependent $J$.
\end{proof}

\begin{example}{2}{Fuzzy surface plots}
  We consider fuzzy surface plots governed by the environmental Hamiltonian $H_{\mathcal E} = \omega a^+a$ ($\mathfrak X$ being the set of observables of a harmonic oscillator or of a single-mode boson field, $\omega$ is the quantum of energy of the environment). $U_{\mathcal E}(t)  =e^{-\imath \omega t a^+a}$ and
  \begin{equation}
    U_{\mathcal E}^\dagger (t) a U_{\mathcal E}(t) = e^{-\imath \omega t} a
  \end{equation}
  It follows that
  \begin{equation}
    |\Lambda(\alpha,t)\rrangle = U_{\mathcal E}^\dagger(t) |\Lambda(\alpha,0)\rrangle = |\Lambda(e^{\imath \omega t}\alpha)\rrangle_0
  \end{equation}
  where $|\Lambda(\alpha)\rrangle_0$ is the quasicoherent state of the time-independent case ($U_{\mathcal E}^\dagger(t)|\alpha\rangle = |e^{\imath \omega t}\alpha\rangle$ for the Perelomov coherent states of the CCR algebra \cite{Perelomov, Puri}). We have then for the fuzzy plane and the fuzzy paraboloids
  \begin{equation}
    A = -\imath \frac{\bar \alpha d\alpha-\alpha d\bar \alpha}{2} + \omega |\alpha|^2 dt + \mathcal O(\epsilon^2)
  \end{equation}
  The flow onto $M_\Lambda$ is then $\varphi^t(\alpha_0) = e^{-\imath \omega t} \alpha_0$ which consists to a rotation on the surface of $M_\Lambda$ around $\alpha=0$ at the frequency $\omega$ as for example fig. \ref{parabolFlow}. 
  \begin{figure}
    \center
    \includegraphics[width=7cm]{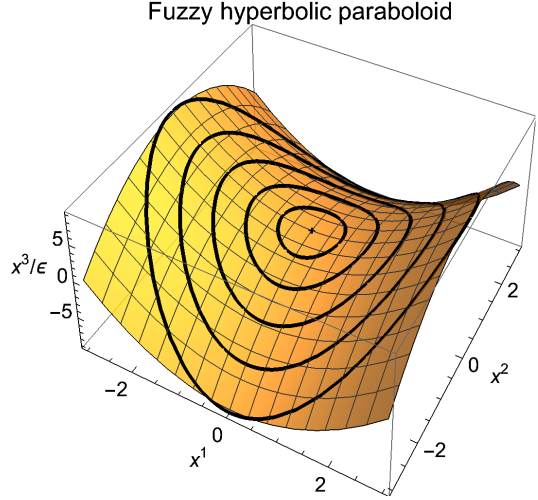}\\
    \caption{\label{parabolFlow} The eigenmanifold $M_\Lambda$ of the fuzzy hyperbolic paraboloid with some orbits of the flow $\varphi^t$ defined by eq.(\ref{eqFlow2}) for a fuzzy space dynamics governed by the Hamiltonian $H_{\mathcal E} = \omega a^+a$}.
  \end{figure}

  In contrast, if we consider a fuzzy surface plot governed by the D2-brane equation of motion eq.(\ref{D2branemotion}) with $X^i(t) = u(t)X^i_0$ for $u(t) \in \mathbb R^+$ such that $u(0)=1$ ($X^i_0$ being the coordinate observables of the time-independent case), we have $\ddot u = 0 \iff u(t)= v t +1$ (for some $v \in \mathbb R$). We have then for the fuzzy plane and the fuzzy paraboloids:
  \begin{eqnarray}
    & & \sigma_i \otimes (X^i(t)-x^i)|\Lambda(x,t)\rrangle = 0 \\
    & \iff & \sigma_i \otimes (u(t)X^i_0 - x^i)|\Lambda(x,t) \rrangle = 0 \\
    & \iff & \sigma_i \otimes (X^i_0 - \frac{x^i}{u(t)}) |\Lambda(x,t) \rrangle = 0 \\
    & \Rightarrow & \left\{ \begin{array}{c} x^i(t) = u(t)x^i_0 \\ |\Lambda(x,t) \rrangle = |\Lambda(\frac{x}{u(t)}) \rrangle_0 \end{array} \right.
  \end{eqnarray}
  where $|\Lambda(\alpha)\rrangle_0$ is the quasi-coherent state of the time-independent case. $M_\Lambda(t)$ grows with an expansion factor $u(t) = vt+1$. We have then
  \begin{equation}
    A = -\imath \frac{\bar \alpha d\alpha - \alpha d\bar \alpha}{2 u(t)^2} + \mathcal O(\epsilon^2)
  \end{equation}

  Now we consider the time-dependent SO(3) transformation $J = \left(\begin{array}{ccc} \cos(\omega t) & -\sin(\omega t) & 0 \\ \sin(\omega t) & \cos (\omega t) & 0 \\ 0 & 0 & 1 \end{array} \right)$ corresponding to a rotation around the third axis at angular velocity $\omega$. We have then $u_J = e^{\frac{\imath}{2} \omega t \sigma_3}$. Now $M_\Lambda(t)$ is just the rotation around the third axis of the eigenmanifold of the time-independent case and
  \begin{equation}
    |\Lambda(\alpha,t)\rrangle = e^{-\frac{\imath}{2} \omega t \sigma_3}|\Lambda(e^{-\imath \omega t}\alpha)\rrangle_0
  \end{equation}
  and
  \begin{equation}
    A = - \imath \frac{\bar \alpha  d\alpha - \alpha d\bar\alpha}{2} + \omega\left(\frac{1}{2}{_0}\llangle \Lambda(\alpha)| \sigma_3|\Lambda(\alpha)\rrangle_0 - |\alpha|^2\right) dt + \mathcal O(\epsilon^2)
  \end{equation}
  
  Note the difference between the dynamics induced by $J$ and $H_{\mathcal E}$. The two ones induces a rotation around $0$ in the base complex plane, but $J$ as rotation in $\mathbb R^3$ around the third axis does not affect the $z$-perturbation in contrast with $H_{\mathcal E}$. $J$ induces a rotation of $\mathfrak M/M_\Lambda$ in $\mathbb R^3$ as a solid whereas $H_{\mathcal E}$ induces a rotation inside $\mathfrak M/M_{\Lambda}$ as a fluid traversed by a current of flow $\varphi^t$.
  
\end{example}

\section{Conclusion}
A fuzzy space $\mathfrak M$ can be endowed with two quantum metrics $\pmb{d\ell^2}_x$ (square length observable) and $\pmb{\gamma}$ (noncommutative symmetric quantum bicovector). The square root of the quantum average of the first one ($\sqrt{\gamma} = \sqrt{\omega_{x(s+ds)}(\pmb{d\ell^2}_{x(s)})}$) measures the length of mean path onto $\mathfrak M$ (the average being defined with respect to the quantum fluctuations of the paths induced by the Heisenberg uncertainties). The quantum average of the square root of the second one ($\dist_{\pmb \gamma} = \omega_x(\sqrt{\pmb{\gamma}})$) measures the mean length of the fluctuating quantum paths. The two associated (minimising) geodesics correspond to minimise the energy uncertainty (for the first one, for the quasicoherent state with a misalignment of the probe; for the second one, for the gap between two quasicoherent states with imposing the microcanonical state as spin state).\\
Gauge theories onto $\mathfrak M$ induces torsions. The Berry phase generator $A$ of the strong adiabatic regime is associated with a contorsional force similar to a Laplace force generated by the magnetic field $F=dA$. The operator valued geometric phase generator $\mathfrak A$ of the weak adiabatic regime is associated with a contorsional force generating spin precessions as in Einstein-Cartan gravity theory. The two associated (autoparallel) geodesics are in accordance with these interpretations.\\
The presence of two metrics is related with the two possible evolutions: smooth moves of the probe and transition from $|\Lambda(x)\rrangle$ to $|\Lambda(y)\rrangle$ after a jump of the probe and induced by the projection after a measurement (which can be modelled by the quantum path $\Dis(y,x)$). If the language of the noncommutative geometry is the more natural to describe fuzzy spaces, these ones are also related to category structures as studied in details \ref{category} because of this double evolution.\\

For the sake of the simplicity we have considered only 3D fuzzy spaces in this paper. However the discussion can be generalised to higher dimensions, with possible interesting applications:

\subsection{Higher dimension}
We can consider a $N$D fuzzy space $(\mathfrak X,\mathbb C^N \otimes \mathcal H, \slashed D_x)$ ($x \in \mathbb R^N$) with (for the Euclidean case) $\mathfrak X = \Lin_{\mathbb R} (X^1,...,X^N,\id)$ and $\slashed D_x = \Sigma_i \otimes (X^i-x^i)$ such that $\tr(\Sigma_i\Sigma_j)=2\delta_{ij}$, or (for the Lorentzian case)  $\mathfrak X = \Lin_{\mathbb R} (X^0,...,X^{N-1},\id)$ and $\slashed D_x = \gamma_\mu \otimes (X^\mu-x^\mu)$ such that $\tr(\gamma_\mu\gamma_\nu)=2\eta_{\mu\nu}$ (Minkowski metric). In these cases $\dim M_\Lambda = N-1$. We have the following applications:
\begin{itemize}
\item $N=3$ and $\Sigma_i = \sigma_i$ (Pauli matrices): string theory BFFS matrix models with orbifoldisation and control of a qubit submitted to environmental noises.
\item $N=4$ and $\gamma_\mu$ the Dirac matrices: string theory IKKT matrix models with orbifoldisation.
\item $N=9$ or $N=10$ and $\gamma_\mu$ generalisations of Dirac matrices: string theory BFSS or IKKT matrix models without orbifoldisation.
\item $N=8$ and $\Sigma_i$ the Gell-Mann matrices: control of a qutrit submitted to environmental noises.
\item $N=9$ and $\Sigma_{(ij)} = \frac{1}{\sqrt 2} \sigma_i \otimes \sigma_j$ : control of two qubits submitted to environmental noises.
\end{itemize}
By ``orbifoldisation'' we mean a dimension reduction as explained in ref. \cite{Berenstein}.\\

It is interesting to note that the fuzzy space models permit to picture an analogy between the two domains of application:

\subsection{Analogy quantum gravity - quantum information}

\rotatebox{90}{\begin{tabular}{c||c||c||c} \scriptsize
  & \textbf{\textit{Noncommutative geometry}} & \textbf{\textit{String theory}}  & \textbf{\textit{Quantum information}} \\
  \hhline{====}
  $\mathbb C^2$ & orientation state space & fermionic string spin state space  & qubit state space \\
  \hline
  $\vec \sigma$ & orientation observables & fermionic string spin observables & qubit observables \\
  \hhline{====}
  $\mathcal H$ & quantum coordinate state space & state space of the fermionic string end & environment state space \\
  \hline
  $\vec X$ & quantum coordinate observables & field operators of D2-brane & environmental noise observables \\
  \hline
  $\Env(\mathfrak X)$ & set of quantum spatial observables & set of D2-brane observables & set of environment observables \\
  \hline
  $\Der(\mathfrak X)$ & quantum tangent vectors & set of D2-brane superoperators & set of environment superoperators\\
  \hline
  $\mathfrak M \times \mathbb R$ & quantum space-time & D2-brane & environment \\
  \hhline{====}
  $\vec x$ & probe (test particle) & D0-brane & qubit control \\
  \hline
  $\mathbb R^3$ & class. space of a Galilean observer & embedding space of strings & control parametric space \\
  \hline
  $t \mapsto x(t)$ & probe move & D0-brane dynamics & quantum computation \\
  \hline
  $|\Lambda \rrangle$ & quasicoherent state of the geometry & ferm. string state of zero displ. energy & bipartite state of zero int. energy\\
  \hline
  $M_\Lambda$ & classical space slice closest to $\mathfrak M$ & D0-brane locations of zero displ. energy  & control manifold of zero int. energy\\
  \hhline{====}
  $\slashed D_x$ & & fermionic string Dirac operator & interaction picture Hamiltonian \\ 
  $\pmb{d\ell^2}_x=\slashed D_x^2$ & square length observable & & \\
  \hline
  $\pmb{\gamma}$ & quantum metric & inducing space metric at thermo. limit &  \\
  \hline
  $\sqrt{\gamma_{ab}ds^ads^b}$ & length of the average path & displ. ener. uncert. meas. misalign. & energy uncert. with meas. misalign. \\
  \hline
  $\dist_{\pmb \gamma}(\omega_{x+\delta x},\omega_x)$ & average length of quantum paths & average displ. ener. gap isotr. uncert. & average energy gap isotropic uncert. \\
  \hline
  $d\tau$ & probe proper time & D0-brane proper time & \\
  $\sqrt{dt^2-d\tau^2}$ & & & ener. uncert. misalign. + time lag \\
  \hhline{====}
  $A$  & shift vector & $U(1)$-gauge potential of the D0-brane & \\
  $e^{-\imath \int A}$  &  & & Berry phase induced by the control \\
  \hline
  $\mathfrak A$ & spacetime Lorentz connection & fermionic string Lorentz connection & \\
  $\Ped^{-\imath \int \mathfrak A}$ & & & logical gate induced by the control \\
  \hline
  $\mathbf A_{\Dis}(L_{\vec X})$ & linking vector observable & quantised gauge field of the D2-brane & \\
  $\Dis(y,x)$ & & & control jump \\
\end{tabular}}

The noncommutative manifold can be seen as a quantised extended solid body or as a direct noncanonical quantisation of spacetime. The string theory model consists to a fermionic string linking a bosonic D2-brane and a D0-brane. The thermodynamical limit consists to a number of bosonic strings tending to infinity with constant density (semi-classical limit to the macroscopic scale). The displacement energy is the ``energy tension'' of the fermionic string. The two branes are submitted to gauge fields (electromagnetic interaction). The quantum information model consists to a qubit controlled by classical fields and interacting with a large environment. For this case, the energy corresponds to the total energy interaction (qubit-control + qubit-environment). The variation of the control fields is assimilated to a quantum computation (by assimilating the resulting evolution operator to a logical gate).\\
The measure misalignment means that the measurement is performed with a small location error; the time lag means that the measurement slightly lags behind the dynamics.\\

This analogy between quantum gravity (described by noncommutative geometry or string theory) and quantum information theory can be related to other puzzling arguments (analogy between the classifications of qubit entanglements and charged string black holes \cite{Borsten}, ER=EPR conjecture \cite{Maldacena}, AdS/CFT correspondance \cite{Ydri}, string gauge structures arising in semi-classical light-matter interaction \cite{SF}) suggesting that the ``spacetime Reality'' is like a quantum computer (the ``It for qubit'' proposal \cite{Zizzi}).

\subsection{Future research directions}
An important ingredient of the present study is the displacement operator $\Dis(y,x)$. But a lot of the fuzzy spaces are not linkable (in the meaning of this paper). In \ref{paralinkable} we have generalised to ``paralinkable'' fuzzy spaces for which it exists a generalised displacement operator. But no argument permits to state that all fuzzy spaces are totally paralinkable (especially, it is not the case if the quasicoherent state is separable at some points but is entangled at some others). It needs then to generalise the present study to unlikable fuzzy spaces. Moreover, the quantum local orientation (the spin degree of freedom) playing an important role, it could be interesting to study the possibility of defining non-orientable fuzzy spaces (at least fuzzy spaces having non-orientable eigenmanifolds).

\appendix
\section{Perturbation theory of fuzzy spaces}\label{perturbation}
Let $\mathfrak M= (\mathfrak X,\mathbb C^2 \otimes \mathcal H,\slashed D_x)$ be a fuzzy space, $(\delta X^i)$ be a set of perturbations of the coordinate observables, $\hat \mathfrak M= (\hat \mathfrak X,\mathbb C^2 \otimes \mathcal H,\hat {\slashed D_x})$ be the perturbed fuzzy space: $\hat X^i = X^i + \delta X^i$. In the whole of this appendix, ``$\simeq$'' stands for equal up to negligible element of magnitude $\mathcal O(\|\delta_{ij}\delta X^i\delta X^j\|)$, where $\|\bullet\|$ is the operator norm induced by $\llangle \bullet | \bullet \rrangle$.\\

The perturbed quasicoherent state is solution of
\begin{equation}
  (\sigma_i \otimes (X^i -x^i) + \sigma_i \otimes (\delta X^i -\delta x^i))|\hat \Lambda \rrangle = 0
\end{equation}
with $x \in M_\Lambda$ (eigen manifold of $\mathfrak M$) and $x + \delta x \in \hat M_\Lambda$ (eigen manifold of $\hat \mathfrak M$).

\subsection{$\mathfrak M$ strongly non-degenerate at $x$}
The perturbation vector $\llangle \delta \vec X \rrangle = \llangle \hat \Lambda |\delta \vec X|\hat \Lambda \rrangle \simeq \llangle \Lambda |\delta \vec X|\Lambda \rrangle$ defines the perturbation direction $\vec p_x = \frac{\llangle \delta \vec X \rrangle }{\|\llangle \delta \vec X \rrangle\|}$. We set then $\delta x^i = \|\delta \vec x\| p^i_x$. At the first order of the perturbation theory, the eigenvalue $\hat \lambda_0$ issued from the zero eigenvalue of $\slashed D_x$ is
\begin{equation}
  \hat \lambda_0(x,\delta x) \simeq \llangle \Lambda(x)|\sigma_i \otimes(\delta X^i-\delta x^i)|\Lambda(x) \rrangle
\end{equation}
The condition that $x+\delta x \in \hat M_\Lambda$ implies that:
\begin{eqnarray}
  \hat \lambda_0 \simeq 0 & \Rightarrow & n_i(x) \delta x^i = \llangle \Lambda(x)|\sigma_i \otimes \delta X^i |\Lambda(x) \rrangle \\
  & \Rightarrow &  \|\delta \vec x\| \vec p_x \cdot \vec n_x  = \llangle \Lambda(x)|\sigma_i \otimes \delta X^i |\Lambda(x) \rrangle \\
  & \Rightarrow & \delta \vec x = \frac{\llangle \Lambda(x)|\sigma_i \otimes \delta X^i |\Lambda(x) \rrangle}{\vec p_x \cdot \vec n_x} \vec p_x
\end{eqnarray}
where $\vec n_x = \llangle \Lambda(x)|\vec \sigma|\Lambda(x)\rrangle$ is the normal vector to $M_\Lambda$ at $x$.\\
If the quasicoherent state is separable, $|\Lambda(x)\rrangle = |O_x \rangle \otimes |\Omega_x \rangle$, then $\llangle \delta \vec X \rrangle = \langle \Omega_x|\delta \vec X|\Omega_x\rangle$, $\llangle \Lambda(x)|\sigma_i \otimes \delta X^i|\Lambda(x)\rrangle = \langle O_x|\vec \sigma|O_x\rangle \cdot \langle \Omega_x|\delta \vec X|\Omega_x\rangle$ and $\vec n_x = \langle O_x|\vec \sigma|O_x\rangle$. It follows that $\delta \vec x = \langle \Omega_x|\delta \vec X|\Omega_x \rangle$.\\
The metric of $\hat M_\Lambda$ is
\begin{equation}
  \hat \gamma_{ab} \simeq \gamma_{ab} + \delta_{ij}\left(\frac{\partial x^i}{\partial s^a} \frac{\partial \delta x^j}{\partial s^b} + \frac{\partial x^i}{\partial s^b} \frac{\partial \delta x^j}{\partial s^a} + \frac{\partial \delta x^i}{\partial s^a}\frac{\partial \delta x^j}{\partial s^b}\right)
\end{equation}

The perturbed quasicoherent state is then
\begin{equation}
  |\hat \Lambda(x+\delta x)\rrangle \simeq |\Lambda(x)\rrangle -\sum_{n>0} \frac{\llangle \lambda_n(x)|\sigma_i\otimes(\delta X^i-\delta x^i)|\Lambda(x)\rrangle}{\lambda_n(x)}|\lambda_n(x)\rrangle
\end{equation}
where $\slashed D_x|\lambda_n(x)\rrangle = \lambda_n(x)|\lambda_n(x)\rrangle$, $\forall x \in M_\Lambda$ ($\lambda_0 = 0$).

\subsection{$\mathfrak M$ only weakly non-degenerate at $x$}\label{perturbWD}
We consider the case of a strictly weakly non-degenerate separable quasicoherent state. In this case, all state $|O_{\Theta}\rangle \otimes |\Omega_x\rangle$ is quasicoherent state with $|O_\Theta \rangle = \cos \theta |0\rangle + e^{\imath \varphi} \sin \theta |1\rangle$ ($\Theta=(\theta,\varphi)$, $(|0\rangle,|1\rangle)$ is the canonical basis of $\mathbb C^2$). Let $|I_\Theta \rangle = - e^{-\imath \varphi} \sin \theta|0\rangle + \cos \theta |1\rangle$ be an orthogonal vector to $|O_\Theta\rangle$. The perturbation direction does not depend on the choice of $\Theta$ since $\llangle \Lambda(x,\Theta)|\delta \vec X|\Lambda(x,\Theta)\rrangle = \langle \Omega_x|\delta \vec X|\Omega_x \rangle$.\\
The eigenvalue matrix is then
\begin{equation}
  [\hat \lambda_0] = \left(\begin{array}{cc} \langle O_\Theta|\sigma_i|O_\Theta\rangle & \langle O_\Theta|\sigma_i|I_\Theta\rangle \\ \langle I_\Theta|\sigma_i|O_\Theta\rangle & \langle I_\Theta|\sigma_i|I_\Theta\rangle \end{array} \right) \langle \Omega_x|(\delta X^i-p^i_x\|\delta \vec x\|)|\Omega_x \rangle
\end{equation}
$[\hat \lambda_0] = 0$ if $\delta \vec x = \langle \Omega_x|\delta \vec X|\Omega_x \rangle$ (the perturbed strictly weakly degenerate quasicoherent state is still strictly weakly degenerate). 
The perturbed quasicoherent states are then
\begin{eqnarray}
  & & |\hat \Lambda(x,\Theta)\rrangle \simeq  |O_\Theta \rangle \otimes|\Omega_x\rangle \nonumber \\
  & & \quad - \sum_{n>0} \frac{\llangle \lambda_n(x)|\sigma_i\otimes(\delta X^i-\delta x^i)|O_\Theta \rangle \otimes |\Omega_x\rangle}{\lambda_n(x)} |\lambda_n(x)\rrangle
\end{eqnarray}

For an entangled quasicoherent state, we have several perturbation directions $\vec p_{x,u} = \llangle \Lambda(x)|u^\dagger u \delta \vec X|\Lambda(x)\rrangle$ where $u \in \Lin(\id,u_1,u_2,u_3)$ with $u_i$ such that $\tr(\rho_\Lambda u_i)=0$ ($u \in U(2)_{a.s.}$, one or two $u_i$ are possibly zero). For each perturbation direction, we have
\begin{equation} [\hat \lambda_0] = \left(\begin{array}{cc} \llangle \Lambda_x| u^\dagger \sigma_i u \otimes (\delta X^i - p^i_{x,u}\|\delta \vec x\|)|\Lambda_x\rrangle  &  ... \\ \vdots & \ddots \end{array} \right)
\end{equation}
in the basis $(u|\Lambda\rrangle, uu_1|\Lambda \rrangle, uu_2|\Lambda \rrangle, uu_3|\Lambda \rrangle)$. The perturbation theory needs to diagonalise $[\hat \lambda_0]$, but in this case the eigenbasis $(\tilde u_0|\Lambda\rrangle, \tilde u_1|\Lambda \rrangle, \tilde u_2|\Lambda \rrangle, \tilde u_3|\Lambda \rrangle)$ is not consistent with the direction $\vec p_{x,u}$. We need to find $u_*$ such that $[\hat \lambda_0]$ is diagonal (or at least is block-diagonal of shape $(1\times 1) \oplus (3 \times 3)$), fixing a single consistent perturbation direction $\vec p_{x,u_*}$. We have then $\delta \vec x = \frac{\llangle \Lambda(x)|u_*^\dagger \sigma_i u_* \otimes \delta X^i|\Lambda(x)\rrangle}{\vec p_{x,u_*} \cdot \vec n_{x,u_*}} \vec p_{x,u_*}$ with $\vec n_{x,u_*} = \llangle \Lambda(x)|u_*^\dagger \vec \sigma u_*|\Lambda(x)\rrangle$.

\section{Paralinkable fuzzy spaces and entangled quantum paths}\label{paralinkable}
\subsection{Generalised displacement operators}
For the sake of the simplicity, in the main text we have treated only the case of totally linkable fuzzy spaces, i.e. fuzzy spaces for which two quasicoherent states can be linked by a transformation $u_{yx} \otimes \Dis(y,x) \in SU(2)_{a.s.} \otimes e^{\imath \mathfrak X}$. This assumption can be considered as too restrictive, but it is pertinent since it implies that the magnitude of the quantum entanglement be constant on $M_\Lambda$ (and so corresponds to fuzzy space homogeneous from the point of view of the quantum information properties). More precisely, if we consider the eigen density matrices: $\rho_\Lambda(y) = u_{yx} \rho_\Lambda(x) u_{yx}^\dagger$, the von Neumann entropy $S_\Lambda = - \tr\left(\rho_\Lambda \ln \rho_\Lambda\right)$ is constant (this one can be considered as a measure of the entanglement since $S_\Lambda=0$ for a pure state $\iff$ a separable quasicoherent state, and $S_\Lambda = \ln 2 = \sup S$ for the microcanonical state $\frac{1}{2} \id$ $\iff$ a maximally entangled quasicoherent state).\\

In this appendix we want to generalise the discussion by choosing weaker assumptions. For a fuzzy space $\mathfrak M$ supposed weakly non-degenerate, we said that two points $x,y \in M_\Lambda$ are paralinkable if it exists an orthonormal basis of $\mathbb  C^2$, $\mathcal B_{x} = (|O_{x}\rangle,I_{x}\rangle)$, two operators $\Dis^\alpha(y,x) \in e^{\imath \Env(\mathfrak X)}$ (with $\alpha \in \{0,1\}$) and $u_{yx} \in U(2)$ such that
\begin{equation}
  |\Lambda(y) \rrangle = u_{yx} P_\alpha(x) \otimes \Dis^\alpha(y,x) |\Lambda(x) \rrangle
\end{equation}
with $P_0(x) = |O_{x}\rangle \langle O_{x}|$ and $P_1(x) = |I_{x}\rangle \langle I_{x}|$. We can write
\begin{eqnarray}
  |\Lambda(x)\rrangle & = & |O_{x}\rangle \otimes |\Lambda^0(x)\rangle + |I_{x}\rangle \otimes |\Lambda^1(x)\rangle \\
  |\Lambda(y)\rrangle & = & |O_{y}\rangle \otimes |\Lambda^0(y)\rangle + |I_{y}\rangle \otimes |\Lambda^1(y)\rangle
\end{eqnarray}
with $u_{yx}$ the matrix of basis change from $\mathcal B_x$ to $\mathcal B_y$. $p_\alpha(x) = \langle \Lambda^\alpha(x)|\Lambda^\alpha(x)\rangle$ is the probability to the spin state be $P_\alpha(x)$ if the quasicoherent state is $|\Lambda(x)\rrangle$ ($p_1+p_0=1$). By construction $|\Lambda^\alpha(y)\rangle = \Dis^\alpha(y,x)|\Lambda^\alpha(x)\rangle$ (no summation onto the indices $\alpha$ which are at the same level). The eigen density matrix in $\mathcal B_{x}$ is $\rho_\Lambda(x) = \left(\begin{array}{cc} p_0(x) & c(x) \\ \overline{c(x)} & p_1(x) \end{array} \right)$ with the coherence defined by $c(x) = \langle \Lambda^1(x)|\Lambda^0(x)\rangle$. Two properties are relaxed with respect to the definition of linkable points, firstly the generalised displacement operator $P_\alpha(x) \otimes \Dis^\alpha(y,x)$ is not separable, and secondly $\Dis^\alpha \in e^{\imath \Env(\mathfrak X)} \varsupsetneq e^{\imath \mathfrak X}$. Since elements of $\Env(\mathfrak X)$ are not self-adjoint, $\Dis^\alpha$ are not necessary unitary, but the normalisation condition implies that
\begin{equation}
  \sum_\alpha \langle \Lambda^\alpha(x)|\Dis^\alpha(y,x)^\dagger \Dis^\alpha(y,x)|\Lambda^\alpha(x)\rangle = 1
\end{equation}
but $p_\alpha(y) = \langle \Lambda^\alpha(x)|\Dis^\alpha(y,x)^\dagger \Dis^\alpha(y,x)|\Lambda^\alpha(x)\rangle \not= \langle \Lambda^\alpha(x)|\Lambda^\alpha(x)\rangle = p_\alpha(x)$. The relaxation of the two assumptions permits to have a change of the occupation probabilities $(p_\alpha)$ (and a change of coherence $c$) between $x$ and $y$, and so a change of the entanglement magnitude: $S_\Lambda(y) \not= S_\Lambda(x)$. But the entanglement class cannot be changed, we cannot link a separable state to an entangled state (as for example with $|\Lambda^1(x)\rangle = 0$ and $|\Lambda^1(y) \rangle \not=0$). So a fuzzy space $\mathfrak M$ totally paralinkable keeps a certain degree of homogeneity since all its quasicoherent states are in the same entanglement class (but with various entropy values).\\

Let $\varrho^\alpha_\Lambda = |\Lambda^\alpha \rangle \langle \Lambda^\alpha|$ and $P^\alpha_\Lambda = \frac{\varrho^\alpha_\Lambda}{\tr \varrho^\alpha_\Lambda}$ (i.e. $\varrho^\alpha_\Lambda = p^\alpha P^\alpha_\Lambda$, and $\tr(P^1P^0) = |c|^2 \not=0$). We have still $P_\Lambda = |\Lambda \rrangle \llangle \Lambda|$. We can introduce the following pure states of $\Env(\mathfrak X)$: $\omega_x = \tr(P_\Lambda(x) \bullet)$ and $\omega^\alpha_x = \tr(P^\alpha_\Lambda(x) \bullet)$, and after some algebra using the trace cyclicality we can show that $\omega_x$ is a convex combination of $(\omega_x^\alpha)$:
\begin{equation}
  \omega_x = p_\alpha(x) \omega^\alpha_x
\end{equation}
and by construction, we have
\begin{equation}
  \omega^\alpha_y = \frac{p_\alpha(y)}{p_\alpha(x)} \omega^\alpha_x \circ \Ad_{\Dis^\alpha(y,x)}
\end{equation}
(without summation of the indices $\alpha$).\\

The generalised displacement operator $P_\alpha \otimes \Dis^\alpha$ can be seen as a quantum entanglement between quantum paths on $\mathfrak M$ and local orientation (spin) of $\mathfrak M$. From $\omega_x$ to $\omega_y$, the quantum path $\Dis^0(y,x)$ is followed if the spin is initially in the pure state $P^0(x)$ whereas the quantum path $\Dis^1(y,x)$ is followed if the spin is initially in the pure state $P^1(x)$. So with respect to the discussion in main sections of this paper, in addition to an entanglement between quantum location and quantum spin (``local'' entanglement of the state $|\Lambda(x)\rrangle$) we can have also entanglement between quantum paths and quantum spin (``nonlocal'' entanglement of displacement operator $P_\alpha(x) \otimes \Dis^\alpha(y,x)$, nonlocal since it is defined onto two points $x$ and $y$ in the $\mathbb R^3$ the classical space of the observer). From $\omega_x$ to $\omega_y$ the system follows ``in parallel'' the two quantum paths $\Dis^\alpha(y,x)$ with probabilities $p_\alpha(x)$ (this is the reason for which we said that the two points are ``paralinkable'').\\

We can use the quantum metric $\pmb \gamma$ with generalised displacement operators but only with its representation on $(\dnc X^i)$ and not with the one on $(\theta_i)$ (since the generators of $\Ad_{\Dis^\alpha}$ are not in $L_{\mathfrak X}$ but in $ L_{\Env(\mathfrak X)}$ -- we recall that $(\theta_i)$ is the dual basis of $(L_{X^i})$ --). With $\Dis^\alpha(y,x) = e^{\imath \delta \Pi^\alpha}$ ($\delta \Pi^\alpha \in \Env(\mathfrak X)$), we have $\dnc X^i(L_{\delta \Pi^\alpha}) = -\imath [\delta \Pi^\alpha,X^i]$ by extension of noncommutative differential calculus from $\Der(\mathfrak X)$ to $\Der(\Env(\mathfrak X))$. So the distance between $\omega_x$ and $\omega_y$ is
\begin{equation}
  \dist_{\pmb \gamma}(\omega_x,\omega_y) = \omega_x\left(\sqrt{\left(\begin{array}{cc} \pmb \gamma(L_{\delta \Pi^0},L_{\delta \Pi^0}) & \pmb \gamma(L_{\delta \Pi^0},L_{\delta \Pi^1}) \\ \pmb \gamma(L_{\delta \Pi^1},L_{\delta \Pi^0}) & \pmb \gamma(L_{\delta \Pi^1},L_{\delta \Pi^1}) \end{array} \right)} \right)
\end{equation}
where the matrices are written in the basis $\mathcal B_x$. The distance between $\omega_x$ and $\omega_y$ is set to be the mean value (in the state $\omega_x$) of the average length of the two quantum paths $\Dis^\alpha(y,x)$ (``average length'' in the meaning of the averaging onto the quantum fluctuations associated with the quantum uncertainties $\Delta X^i$).\\
We can also define a distance without taking into account of the correlations between the two quantum paths:
\begin{eqnarray}
  & & \dist_{\pmb \gamma}^{off}(\omega_x,\omega_y) \nonumber \\
  & & \quad = \omega_x\left(\left(\begin{array}{cc} \sqrt{\pmb \gamma(L_{\delta \Pi^0},L_{\delta \Pi^0})} & 0 \\ 0 & \sqrt{\pmb \gamma(L_{\delta \Pi^1},L_{\delta \Pi^1})} \end{array} \right) \right) \\
  & & \quad = p_\alpha(x) \omega_x^\alpha(\sqrt{\pmb \gamma(L_{\delta \Pi^\alpha},L_{\delta \Pi^\alpha})})
\end{eqnarray}

We can also introduce spinors of linking vector observables:
\begin{equation}
 \pmb{\delta \vec \ell}_{yx} = \left(\begin{array}{cc} \pmb{\delta \vec \ell}_{yx}^0 & \pmb{\delta \vec C}_{yx}^{01} \\ \pmb{\delta \vec C}_{yx}^{10} & \pmb{\delta \vec \ell}_{yx}^1 \end{array}\right) \quad,\quad  \pmb{\delta \vec \ell}_{yx}^{off} = \left(\begin{array}{cc} \pmb{\delta \vec \ell}_{yx}^0 & \vec  0 \\ \vec 0 & \pmb{\delta \vec \ell}_{yx}^1 \end{array}\right)
\end{equation}
where the linking vector observables are $\pmb{\delta \vec \ell}_{yx}^\alpha =  \Dis^\alpha(y,x)^\dagger [\vec X,\Dis^\alpha(y,x)] = -\imath [\delta \Pi^\alpha,\vec X] + \mathcal O(\|\delta \vec x\|^2)$ and the correlators are $\pmb{\delta \vec C}_{yx}^{01} = \Dis^0(y,x)^\dagger [\vec X,\Dis^1(y,x)] = -\imath [\delta \Pi^1,\vec X] + \mathcal O(\|\delta \vec x\|^2)$. We can see that
\begin{equation}
  \pmb{\delta \vec \ell}_{yx}^{off} \cdot \pmb{\delta \vec \ell}_{yx} = \left(\begin{array}{cc} \pmb \gamma(L_{\delta \Pi^0},L_{\delta \Pi^0}) & \pmb \gamma(L_{\delta \Pi^0},L_{\delta \Pi^1}) \\ \pmb \gamma(L_{\delta \Pi^1},L_{\delta \Pi^0}) & \pmb \gamma(L_{\delta \Pi^1},L_{\delta \Pi^1}) \end{array} \right) + \mathcal O(\|\delta \vec x\|^2)
\end{equation}
and then $\dist_{\pmb \gamma}(\omega_x,\omega_y) = \omega_x(\sqrt{\pmb{\delta \vec \ell}_{yx}^{off} \cdot \pmb{\delta \vec \ell}_{yx}})$ and $\dist_{\pmb \gamma}^{off}(\omega_x,\omega_y) = \omega_x(|\pmb{\delta \ell}^{off}_{yx}|) = p_\alpha(x) \omega^\alpha_x(|\pmb{\delta \ell}^\alpha_{yx}|)$.\\

\begin{example}{2}{Fuzzy hyperboloid}
  We can rewrite the quasicoherent state of the fuzzy hyperboloid as:
  \begin{equation}
    |\Lambda(\alpha)\rrangle = \left(1-\sqrt{2}\sigma_1 \otimes (a-\alpha)\hat N_\alpha^{-1} X^3\right) \mathcal D(\alpha) \left(\begin{array}{c} |0\rangle \\ 0 \end{array} \right) + \mathcal O(\epsilon^2)
  \end{equation}
  where $\hat N_\alpha = (a^\dagger-\bar \alpha)(a-\alpha) = \sum_{n=0}^{+\infty} n|n\rangle_\alpha{_\alpha}\langle n|$, $\hat N_\alpha^{-1} = \sum_{n=1}^{+\infty} \frac{1}{n} |n\rangle_\alpha{_\alpha}\langle n|$ being its pseudo-inverse, and where $\mathcal D(\alpha)$ is the displacement operator of the CCR algebra coherent states \cite{Perelomov}. It follows that
  \begin{eqnarray}
    |\Lambda(\beta) \rrangle & = &  \left(1-\sqrt{2}\sigma_1 \otimes (a-\beta)\hat N_\beta^{-1} X^3\right) \mathcal D(\beta) \nonumber \\
    & & \quad \times  \mathcal D(\alpha)^{-1} \left(1+\sqrt{2}\sigma_1 \otimes (a-\alpha)\hat N_\alpha^{-1} X^3\right) |\Lambda(\alpha)\rrangle + \mathcal O(\epsilon)
  \end{eqnarray}
  and then after some algebra
  \begin{eqnarray}
    |\Lambda(\beta) \rrangle & = & e^{\imath \Im(\alpha \bar \beta)} (1+\sqrt{2}\sigma_1 \otimes (a-\beta)\hat N_\beta^{-1} (\mathcal D(\beta-\alpha)X^3\mathcal D(\beta-\alpha)^{-1}-X^3)) \nonumber \\
    & & \quad \times \mathcal D(\beta-\alpha) |\Lambda(\alpha) \rrangle + \mathcal O(\epsilon^2)
  \end{eqnarray}
  Let $W_{\beta \alpha} \equiv \sqrt{2} (a-\beta)\hat N_\beta^{-1} (\mathcal D(\beta-\alpha)X^3\mathcal D(\beta-\alpha)^{-1}-X^3)$, in the basis $(|0\rangle,|1\rangle)$ we have then
  \begin{equation}
    |\Lambda(\beta)\rrangle = e^{\imath \Im(\alpha \bar \beta)} \left(\begin{array}{cc} \mathcal D(\beta-\alpha) & W_{\beta \alpha} \mathcal D(\beta-\alpha) \\ W_{\beta \alpha} \mathcal D(\beta-\alpha) & \mathcal D(\beta-\alpha) \end{array} \right)|\Lambda(\alpha) \rrangle + \mathcal O(\epsilon^2)
  \end{equation}
  in the new basis defined by $|\pm\rangle = \frac{|0\rangle\pm|1\rangle}{\sqrt 2}$  we have
   \begin{eqnarray}
     |\Lambda(\beta)\rrangle & = & e^{\imath \Im(\alpha \bar \beta)} \left(\begin{array}{cc} (1+W_{\beta \alpha}) \mathcal D(\beta-\alpha) & 0 \\ 0 & (1-W_{\beta \alpha}) \mathcal D(\beta-\alpha) \end{array} \right)|\Lambda(\alpha) \rrangle \nonumber \\
     & & \qquad + \mathcal O(\epsilon^2) \\
     & = & e^{\imath \Im(\alpha \bar \beta)} \left(\begin{array}{cc} e^{W_{\beta \alpha}} \mathcal D(\beta-\alpha) & 0 \\ 0 & e^{-W_{\beta \alpha}} \mathcal D(\beta-\alpha) \end{array} \right)|\Lambda(\alpha) \rrangle \nonumber \\
     & & \qquad + \mathcal O(\epsilon^2) \\
     & = & P_\alpha \otimes \Dis^\alpha(\beta,\alpha) |\Lambda(\alpha)\rrangle + \mathcal O(\epsilon^2)
   \end{eqnarray}
   with $P_\pm = |\pm\rangle\langle \pm|$, and $\Dis^\pm(\beta,\alpha) = e^{\imath \Im(\alpha \bar \beta)} e^{\pm W_{\beta \alpha}} \mathcal D(\beta-\alpha)$. The fuzzy hyperboloid is then paralinkable (at the first order of the perturbation theory) except at $\alpha=0$.  $\delta \Pi^\pm_\alpha = -\imath (\delta \alpha a^+ - \delta \bar \alpha a) \mp \imath \delta W_{\alpha} + \mathcal O(\max(\epsilon,|\delta \alpha|)^2)$, with $\delta W_\alpha = \sqrt{2}(a-\alpha)\hat N^{-1}_\alpha [\delta \alpha a^+-\delta \bar \alpha a,X^3]$.\\

   The fuzzy hyperboloid model corresponds to a two-sheet hyperboloid, for which we have considered only the upper sheet. The second sheet is obtained by $\epsilon \to -\epsilon$ ($X^3 \to - X^3$). The quasicoherent states of the two sheets are related by a simple internal global gauge change:
   \begin{equation}
     |\Lambda_{-\epsilon}(\alpha) \rrangle = \sigma_3 |\Lambda_{+\epsilon}(\alpha)\rrangle
   \end{equation}
   ($\sigma_3 = |0\rangle\langle0|-|1\rangle\langle 1|=|-\rangle\langle+|+|+\rangle\langle -| = \sigma_1^{\mathcal B}$). $|\Lambda_{\pm \epsilon}(\alpha) \rrangle$ are then linkable (with a trivial displacement operator, their distance by $\pmb \gamma$ is then zero). By definition the two sheets $M_{\Lambda,\pm\epsilon}$ are each connected but they are not connected each other. In contrast, the two sheets are each not linkable (they are only paralinkable except at $\alpha=0$) but they are linked each other (more precisely, each point is linkable with its symmetric point by the reflection $x^3 \to - x^3$).\\
   
   In the basis $\mathcal B = (|\pm\rangle)$ the quasicoherent state is
   \begin{equation}
     |\Lambda(\alpha) \rrangle = \frac{1}{\sqrt 2} \left(\begin{array}{c} |0\rangle_\alpha \\ |0\rangle_\alpha \end{array}\right) - \sum_{n=1}^{+\infty} \frac{{_\alpha}\langle n|X^3|0\rangle_\alpha}{\sqrt{2n}} \left(\begin{array}{c} |n-1\rangle_\alpha \\ - |n-1\rangle_\alpha \end{array} \right) + \mathcal O(\epsilon^2)
   \end{equation}
   It follows that the probabilities to follows the quantum paths $\Dis^\pm$ are
   \begin{eqnarray}
     p_\pm(\alpha) & = & \frac{1}{2} \mp \Re\left({_\alpha}\langle 1|X^3|0\rangle_\alpha\right) + \mathcal O(\epsilon^2) \\
     & = & \frac{1}{2} \mp\epsilon e^{-|\alpha|^2} \sum_{n=0}^{+\infty} \frac{\sqrt{r^2+3/2+n}-\sqrt{r^2+1/2+n}}{n!} |\alpha|^{2n} \Re(\alpha) \nonumber \\
     & & \qquad + \mathcal O(\epsilon^2)
   \end{eqnarray}
   and the coherence is
   \begin{eqnarray}
     c(\alpha) & = &  \frac{1}{2} - \Im\left({_\alpha}\langle 1|X^3|0\rangle_\alpha\right) + \mathcal O(\epsilon^2) \\
     & = & \frac{1}{2} - \epsilon e^{-|\alpha|^2} \sum_{n=0}^{+\infty} \frac{\sqrt{r^2+3/2+n}-\sqrt{r^2+1/2+n}}{n!} |\alpha|^{2n} \Im(\alpha) \nonumber \\
     & & \qquad + \mathcal O(\epsilon^2)
   \end{eqnarray}
   The linear entropy of the quasicoherent state is then
   \begin{eqnarray}
     S_\Lambda^{lin}(\alpha) & = & 1- \tr(\rho_\Lambda^2) \\
     & = & 2 \Im\left({_\alpha}\langle 1|X^3|0\rangle_\alpha\right) + \mathcal O(\epsilon^2) \\
     & = & 2 \epsilon e^{-|\alpha|^2} \sum_{n=0}^{+\infty} \frac{\sqrt{r^2+3/2+n}-\sqrt{r^2+1/2+n}}{n!} |\alpha|^{2n} \Im(\alpha) \nonumber \\
     & & \qquad + \mathcal O(\epsilon^2)
   \end{eqnarray}
   (the linear entropy is the first order approximation of the von Neumann entropy). 
\end{example}

\subsection{Noncommutative gauge potentials}
We can introduce a new noncommutative ``nonabelian'' gauge potential:
\begin{equation}
  \pmb{\mathfrak A}_\Dis = \imath \left(\begin{array}{cc} \Dis^{0\dagger} \dnc \Dis^0 & \Dis^{0\dagger} \dnc \Dis^1 \\ \Dis^{1\dagger} \dnc \Dis^0 & \Dis^{1\dagger} \dnc \Dis^1 \end{array} \right) \in \Omega^1_{\Der}(\Env(\mathfrak X),\mathfrak{gl}(2,\mathbb C))
\end{equation}
for which we have $\pmb{\mathfrak A}_{\Dis(y,x)}(L_{\vec X}) = \pmb{\delta \vec \ell}_{yx}$.
We have then
\begin{equation}
  \tr\left(\varrho_\Lambda(x) \pmb{\mathfrak A}_{\Dis(y,x)}^{off}(L_{X^i})\right) = y^i - x^i
\end{equation}
with $\pmb{\mathfrak A}_{\Dis(y,x)}^{off}$ the block diagonal matrix of blocks $(\imath \Dis^{\alpha\dagger} \dnc \Dis^\alpha)$. The noncommutative ``nonabelian'' curvature $\pmb{\mathfrak F}_\Dis = \dnc \pmb{\mathfrak A}_\Dis - \imath [\pmb{\mathfrak A}_\Dis,\pmb{\mathfrak A}_\Dis] \in \Omega^2_\Der(\Env(\mathfrak X),\mathfrak{gl}(2,\mathbb C))$ is not zero.\\

$\pmb{\mathfrak A}_\Dis$ is a noncommutative equivalent of $\mathfrak A$ which has $A = \tr(\rho_\Lambda \mathfrak A)$ as ``abelianisation''. In a same way, we can define a noncommutative ``abelianisation'' of $\pmb{\mathfrak A}_\Dis$:
\begin{eqnarray}
  \mathbf A_\Dis & = & \tr_{\mathbb C^2} \left(\rho_\Lambda \pmb{\mathfrak A}_\Dis \right) \\
  & = & p_\alpha \mathbf A^\alpha_\Dis + \mathbf C_\Dis \in \Omega^1_\Der(\Env(\mathfrak X))
\end{eqnarray}
with $\mathbf A^\alpha_\Dis = \imath \Dis^{\alpha \dagger} \dnc \Dis^\alpha$ and $\mathbf C_\Dis = \imath c \Dis^{1\dagger} \dnc \Dis^0 + \imath \bar c \Dis^{0\dagger} \dnc \Dis^1 $.\\
The noncommutative ``abelian'' curvature is also nonzero in this case (if $p_\alpha \not=0$):
\begin{eqnarray}
  \mathbf F_\Dis[L_X,L_Y] & = & \dnc \mathbf A_\Dis(L_X,L_Y) -\imath [\mathbf A_\Dis(L_X),\mathbf A_\Dis(L_Y)] \\
  & = & \imath p_0p_1 \varsigma_{\alpha \beta} [\mathbf A^\alpha_\Dis(L_X),\mathbf A^\beta_\Dis(L_Y)] \nonumber \\
  & & \quad + \dnc \mathbf C_\Dis(L_X,L_Y) -\imath [\mathbf C_\Dis(L_X),\mathbf C_\Dis(L_Y)] \nonumber \\
  & & \quad -\imath p_\alpha [\mathbf A_\Dis^\alpha(L_{[X}),\mathbf C_\Dis(L_{Y]})] 
\end{eqnarray}
where $\varsigma_{\alpha \beta} \equiv (-1)^{\alpha -\beta}$. We have also $\mathbf A^{off}_\Dis = \tr_{\mathbb C^2} \left(\rho_\Lambda \pmb{\mathfrak A}_\Dis^{off} \right) = p_\alpha \mathbf A^\alpha_\Dis$ and $\mathbf F_\Dis^{off}(L_X,L_Y) = \imath p_0p_1 \varsigma_{\alpha \beta} [\mathbf A^\alpha_\Dis(L_X),\mathbf A^\beta_\Dis(L_Y)]$.

\section{The Lorentz connection}
\subsection{Frame changes on $\Omega$}\label{frameOmega}
\begin{prop}
  Under a local external gauge transformation $J \in \underline{SO(3)}_{\mathbb R^3}$, $\Omega$ becomes:
  \begin{eqnarray}
    \hat \Omega^{jk}_m & = & {J_{j'}}^j {J_{k'}}^k \Omega^{j'k'}_{m'} ({J_m}^{m'}+\partial_m {J_p}^{m'}x^p) + \frac{1}{2} J^{[jl} \partial_m {J^{k]}}_l \\
    \hat \Omega^{i0}_m & = & {J_{i'}}^i \Omega^{i'0}_{m'} ({J_m}^{m'}+\partial_m {J_p}^{m'}x^p)
  \end{eqnarray}
  with ${J_{j'}}^j = {[J^t]^j}_{j'} = {[J^{-1}]^j}_{j'}$ and $\Omega^{\mu \nu} = \Omega^{\mu \nu}_m \frac{\partial x^m}{\partial s^a} ds^a$ ($(s^1,s^2)$ being a local curvilinear coordinate system on $M_\Lambda$). 
\end{prop}
\begin{proof}
  \begin{eqnarray}
    \tr(\sigma^i \hat \mathfrak A)  & = & \tr(\sigma^iu_{J}^{-1} J^* \mathfrak Au_J) +\imath\tr(\sigma^i u_J^{-1}du_J)  \\
    & = & \tr(\hat \sigma^i J^* \mathfrak A) + \imath \tr(\hat \sigma^i du_J u_J^{-1}) \\
    & = & {J_{i'}}^i \tr(\sigma^{i'} J^* \mathfrak A) + \imath {J_{i'}}^i \tr(\sigma^i du_J u_J^{-1})
  \end{eqnarray}
  with $\hat \sigma^i = u_J \sigma^i u_J^{-1} = {J_{i'}}^i \sigma^{i'}$.\\
  \begin{equation}
    J^* \mathfrak A = \left. \mathfrak A_{m'} \right|_{J^{-1}x} ({J_m}^{m'}dx^m + \partial_m {J_p}^{m'}x^p dx^m)
  \end{equation}
  
  Since $J \in SO(3)$, we know that
  \begin{eqnarray}
    {J^{i'}}_{i}{J^{j'}}_j \delta_{i'j'} & = & \delta_{ij} \\
    {\varepsilon_{i'j'}}^{k'} {J^{i'}}_i {J^{j'}}_j &= & {\varepsilon_{ij}}^k {J^{k'}}_k
  \end{eqnarray}
  ($J$ preserves the Euclidean inner product and the vector cross product). We have then
  \begin{eqnarray}
    {\varepsilon^{jk}}_i {J_{i'}}^i \Re \tr(\sigma^{i'} \mathfrak A_{m'}) & = & {\varepsilon^{j'k'}}_{i'} {J_{j'}}^j {J_{k'}}^k \Re \tr(\sigma^{i'} \mathfrak A_{m'}) \\
    & = & -{J_{j'}}^j {J_{k'}}^k \Omega^{j'k'}_{m'}
  \end{eqnarray}
  Moreover
  \begin{eqnarray}
    {J^i}_j \sigma_i & = & u_J \sigma_j u_J^{-1} \\
    \Rightarrow d{J^i}_j \sigma_i & = & du_J \sigma_j u_J^{-1} - u_J \sigma_j u_J^{-1} du_J u_J^{-1} \\
    \Rightarrow d{J^i}_j \sigma_i & = & [du_J u_J^{-1},{J^i}_j \sigma_i]
  \end{eqnarray}
  Then
  \begin{eqnarray}
    {[J^{-1}]^j}_i d{J^k}_j \sigma_k & = & [du_J u_J^{-1},\sigma_i] \\
    & = & \frac{1}{2} \tr(\sigma^l du_J u_J^{-1}) [\sigma_l,\sigma_i] \\
    & = & \imath \tr(\sigma^l du_J u_J^{-1}) {\varepsilon_{li}}^k \sigma_k
  \end{eqnarray}
  So $\tr(\sigma^l du_J u_J^{-1}) = \frac{\imath}{2} {\varepsilon^{li}}_k {[J^{-1}]^j}_i d{J^k}_j$ and then
  \begin{eqnarray}
    {\varepsilon^{jk}}_i \imath \tr(\sigma^idu_J u_J^{-1}) & = & -\frac{1}{2} {\varepsilon^{jk}}_i {\varepsilon^{il}}_m{[J^{-1}]^p}_l d{J^m}_p \\
    & = & -\frac{1}{2}(\delta^{jl}\delta^k_m - \delta^j_m \delta^{kl}) {[J^{-1}]^p}_l d{J^m}_p \\
    & = & -\frac{1}{2} [J^{-1}]^{pj} d{J^k}_p + \frac{1}{2}[J^{-1}]^{pk}d{J^j}_p \\
    & = & \frac{1}{2} (J^{kp} d{J^j}_p - J^{jp} d{J^k}_p)
  \end{eqnarray}
  Finally we have
  \begin{eqnarray}
    \hat \Omega^{jk}_m & = & -{\varepsilon^{jk}}_i \Re \tr(\sigma^i \hat \mathfrak A_m) \\
    & = & {J_{j'}}^j {J_{k'}}^k \Omega^{j'k'}_{m'} {J_m}^{m'} - \frac{1}{2}(J^{kp} \partial_m {J^j}_p - J^{jp} \partial _m {J^k}_p)
  \end{eqnarray}
  and
  \begin{eqnarray}
    \hat \Omega^{i0}_m & = & \Im \tr(\sigma^i \mathfrak A_m) \\
    & = & {J_{i'}}^i \Omega^{i'0}_{m'} {J_m}^{m'}
  \end{eqnarray} 
\end{proof}

\subsection{Covariant derivatives on $M_\Lambda$}\label{coderiv}
\begin{defi}[Covariant derivatives]
  On $\underline{\mathbb C^2 \otimes \mathcal H}_{M_\Lambda}$ and on $\underline{\mathcal L(\mathbb C^2)}_{M_\Lambda}$ we define the covariant derivatives associated with the gauge potential $\mathfrak A$ by
  \begin{eqnarray}
    \nabla_a |\Phi \rrangle & = & \left(\frac{\partial}{\partial s^a} - \imath \mathfrak A_a \right) |\Phi \rrangle \\
    \nabla_a O & = & \frac{\partial O}{\partial s^a} - \imath [\mathfrak A_a,O]_\dagger
  \end{eqnarray}
  $\forall |\Phi \rrangle \in \underline{\mathbb C^2 \otimes \mathcal H}_{M_\Lambda}$ and $\forall O \in \underline{\mathcal L(\mathbb C^2)}_{M_\Lambda}$, with $[\mathfrak A_a,O]_\dagger \equiv \mathfrak A_a O - O \mathfrak A^\dagger_a$.
\end{defi}
These definitions are consistent with the weak adiabatic approximation:
\begin{eqnarray}
  & & |\Psi(t) \rrangle = \Ted^{-\imath \int_0^t \mathfrak A_a \dot s^a dt}|\Lambda(x(t))\rrangle \nonumber \\
  & & \Rightarrow \frac{d}{dt} |\Psi(t) \rrangle = \Ted^{-\imath \int_0^t \mathfrak A_a \dot s^a dt} \left. \nabla_a |\Lambda(x)\rrangle \right|_{x=x(t)} \dot s^a
\end{eqnarray}
and let $\rho(t) = \tr_{\mathcal H} |\Psi(t)\rrangle \llangle \Psi(t)|$ be the spin density matrix, we have
\begin{eqnarray}
  & & \rho(t)^\dagger = \left(\Ted^{-\imath \int_0^t \mathfrak A_a \dot s^a dt}\right)^\dagger \rho_\Lambda(x(t))\Ted^{-\imath \int_0^t \mathfrak A_a \dot s^a dt} \nonumber \\
  & & \Rightarrow \tr\left(\rho(t)^\dagger O(x(t))\right) \nonumber \\
  & & \quad =  \tr\left(\rho_\Lambda(x(t)) \Ted^{-\imath \int_0^t \mathfrak A_a \dot s^a dt} O(x(t)) \left(\Ted^{-\imath \int_0^t \mathfrak A_a \dot s^a dt}\right)^\dagger \right)
\end{eqnarray}
where $\rho_\Lambda = \tr_{\mathcal H}|\Lambda \rrangle \llangle \Lambda|$ is the eigen density matrix, and so
\begin{equation}
  \frac{d}{dt} \tilde O(t) = \Ted^{-\imath \int_0^t \mathfrak A_a \dot s^a dt} \left. \nabla_a O(x) \right|_{x=x(t)} \dot s^a \left(\Ted^{-\imath \int_0^t \mathfrak A_a \dot s^a dt}\right)^\dagger
\end{equation}
with $\tilde O(t) = \Ted^{-\imath \int_0^t \mathfrak A_a \dot s^a dt} O(x(t)) \left(\Ted^{-\imath \int_0^t \mathfrak A_a \dot s^a dt}\right)^\dagger$
\begin{prop}
  The covariant derivative onto $\underline{\mathcal L(\mathbb C^2)}_{M_\Lambda}$ induces the following covariant derivative onto the 4-vectors $(v^0,\vec v) \in \underline{\mathbb R^4}_{M_\Lambda}$:
  \begin{eqnarray}
    \nabla_a \vec v & = & \frac{\partial}{\partial s^a} \vec v  + \Re\tr(\vec \sigma \mathfrak A_a) \times \vec v + \Im\tr(\vec \sigma \mathfrak A_a)v^0 + \Im \tr(\mathfrak A_a) \vec v\\
    \nabla_a v^0 & = & \frac{\partial}{\partial s^a} v^0 + \Im\tr(\vec \sigma \mathfrak A_a) \cdot \vec v + \Im \tr(\mathfrak A_a) v^0
  \end{eqnarray}
\end{prop}
\begin{proof}
  Let $O = v^0 \id + v^i \sigma_i \in \mathcal L(\mathbb C^2)$ be the image of $(v^0,\vec v)$ by the isomorphism between $\mathbb R^4$ and $\mathcal L(\mathbb C^2)$ induced by the basis $(\id,\sigma_1,\sigma_2,\sigma_3)$.
  \begin{eqnarray}
    \nabla_a O & = & \frac{\partial O}{\partial s^a} -\imath (\mathfrak A_a(v^0\id + v^i\sigma_i)-(v^0\id + v^i\sigma_i)\mathfrak A_a^\dagger) \\
    & = & \frac{\partial O}{\partial s^a} -\frac{\imath}{2} v^i (\tr(\sigma^j \mathfrak A_a) \sigma_j \sigma_i - \overline{\tr(\sigma^j \mathfrak A_a)} \sigma_i \sigma_j) \nonumber \\
    & & \quad + v^0 \Im\tr(\sigma^j \mathfrak A_a) \sigma_j + \Im \tr(\mathfrak A_a) O \\
    & = & \frac{\partial O}{\partial s^a} -\frac{\imath}{2} v^i \Re\tr(\sigma^j \mathfrak A_a) [\sigma_j,\sigma_i] +\frac{1}{2} v^i \Im\tr(\sigma^j \mathfrak A_a) \{\sigma_j,\sigma_i\} \nonumber \\
    & & \quad + v^0 \Im\tr(\sigma^j \mathfrak A_a) \sigma_j + \Im \tr(\mathfrak A_a) O \\
    & = & \frac{\partial O}{\partial s^a} + {\varepsilon_{ji}}^k \Re\tr(\sigma^j \mathfrak A_a) v^i \sigma_k + \delta_{ij} v^i \Im\tr(\sigma^j \mathfrak A_a) \nonumber \\
    & & \quad + v^0 \Im\tr(\sigma^j \mathfrak A_a) \sigma_j + \Im \tr(\mathfrak A_a) O \\
    & = & \left(\frac{\partial v^i}{\partial s^a} + {\varepsilon^i}_{jk} \Re\tr(\sigma^j \mathfrak A_a) v^k + v^0 \Im\tr(\sigma^i \mathfrak A_a) + \Im \tr(\mathfrak A_a) v^i\right) \sigma_i \nonumber \\
    & & + \left(\frac{\partial v^0}{\partial s^a} + v^i \Im\tr(\sigma_i \mathfrak A_a) + \Im \tr(\mathfrak A_a) v^0 \right) \id
  \end{eqnarray}
\end{proof}
By restricting the covariant derivative to $\underline{\mathbb R^3}_{M_\Lambda}$:
\begin{equation}
  \nabla_a \vec v = \frac{\partial}{\partial s^a}\vec v + \Re\tr(\vec \sigma \mathfrak A_a) \times \vec v
\end{equation}
we have
\begin{equation}
  \Gamma^i_{bc} = (\nabla_b e_c)^i
\end{equation}

\begin{prop}
 Let $g(t) \equiv \Ted^{-\imath \int_0^t \mathfrak A_a \dot s^a dt}$. If $t \mapsto x(s(t)) \in M_\Lambda$ is an autoparallel geodesic: $\ddot s^a + \Gamma^a_{bc} \dot s^b \dot s^c = 0$, then
  \begin{equation}
    \frac{d}{dt} \left( g(t)\frac{d\slashed D_x}{dt} g(t)^\dagger \right) =  \frac{\imath}{2} g(t)\left\{ (\mathfrak A_a - \mathfrak A_a^\dagger)\dot s^a , \frac{d\slashed D_x}{dt} \right\} g(t)^\dagger
  \end{equation}
\end{prop}

\begin{proof}
\begin{eqnarray}
  \frac{d}{dt}\left(g \frac{d\slashed D_x}{dt} g^\dagger \right) & = & g \left( \dot s^a \nabla_a \frac{d \slashed D_x}{dt} \right) g^\dagger \\
  & = & g \left( \frac{\partial}{\partial s^a} \left(\frac{\partial \slashed D_x}{\partial s^b} \dot s^b \right) \dot s^a -\imath \left[\mathfrak A_a,\frac{\partial \slashed D_x}{\partial s^b}\right]_\dagger \dot s^a \dot s^b \right) g^\dagger \\
  & = & g \left(\frac{\partial \slashed D_x}{\partial s^b} \ddot s^b + \left(\frac{\partial^2 \slashed D_x}{\partial s^a \partial s^b} -\imath \left[\mathfrak A_a,\frac{\partial \slashed D_x}{\partial s^b}\right]_\dagger \right)\dot s^a \dot s^b \right) g^\dagger \\
  & = & g \left(\frac{\partial \slashed D_x}{\partial s^b} \ddot s^b + \left(\nabla_a \frac{\partial \slashed D_x}{\partial s^b} \right)\dot s^a \dot s^b \right) g^\dagger
\end{eqnarray}
but $\frac{\partial \slashed D_x}{\partial s^b} = - \sigma_i \frac{\partial x^i}{\partial s^b} = -\sigma_i e^i_b$, so
\begin{eqnarray}
  \frac{d}{dt}\left(g \frac{d\slashed D_x}{dt} g^\dagger \right) & = & - g \left(\sigma_i e^i_b \ddot s^b + \nabla_a (\sigma_i e^i_b) \dot s^a \dot s^b \right)g^\dagger \\
  & = & - g \left(e^i_b \ddot s^b + (\nabla_a e_b)^i \dot s^a \dot s^b + \Im \tr(\mathfrak A_a) e^i_b\right) \sigma_i g^\dagger \nonumber \\
  & & \quad - \Im \tr(\sigma_i \mathfrak A_a) e^i_b gg^\dagger \\
  & & - g e^i_c \left( \ddot s^c + \Gamma^c_{ab} \dot s^a \dot s^b \right) g^\dagger \nonumber \\
  & & \quad - \frac{\imath}{2} g \left\{\mathfrak A_a - \mathfrak A_a^\dagger, \sigma_i e^i_b \right\} g^\dagger \dot s_a \dot s_b
\end{eqnarray}
\end{proof}

$\mathfrak A$ is almost surely self-adjoint, in the case where it is self-adjoint (implying $g^\dagger = g^{-1}$) we have then
\begin{equation}
  \frac{d}{dt}\left(g \frac{d\slashed D_x}{dt} g^{-1} \right) = 0 \Rightarrow g(t) \frac{d\slashed D_x}{dt} g^{-1}(t) = \left. \frac{d\slashed D_x}{dt} \right|_{t=0}
\end{equation}
and then
\begin{eqnarray}
  & & \llangle \Lambda(x(t))|\frac{d\slashed D_x}{dt} |\Lambda(x(t))\rrangle = \llangle \Psi(t)|\left. \frac{d\slashed D_x}{dt} \right|_{t=0} |\Psi(t) \rrangle \\
  & \iff & \llangle \Psi(t)|\left. \frac{d\slashed D_x}{dt} \right|_{t=0} |\Psi(t) \rrangle  = -\llangle \Lambda(x(t))|\sigma_i|\Lambda(x(t))\rrangle \dot x^i = 0
\end{eqnarray}
by using the adiabatic approximation $|\Psi(t)\rrangle = g(t)|\Lambda(x(t))\rrangle$ ($\llangle \Lambda|\sigma^i|\Lambda\rrangle \partial_i$ is a normal vector of $M_\Lambda$ at $x$ whereas $\dot x^i \partial_i$ is a tangent vector of $M_\Lambda$ at $x$).

\section{The categorical fibre bundle}\label{category}
Quantum mechanics deal with two evolution law. A time continuous evolution laws (by the Schr\"odinger or the Dirac equations) associated with the smooth changes of the probabilities with respect to the time, and a discontinuous abrupt law (by the Born projection rule) associated with the state change after a measurement (and depending on the output of this one). This double aspect is also present with Fuzzy spaces. The Schr\"odinger-like equation is associated with the continuous evolution induced by the smooth moves of the probe $t \mapsto x(t)$. The length of the paths $\mathscr C$ in $M_\Lambda$ by $\gamma$ is a geometric measure of this evolution (by the energy uncertainty with probe misalignment). But, since the quasicoherent states are not separated $\llangle \Lambda(x)|\Lambda(y)\rrangle \not= \delta(y-x)$, a Fuzzy space in the state $|\Lambda(x)\rrangle$ can be projected in the state $|\Lambda(y)\rrangle$ after a measurement. This evolution can be modelled by the displacement operator $\Dis(y,x)$ viewed as a quantum path. The length of the quantum paths by $\pmb{\gamma}$ ($\dist_{\pmb \gamma}$) is a geometric measure of this evolution (by the uncertainty concerning the gap energy between the two points without considering the effects of the orientation state). This two evolutions, smooth paths and quantum jumps, can be naturally incorporate in a category structure. More precisely in a categorical manifold $\mathscr M_\Lambda$, for which $\Obj \mathscr M_\Lambda$ is a topological manifold (permitting to describe the smooth paths) and for which the arrows of $\Morph \mathscr M_\Lambda$ describe quantum jumps (evolutions by the Born projection rule) from their sources to their targets.

\subsection{Adiabatic transport along a pseudo-surface}\label{pseudosurface}
The introduction of the displacement operators is related to the category $\mathscr M_\Lambda$ built on the eigenmanifold. The gauge structure is then associated with a fibration over $\mathscr M_\Lambda$ where the adiabatic transports appear as horizontal lifts of ``paths'' on $\mathscr M_\Lambda$. But the good notion of ``path'' onto a categorical manifold is a pseudo-surface \cite{Viennot5}: $s \mapsto (\varphi,x(s)) \in \Morph \mathscr M_\Lambda$ where $s \mapsto x(s)$ defines a path $\mathscr C$ onto $\Obj \mathscr M_\Lambda = M_\Lambda$.\\

By considering a path $\mathscr C$ from $x$ to $y$ on $M_\Lambda$ and $\varphi \in \Diff M_\Lambda$, we have for the strong adiabatic transport: if $|\Psi(0) \rrangle = u_{\varphi(x),x} \otimes \Dis(\varphi(x),x)|\Lambda(x) \rrangle$ (with $u_{\varphi(x),x} \in U(2)_{a.s.}$):
\begin{equation}
  e^{-\imath \int_{\mathscr C} \tilde \eta_\varphi} u_{\varphi(y),y} \otimes \Dis(\varphi(y),y) e^{-\imath \int_{\mathscr C} A} |\Lambda(y) \rrangle = e^{-\imath \int_{\varphi(\mathscr C)} A} |\Lambda(\varphi(y)) \rrangle
\end{equation}
with $\tilde \eta_\varphi = \varphi^* A - A \in \Omega^1(M_\Lambda,\mathbb R)$. $\Dis(\varphi(y),y) = \Ped_{n.c.}^{-\imath \int_{\pmb \omega(\varphi,y)} \mathbf A_\Dis}$ ($\dnc \Dis = -\imath \Dis \mathbf A_\Dis$) is assimilated to a noncommutative geometric phase. The previous equality can be represented by the following diagram:
\begin{equation}
  \begin{CD}
    \overset{x}{\bullet} @>{\Ped_{n.c.}^{-\imath \int_{\pmb \omega(\varphi,x)} \mathbf A_\Dis}}>{(\varphi,x)}> \overset{\varphi(x)}{\bullet} \\
    @V{e^{-\imath \int_{\mathscr C} A}}V{\mathscr C}V  @V{\varphi(\mathscr C)}V{e^{-\imath \int_{\varphi(\mathscr C)} A}}V \\
      \underset{y}{\bullet} @>{(\varphi,y)}>{e^{-\imath \int_{\mathscr C} \tilde \eta_\varphi} \Ped_{n.c.}^{-\imath \int_{\pmb \omega(\varphi,y)} \mathbf A_\Dis}}>  \underset{\varphi(y)}{\bullet}
  \end{CD}
\end{equation}
$|\Psi(T) \rrangle = e^{-\imath \int_{\mathscr C} \eta_\varphi} u_{\varphi(y),y}\otimes \Dis(\varphi(y),y) e^{-\imath \int_{\mathscr C} A} |\Lambda(y) \rrangle$ is then the adiabatic transport of $|\Lambda \rrangle$ along the pseudo-surface $(\varphi,\mathscr C) \in \Diff M_\Lambda \times \mathcal PM_\Lambda$ on the category $\mathscr M_\Lambda$. 
\begin{eqnarray}
  \tilde \eta_\varphi & = & \varphi^* A - A \\
  & = & -\imath \llangle \Lambda(\varphi(x))|d|\Lambda(\varphi(x))\rrangle +\imath \llangle \Lambda(x)|d|\Lambda(x)\rrangle \\
  & = & -\imath \llangle \Lambda(x)|u_{\varphi(x),x}^{-1} \Dis(\varphi(x),x)^\dagger d u_{\varphi(x),x} \Dis(\varphi(x),x)|\Lambda(x)\rrangle \nonumber \\
  & & \qquad +\imath \llangle \Lambda(x)|d|\Lambda(x)\rrangle \\
  & = & \underbrace{-\imath \omega_x\left(\Dis(\varphi(x),x)^\dagger d\Dis(\varphi(x),x)\right)}_{\eta_\varphi} - \imath \omega_x(u^{-1}_{\varphi(x),x} du_{\varphi(x),x}) 
\end{eqnarray}
$\eta \in \underline{\Omega^1(M_\Lambda,\mathbb R)}_{\Diff M_\Lambda}$ is called gauge potential-transformation \cite{Viennot5}. $A$ is the gauge object-potential on $\mathscr M_\Lambda$ (defined on the elements of $\Obj (\mathscr M_\Lambda) = M_\Lambda$) whereas $\eta$ is the gauge arrow-potential on $\mathscr M_\Lambda$ (defined on the elements of $\Morph (\mathscr M_\Lambda) = \Diff M_\Lambda \times M_\Lambda$). Note that $\pmb \eta = -\imath \Dis^\dagger d\Dis$ and $\mathbf A_\Dis = \imath \Dis^\dagger \dnc \Dis$ have a similar form but the differentials are different (the one of $M_\Lambda$ for the first and the one of $\mathfrak M$ for the second). $\pmb \eta + \mathbf A_\Dis \in \Omega^1_{\Der}(\mathfrak X \otimes \mathcal C^\infty(M_\Lambda))$ can be viewed as a noncommutative gauge potential onto a noncommutative manifold defined by $\mathfrak X \otimes \mathcal C^\infty(M_\Lambda) = \mathcal C^\infty(M_\Lambda,\mathfrak X)$ ($\mathfrak X$-valued functions on $M_\Lambda$).

$B = d\eta \in \underline{\Omega^2(M_\Lambda,\mathbb R)}_{\Diff M_\Lambda}$ is called curving \cite{Viennot5}.

\begin{example}{2}{Fuzzy surface plots}
For a diffeomorphism $\varphi \in \Diff M_\Lambda$, the gauge potential-transformation is
  \begin{eqnarray}
    \eta_{\Dis(\varphi(\alpha),\alpha)} & = & -\imath \Dis(\varphi(\alpha),\alpha)^\dagger d\Dis(\varphi(\alpha),\alpha) \\
    & = & -\imath \left(\left(\frac{\partial \varphi}{\partial \alpha}-1\right) a^+-\frac{\bar \varphi}{\partial \alpha} a \right) d\alpha \nonumber \\
    & & -\imath \left(\frac{\partial \varphi}{\partial \bar \alpha} a^+ - \left(\frac{\partial \bar \varphi}{\partial \bar \alpha}-1\right)a \right)d\bar \alpha \\
    \omega_\alpha(\eta_{\Dis(\varphi(\alpha),\alpha)}) & = & -\imath \left(\left(\frac{\partial \varphi}{\partial \alpha}-1\right) \bar \alpha -\frac{\bar \varphi}{\partial \alpha} \alpha \right) d\alpha \nonumber \\
    & & -\imath \left(\frac{\partial \varphi}{\partial \bar \alpha} \bar \alpha  - \left(\frac{\partial \bar \varphi}{\partial \bar \alpha}-1\right)\alpha \right)d\bar \alpha
  \end{eqnarray}
  and so the adiabatic curving is
  \begin{equation}
    B = d\omega_\alpha(\eta_{\Dis(\varphi(\alpha),\alpha)}) = -\imath \left(2 - \frac{\partial \varphi}{\partial \alpha} -\frac{\partial \bar \varphi}{\partial \bar \alpha} \right) d\alpha \wedge d\bar \alpha
  \end{equation}
  For $\varphi$ corresponding to an external gauge change $\varphi(x^1,x^2) = (\cos \vartheta x^1 + \sin \vartheta x^2, -\sin \vartheta x^1 + \cos \vartheta x^2) \iff \varphi(\alpha) = e^{-\imath \vartheta} \alpha$, we have $B = 2\imath (\cos \vartheta-1) d\alpha \wedge d\bar \alpha = 4(\cos \vartheta-1) dx^1 \wedge dx^2$.\\
\end{example}

\subsection{The fibre bundle on $\mathscr M_\Lambda$}\label{fibre}
Let $\{\mathcal U^\alpha\}$ be a good open cover of $M_\Lambda$. The choice of the quasicoherent state $|\Lambda(x)\rrangle_{|\mathcal U^\alpha}$ constitutes a trivialising local section of a line bundle over $M_\Lambda$. This one is the associated line bundle of a $U(1)$-principal bundle over $M_\Lambda$, $P_\Lambda$, of local trivialisation denoted by $\phi^\alpha_P : \mathcal U^\alpha \times U(1) \xrightarrow{\simeq} {P_{\Lambda}}_{|\mathcal U^\alpha}$: $\forall x \in \mathcal U^\alpha, h\in U(1), z \in \mathbb C$
\begin{equation}
 [\phi^\alpha_P(x,hh'),h'z]_{h' \in U(1)} = hz|\Lambda(x)\rrangle_{|\mathcal U^\alpha}
\end{equation}
$P_\Lambda$ is endowed with a connection of local data the gauge potential $A$ and the curvature $F$. $P_\Lambda$ is the usual Berry-Simon adiabatic bundle \cite{Shapere, Bohm}.\\

We suppose that $\forall x,y \in M_\Lambda$, $U_x \simeq U_y$ ($U_x \subset SU(2)_{a.s.}$ such that $U_x \cdot |\Lambda(x)\rrangle = \ker \slashed D_x$). We denote by $U_0$ the model group (possibly one of the groups $U_x$ at a certain point) and by $\phi^\alpha_{Ux} : U_0 \xrightarrow{\simeq} U_x$ the group homomorphism supposed continuous with respect to $x$ on $\mathcal U^\alpha$. Let $\overline{U_0} = \bigcup_{u \in SU(2)} u U_0 u^{-1}$ be the normal closure of $U_0$ and $\tilde G = SU(2) \ltimes \overline{U_0}$ be the semi-direct product of groups defined by :
\begin{equation}
 \forall u,u' \in SU(2), \forall v,v' \in U_0 \quad (u',v')(u,v) = (u'u,u^{-1}v'uv)
\end{equation}
here $SU(2)$ stands for the spin rotations induced by external gauge changes in $SO(3)$. $\tilde G$ acts on $|\Lambda \rrangle$ as gauge changes via
\begin{equation}
  (u,v)\cdot|\Lambda(x)\rrangle_{|\mathcal U^\alpha} = u\phi^\alpha_{Ux}(v)|\Lambda(x)\rrangle_{|\mathcal U^\alpha}
\end{equation}
the semi-direct product being chosen to be compatible with its action $(u',v')(u,v)\cdot|\Lambda(x)\rrangle_{|\mathcal U^\alpha} = u'\phi^{\alpha}_{Ux}(v')u\phi^{\alpha}_{Ux}(v)|\Lambda(x)\rrangle_{|\mathcal U^\alpha}$. These elements define a $\tilde G$-principal bundle over $M_\Lambda$, $Q_\Lambda$, of local trivialisation $\phi_Q^\alpha : \mathcal U^\alpha \times \tilde G \xrightarrow{\simeq} {Q_\Lambda}_{|\mathcal U^\alpha}$ defined by: $\forall x \in \mathcal U^\alpha, (u,v)\in \tilde G$
\begin{equation}
  \phi^\alpha_Q(x,u,v) = (x,u,\phi^\alpha_{Ux}(v))
\end{equation}
If $U_x$ is trivial, $\tilde G = SU(2)$ and $Q_\Lambda$ is trivial. $\mathfrak A^{off} = \mathfrak A - \tr \mathfrak A$ is the gauge potential of the connection endowing $Q_\Lambda$ which has $\mathfrak F^{off} = \mathfrak F - \tr \mathfrak F$ as curvature. Since $\imath \Omega^{\mu \nu} \frac{1}{2}\{\sigma_\mu,\sigma_\nu\} =  \mathfrak A^{off}$ (with $\sigma_0 = \id$), $\mathfrak A^{off}$ and $\mathfrak F^{off}$ are just representations of the Lorentz connection and of the Riemann curvature. $Q$ is then the frame change principal bundle of $M_\Lambda$.\\

Let $G = U(1) \times \tilde G = U(2) \rtimes \overline{U_0}$, the different gauge changes are associated with the following central extension of groups:
$$ 1 \to U(1) \to G \to G/U(1)=\tilde G \to 1 $$
It is known that a $G$-fibration compatible with $P_\Lambda$ and $Q_\Lambda$ is not a principal bundle but a categorical bundle (twisted bundle \cite{Mackaay}), i.e. a bundle for which the total space is not a manifold but a category. Let $\mathscr G$ be the groupoid with $\Obj(\mathscr G) = \tilde G$ and $\Morph(\mathscr G) = (G \times G)/U(1) \equiv \{[g_2h,g_1h]_{h \in U(1)}\}$ with the following  source, target and identity maps
\begin{equation}
  s[g_2,g_1] = g_1U(1) \quad t[g_2,g_1]=g_2U(1) \quad \id_{gU(1)} = [g,g]
\end{equation}
where $gU(1) \in G/U(1) = \tilde G$ denotes a coset. The arrow composition (vertical composition of arrows) is defined by
\begin{equation}
  [g_3,g_2] \circ [g_2h,g_1] = [g_3h,g_1]
\end{equation}
and the groupoid low (horizontal composition of arrows) is defined by
\begin{equation}
  [g_2',g_1'][g_2,g_1] = [g_2'g_2,g_1'g_1]
\end{equation}
The categorical bundle $\mathscr P_\Lambda$ over $\mathscr M_\Lambda$ is defined by
\begin{eqnarray}
  \Obj({\mathscr P_{\Lambda}}_{|\mathcal U^\alpha}) & = & \bigsqcup_{x \in \mathcal U^\alpha} G_x \subset \underline{U(2)}_{\mathcal U^\alpha} \\
  \Morph({\mathscr P_{\Lambda}}_{|\mathcal U^\alpha}) & = & \bigsqcup_{x \in \mathcal U^\alpha} \bigsqcup_{\varphi \in \Morph \mathcal U^\alpha} \Hom(G_x,G_x) \times G_x \nonumber \\
  & & \qquad \qquad \subset \underline{\Aut U(2) \times U(2)}_{\mathcal U^\alpha}
\end{eqnarray}
where $U(2)$ stands for the unitary operators of $\mathbb C^2$, $G_x \equiv \phi_P^\alpha(U(1))\times SU(2) \rtimes \overline{U_x}$; the projection functor $\mathscr P_\Lambda \to \mathscr M_\Lambda$ being canonically defined by the external direct product. The source, target and identity maps, and the composition of arrows are defined by $\forall g \in G_x$, $\forall f \in \Hom(G_x,G_x)$, $\forall f' \in \Hom(G_{\varphi(x)},G_{\varphi(x)})$:
\begin{eqnarray}
  s(f,g) = g & \qquad & \id_g = (\id_{G_x},g) \\ t(f,g) = f(g) & \qquad & (f',f(g))\circ (f,g) = (f' \circ f,g)
\end{eqnarray}
Note that to define globally $\mathscr P_\Lambda$ on the whole of $\mathscr M_\Lambda$ we need to consider for each open set $\mathcal U^\alpha$ the associated category $\mathscr U^\alpha$ with $\Obj(\mathscr U^\alpha) = \mathcal U^\alpha$ and $\Morph \mathscr U^\alpha = \Diff \mathcal U^\alpha \times \mathcal U^\alpha$, within the functors $\cup$ and $\cap$ such that: $\Obj(\mathscr U^\alpha \cup \mathscr U^\beta) = \mathcal U^\alpha \cup \mathcal U^\beta$, $\Obj(\mathscr U^\alpha \cap \mathscr U^\beta) = \mathcal U^\alpha \cap \mathcal U^\beta$, $\Morph(\mathscr U^\alpha \cup \mathscr U^\beta) = \Diff(\mathcal U^\alpha \cup \mathcal U^\beta) \times \mathcal U^\alpha \cup \mathcal U^\beta \supset \Diff \mathcal U^\alpha \cup \Diff \mathcal U^\beta \times \mathcal U^\alpha \cup \mathcal U^\beta$, and  $\Morph(\mathscr U^\alpha \cap \mathscr U^\beta) = \Diff(\mathcal U^\alpha \cap \mathcal U^\beta) \times \mathcal U^\alpha \cap \mathcal U^\beta \subset \Diff \mathcal U^\alpha \cap \Diff \mathcal U^\beta \times \mathcal U^\alpha \cap \mathcal U^\beta$ (where $\varphi \in \Diff \mathcal U^\alpha$ is naturally extended in $\Diff M_\Lambda$ by the identity on $M_\Lambda \setminus \mathcal U^\alpha$, $\Diff \mathcal U^\alpha$ being defined by the diffeomorphism of the open set $\mathcal U^\alpha$ asymptotically leaving invariant the border of $\mathcal U^\alpha$).\\
The local trivialisation functor of $\mathscr P_\Lambda$ is $\phi^\alpha_{\mathscr P} : \mathscr U^\alpha \times \mathscr G \xrightarrow{\simeq} {\mathscr P_{\Lambda}}_{|\mathscr U^\alpha}$ with: $\forall x \in \mathcal U^\alpha$, $\forall g\in G$, $\forall \varphi \in \Morph \mathcal U^\alpha$, $\forall u_t,u_s \in U(2)$, $\forall v_t,v_s \in \overline{U_0}$, $\forall h \in U(1)$
\begin{eqnarray}
  & & \phi^\alpha_{\mathscr P}(x,gU(1))  =  \phi^\alpha_Q(x,g) \\
  & & \phi^\alpha_{\mathscr P}\left((\varphi,x),(u_tv_th,u_s v_s)\right)  = \nonumber \\
  & & \left(L\left(\phi^\alpha_P(\varphi(x),h)\phi^\alpha_Q(\varphi(x),u_t,v_t)\right) \circ \phi^\alpha_{U\varphi(x)} \circ {\phi^{\alpha}_{Ux}}^{-1} , \phi^\alpha_Q(u_s,g_s) \right)
  \end{eqnarray}
$L$ being the left action of $U(2)$ (as set of unitary operators of $\mathbb C^2$) on itself. This definition is consistent with the source and target maps:
\begin{eqnarray}
  & & s(\phi^\alpha_{\mathscr P}\left((\varphi,x),(u_tv_th,u_s v_s)\right)) = \phi^\alpha_Q(x,u_sv_s) \\
  & & t(\phi^\alpha_{\mathscr P}\left((\varphi,x),(u_tv_th,u_s v_s)\right)) = \nonumber \\
  & & \qquad \phi^\alpha_P(\varphi(x),h) u_tu_s \phi^\alpha_{U\varphi(x)}(u_s^{-1}v_tu_sv_s)
\end{eqnarray}
$\mathscr G$ acts on $|\Lambda \rrangle$ as gauge changes via
\begin{eqnarray}
  & & (u_sv_s) \cdot |\Lambda(x)\rrangle_{|\mathcal U^\alpha}  =  u_s \phi^\alpha_{Ux}(v_s)|\Lambda(x)\rrangle_{|\mathcal U^\alpha} \\
  & & \left[u_tv_th,u_sv_s\right] \cdot_\varphi |\Lambda(x) \rrangle_{|\mathcal U^\alpha}  =  \nonumber \\
  & & \phi^\alpha_P(\varphi(x),h)u_t \phi^\alpha_{U\varphi(x)}(v_t) u_s \phi^\alpha_{U\varphi(x)}(v_s) \otimes \Dis(\varphi(x),x) |\Lambda(x) \rrangle_{|\mathcal U^\alpha}
\end{eqnarray}
$\mathscr P_\Lambda$ is endowed with a 2-connection described by the gauge potentials $(\mathfrak A,\eta, H)$, the fake curvature $\mathfrak F$ and the curving $B$ \cite{Viennot1, Viennot5}. Note that we have two gauge potential-transformations, $\eta$ which is related to the arrows of the base category $\mathscr M_\Lambda$ and $H$ which is related to the arrows of the structure groupoid $\mathscr G$ and defined by $H \rho_\Lambda = -\imath \langle \Lambda(x)|\Dis(\varphi(x),x)^{-1}d\Dis(\varphi(x),x)|\Lambda(x)\rangle_*$. $H$ is the non-abelian gauge potential-transformation associated with $\mathfrak A$ (it plays the same role than $\eta$ for $A$). $\tr(\rho_\Lambda H) = \eta$, $H$ is then a gauge potential-transformation random variable which has $\eta$ as mean value in the statistical distribution of states defined by $\rho_\Lambda$.\\

Let $U(\mathfrak X) \subset \Env(\mathfrak X)$ be the set of unitary operators in $\Env(\mathfrak X)$. Let $\mathscr E$ be the category such that $\Obj(\mathscr E) = \mathbb C^2 \otimes \mathcal H$, $\Morph(\mathscr E) = U(\mathfrak X) \times \mathbb C^2\otimes \mathcal H$, with the following source, target and identity maps, and composition of arrows
\begin{eqnarray}
  s(U,|\Psi\rrangle) = |\Psi \rrangle & \qquad & \id_{\Psi} = (\id,|\Psi\rrangle) \\
  t(U,|\Psi\rrangle) = U|\Psi \rrangle & \qquad & (U',U|\Psi\rrangle) \circ (U,|\Psi \rrangle) = (U'U,|\Psi\rrangle)
\end{eqnarray}
There is a canonical projection functor $\mathscr E \to \mathscr P\mathfrak X$, defined by the $|\Psi\rrangle \mapsto \tr(|\Psi\rrangle \llangle \Psi| \bullet)$ and the projection induced by the quotient $U(\mathfrak X)/\mathcal Z(\Env(\mathfrak X)) \simeq \InnAut\Env(\mathfrak X)$. $\mathscr E \to \mathscr P\mathfrak X$ can be viewed as a noncommutative bundle of flat connection defined by the gauge potential $\mathbf A_\Dis \in \Omega^1_\Der(\mathfrak X)$. The functor $\pmb \omega : \mathscr M_\Lambda \to  \mathscr P\mathfrak X$ permits to consider the associated bundle $\mathscr P_\Lambda \times_{\mathscr G} \mathscr E_{|\pmb \omega(\mathscr M_\Lambda)}$ over $\mathscr M_\Lambda$, defined by
\begin{eqnarray}
  & & \Obj(\mathscr P_\Lambda \times_{\mathscr G} \mathscr E_{|\pmb \omega(\mathscr M_\Lambda)}) =  \{[p\tilde g^{-1},\tilde g |\Lambda(x) \rrangle]_{\tilde g \in \tilde G}, p \in \Obj(\mathscr P_\Lambda)_x\} \\
  & & \Morph(\mathscr P_\Lambda \times_{\mathscr G} \mathscr E_{|\pmb \omega(\mathscr M_\Lambda)})  = \nonumber \\
  & & \left\{[(f,p)[g_th,g_s]^{-1}, \right. \nonumber \\
    & & (g_tg_s \Dis(\varphi(x),x) g_s^{-1}g_t^{-1},hg_tg_s|\Lambda(x)\rrangle)]_{[g_th,g_s]\in G\times G/U(1)} \nonumber \\
    & & \left. , (f,p) \in \Morph(\mathscr P_\Lambda)_{(\varphi,x)}\right\}
\end{eqnarray}

The strong adiabatic transport formulae: $e^{-\imath \int_{\mathscr C} A} |\Lambda(x) \rrangle$ and $e^{-\imath \int_{\mathscr C} (A+\eta_\varphi)} \Dis(\varphi(x),x)|\Lambda(x)\rrangle$, corresponds to the action of the source and the target on the arrow $\left[e^{-\imath \int_{\mathscr C} \eta_\varphi},e^{-\imath \int_{\mathscr C} A}\right]$, and the weak adiabatic transport formulae: $\Ped^{-\imath \int_{\mathscr C} \mathfrak A} |\Lambda(x) \rrangle$ and $\Ped^{-\imath \int_{\mathscr C} (\mathfrak A+H_\varphi)} \otimes \Dis(\varphi(x),x)|\Lambda(x)\rrangle$  corresponds to the action of the source and the target on the arrow $\left[\Ped^{-\imath \int_{\mathscr C} (\mathfrak A+H_\varphi)}(\Ped^{-\imath \int_{\mathscr C} \mathfrak A})^{-1},\Ped^{-\imath \int_{\mathscr C} \mathfrak A}\right]$. The adiabatic transport of $|\Lambda \rrangle$ along the pseudo-surface $(\varphi,\mathscr C) \in \Diff M_\Lambda \times \mathcal PM_\Lambda$ can be then viewed as the the horizontal lift of the pseudo-surface $(\varphi,\mathscr C)$ in $\mathscr P_\Lambda$.

\subsection{Noncommutative versus categorical geometries}\label{NCvsCat}
\begin{tabular}{c||c|c}
  \hline
  \textit{\textbf{Noncommutative geometry}} & \multicolumn{2}{c}{\textit{\textbf{Categorical geometry}}} \\
  & \textit{\textbf{Objects}} & \textit{\textbf{Arrows}} \\
  \hline
  $\InnAut\Env(\mathfrak X) \times \mathcal P(\mathfrak M)$ & $M_\Lambda$ & $\Diff M_\Lambda \times M_\Lambda$ \\
  {\footnotesize automorphisms and pure states} & {\footnotesize eigenmanifold} & {\footnotesize eigenmanifold diffeomorphisms} \\
  \hline
  $\Env(\mathfrak X)$ & $\mathcal C^\infty(M_\Lambda)$ & $\Diff M_\Lambda \times \mathcal C^\infty(M_\Lambda)$ \\
  {\footnotesize quantum observables} & {\footnotesize classical observables} & {\footnotesize transformations of observables} \\
  \hline
  $\Der(\mathfrak X)$ & $TM_\Lambda$ & $\underline{TM_\Lambda}_{\Diff M_\Lambda}$ \\
  {\footnotesize derivatives (commutators)} & {\footnotesize tangent vectors} & {\footnotesize diff. dependent tangent vectors} \\
  \hline
  $\pmb{d\ell^2}_x$ / $\pmb \gamma = \delta_{ij} \dnc X^i \otimes \dnc X^j$ & $\gamma_{ab}ds^ads^b$ & $\dist_{\pmb \gamma}(\omega_x,\omega_{\varphi(x)})$ \\
   {\footnotesize square length observable /} & {\footnotesize square length} & {\footnotesize average length}\\
        {\footnotesize noncommutative inner product} & {\footnotesize of the mean path} & {\footnotesize of the quantum paths}\\
  \hline
  $\mathbf A_\Dis \in \Omega^1_{\Der}(\mathfrak X)$ & $A \in \Omega^1(M_\Lambda,\mathbb R)$ /  & $\eta \in \underline{\Omega^1(M_\Lambda,\mathbb R)}_{\Diff M_\Lambda}$ / \\
  & $\mathfrak A \in \Omega^1(M_\Lambda,\mathfrak u(2)_{a.s.})$ & $H \in \underline{\Omega^1(M_\Lambda,\mathfrak u(2)_{a.s.})}_{\Diff M_\Lambda}$ \\
  {\footnotesize linking vector observable} & {\footnotesize shift vector /} & {\footnotesize gauge potential transformations} \\
  & {\footnotesize Lorentz connection} & \\    
  \hline
  $0$ & $F \in \Omega^2(M_\Lambda,\mathbb R)$ / & $B \in \underline{\Omega^2(M_\Lambda,\mathbb R)}_{\Diff M_\Lambda}$ / \\
  & $\mathfrak F \in \Omega^2(M_\Lambda,\mathfrak u(2))_{a.s.}$ & $dH-\imath H \wedge H$\\
  {\footnotesize noncommutative flatness} & {\footnotesize contorsion / Riemann curvature} & \footnotesize{curvings} \\
  \hline
\end{tabular}\\

The category having $\Diff M_\Lambda \times \mathcal C^\infty(M_\Lambda)$ as arrow set is endowed with the following source, target and identity maps: $s(\varphi,f) = f$, $t(\varphi,f)=f\circ \varphi$, $\id_f = (\id_{M_\Lambda},f)$; and with the following arrow composition : $(\varphi',f\circ \varphi) \circ (\varphi,f) = (\varphi \circ \varphi',f)$.\\
If we consider a paralinkable fuzzy space (\ref{paralinkable}) we have a complete symmetry for the gauge fields between the noncommutative geometry and the categorical geometry, with $\mathbf A_\Dis \in \Omega^1_\Der(\Env(\mathfrak X))$, $\pmb{\mathfrak A}_\Dis \in \Omega^1(\Env(\mathfrak X),\mathfrak{gl}(2,\mathbb C))$, $\mathbf F_\Dis \in \Omega^2_\Der(\Env(\mathfrak X))$ and $\pmb{\mathfrak F}_\Dis \in \Omega^2(\Env(\mathfrak X),\mathfrak{gl}(2,\mathbb C))$.\\

The two geometries are related by applications associated with $\omega = \tr(P_\Lambda \bullet) \in \underline{\mathcal P(\mathfrak M)}_{M_\Lambda}$:
$$ \begin{CD}
  \InnAut\Env(\mathfrak X) \times \mathcal P(\mathfrak M) @<{\pmb \omega}<< \Diff M_\Lambda \times M_\Lambda \\
  \mathcal P(\mathfrak M) @<{\omega}<< M_\Lambda \\
  \Der(\mathfrak X) @>{\pi_x \omega_{x*}}>> T_xM_\Lambda \\
  \Omega^1_\Der(\mathfrak X) @<{\omega_x^* \pi_x^*}<< \Omega^1_xM_\Lambda
\end{CD} $$

In a same way, the different metrics are related by $\omega$:\\
$$ \begin{tikzcd}
  \pmb{d\ell^2}_{x(s)} \arrow[rd, "\omega_{x(s+ds)}"] & \phantom{A} \\
  \pmb \gamma \arrow[rd, "\omega_x(\sqrt{\bullet(L_\varphi,L_\varphi)})"'] & \arrow[l, "\omega_x^* \pi_x^*"']  \gamma_{ab}ds^a ds^b \\
  & \dist_{\pmb \gamma} (\omega_{\varphi(x)},\omega_x)
\end{tikzcd} $$
where $e^{L_\varphi} = \Ad_{\Dis(\varphi(x),x)}$ (this relation holds only if $\varphi(x)$ and $x$ are linkable).\\

It is interesting to point out the relation between these two nonlocal geometries. Noncommutative geometry is nonlocal because of the quantum uncertainties $\Delta_x X^i$ and of the non-separability of quasicoherent states $\llangle \Lambda(y)|\Lambda(x)\rrangle \not= \delta(y-x)$. The categorical geometry is nonlocal because of the arrows which model the possibility to jump ``abruptly'' from the source to the target (the arrows being associated with the displacement operator). For the metric properties, the two notions of nonlocality are intrinsically related : $\pmb{d\ell^2} \to \gamma$ (square length of the mean path $\leftrightarrow$ distance between objects) $\pmb \gamma \to \dist_{\pmb \gamma}$ (average length of the paths $\leftrightarrow$ length of arrows). The reason of this relation is obvious, $\Obj \mathscr M_\Lambda$ is the manifold of the mean values (in the quasicoherent states) of $\vec X$, whereas $\pmb \omega(\Morph \mathscr M_\Lambda) = \{(\Ad_{\Dis(\varphi(x),x)},\omega_x)\}_{x \in M_\Lambda, \varphi \in \Diff M_\Lambda}$ is associated with the displacement operator $\Dis(y,x)$ which models the transformation of the quasicoherent state for a measurement with the probe at $x$ followed by a measurement with probe at $y$, without any intermediate measurement. $\Dis(y,x)$ can be then assimilated to the quantum path between the two measurements (submitted to the quantum fluctuations).

\end{document}